%% file: main.tex
\documentclass[twocolumn,10pt]{article}
\usepackage[margin=0.85in]{geometry}
\newif\ifanon\anonfalse
\newif\ifarxiv\arxivtrue
\newif\ifconf\conffalse
\usepackage[utf8]{inputenc}
\usepackage{amsmath,amssymb,amsthm,mathtools}
\usepackage{booktabs,graphicx,xcolor}
\usepackage[shortlabels]{enumitem}
\usepackage[round]{natbib}
\usepackage{stfloats}
\usepackage[colorlinks=true,linkcolor=blue!55!black,citecolor=blue!55!black,urlcolor=blue!55!black]{hyperref}

\newcommand{\todo}[1]{}

\newtheorem{theorem}{Theorem}
\newtheorem{proposition}{Proposition}
\newtheorem{corollary}{Corollary}
\newtheorem{lemma}{Lemma}
\theoremstyle{definition}
\newtheorem{assumption}{Assumption}
\theoremstyle{remark}
\newtheorem{remark}{Remark}

\newcommand{\D}[1]{\mathcal D_{#1}}
\newcommand{\Real}{\mathbb R}
\newcommand{\E}{\mathbb E}
\newcommand{\Prob}{\mathbb P}
\newcommand{\Neff}{N_{\mathrm{eff}}}
\newcommand{\hdelta}{\widehat\delta}
\newcommand{\homega}{\widehat\omega}
\newcommand{\hc}{\widehat c_\alpha}
\newcommand{\Rhat}{\widehat{R}}
\newcommand{\dB}{d_B}
\newcommand{\Winf}{W_{\infty,0}}
\DeclareMathOperator{\tr}{tr}

\title{Where Do Two Populations of Persistence Diagrams Differ?\\
  Calibrated Local Inference at a Fixed Budget}
\author{\begin{tabular}{ccc}
  Pramita Bagchi$^{1}$ & Edward Bae$^{1}$ & Atish Mitra$^{2}$ \\[4pt]
  Alexander D. Silberman$^{1}$ & \v{Z}iga Virk$^{3}$ & Sushovan Majhi$^{1}$
  \end{tabular}\\[8pt]
  {\normalsize $^{1}$The George Washington University \quad
   $^{2}$Montana Technological University \quad
   $^{3}$University of Ljubljana}}
\date{}

\begin{document}
\maketitle

\begin{abstract}
Many two-sample tests for populations of persistence diagrams assess global differences without identifying the regions of the birth--death plane that contribute to them. We study simultaneous
inference for local mean contrasts when the number of available diagrams is fixed. They are differences in expected weighted feature mass within $\ell_\infty$ neighborhoods at several centers and radii. We estimate these contrasts using additive
landmark responses. A Gaussian multiplier bootstrap calibrates simultaneous confidence intervals while allowing unequal group covariances. The neighborhoods whose intervals
exclude zero form a map with approximate family-wise error control, and selecting a subset of original intervals for display preserves their joint coverage guarantee. On the simultaneous coverage event, every reported neighborhood lies within twice its radius of the support of the mean-measure difference. A geometric result gives sufficient radius conditions for a displaced
feature to produce a nonzero contrast. A comparison of sufficient
detection thresholds quantifies the tradeoff between reducing the
number of tested coordinates and reserving observations for an
independent pilot. In simulations with $40$ to $120$ diagrams per class, the bands achieved
$94\%$--$98\%$ simultaneous coverage under both the strict null and
equal means with unequal covariances. In the latter setting, a
permutation maximum and the pooled-$t$ implementation of the
two-stage persistence-image test of Moon and Lazar rejected in up
to $32\%$ and $26\%$ of runs, respectively. In the fixed-budget
simulations, spending a third of the observations on a pilot to
choose landmarks or radii located changes less often than a
prespecified grid at a single radius. On the MUTAG benchmark, the localized region concentrates on
rings of fused-ring systems, an exploratory reading.
\end{abstract}

\section{Introduction}
\label{sec:intro}

A persistence diagram records when topological features appear and disappear along a filtration. In population studies, each observation is represented by a persistence diagram; observations might be images, networks, or spatial point patterns. Groups may differ in how often particular features occur, their birth and death values, or their variability across observations. Populations of diagrams can be compared using permutation tests on diagram distances \citep{RobinsonTurner2017} or tests based on kernel and mean embeddings \citep{Reininghaus2015,Kusano2016,Gretton2012,PaperIII}.
Other approaches use persistence landscapes \citep{Bubenik15}, pooled lifetimes \citep{STRAND}, or persistence intensities \citep{HanKimKim2026}.

For local inference, \citet{MoonLazar2023} test filtered persistence-image pixels with false-discovery control, while \citet{Chen2017PIF} provide pointwise bands for persistence intensities.
Confidence sets for a single diagram address a different inferential target \citep{Fasy2017,Chazal2014,Chazal2014b}. Our focus is simultaneous confidence intervals for local mean contrasts across birth--death locations and neighborhood radii.

Local inference has a resolution problem. A small neighborhood separates
nearby changes, but its response is zero in most diagrams; a large one
collects more mass but can merge a feature lost in one place with a
feature gained nearby. Choosing a radius changes both the estimand and
its precision, and a method that leaves the choice unspecified is hard to
use when few diagrams are available.

\paragraph{This paper.}
We estimate local mean contrasts using additive landmark responses from PLACE and PALACE \citep{Mitra2024,PALACE,PLACE}; Figure~\ref{fig:map} shows the resulting map on a molecular benchmark. Each response is a weighted count of diagram points in an
$\ell_\infty$ neighborhood at a specified center and radius. Its mean contrast is an integral against the difference of the two mean counting measures \citep{DivolChazal2019,DivolPolonik2019,DivolLacombe2021}. A Gaussian multiplier bootstrap calibrates all coordinates jointly \citep{Chernozhukov2013-jl,Chernozhukov2017}.
The band-validity theorem (Theorem~\ref{thm:band}) follows by applying their approximation results to these responses.
\begin{enumerate}[label=(\roman*),nosep]
\item \emph{Simultaneous local inference}
   (\S\ref{sec:method}, Corollary~\ref{cor:fwer}): joint intervals cover
  contrasts over centers and radii. Reporting a data-dependent subset
  of the original intervals preserves their simultaneous coverage
  guarantee. On the coverage event, a spatial bound relates reported
  neighborhoods to the support of the mean-measure difference
  (Proposition~\ref{prop:extent}). Under a localized displacement model,
  we also give radius conditions ensuring a nonzero mean contrast
  (Theorem~\ref{thm:local}).
\item \emph{Neighborhood choice at a fixed budget} (Corollary~\ref{cor:pilot},
  \S\ref{sec:exp-localize-budget}): for pilot selection of a subfamily from a prespecified family,
  we compare sufficient detection thresholds, accounting for both
  the loss of inference observations and the reduction in the
  number of tested coordinates. In our main simulation designs,
  a prespecified grid at a single radius, using every diagram,
  gave the highest localization rates.
\item \emph{Lifetimes are not enough}: any statistic of the lifetime
  multisets has power equal to its level against shifts along the
  diagonal (Proposition~\ref{prop:invariance}).
\end{enumerate}

\begin{figure}[t]
  \centering
  \IfFileExists{figures/map_mutag_rev.pdf}{\includegraphics[width=\columnwidth]{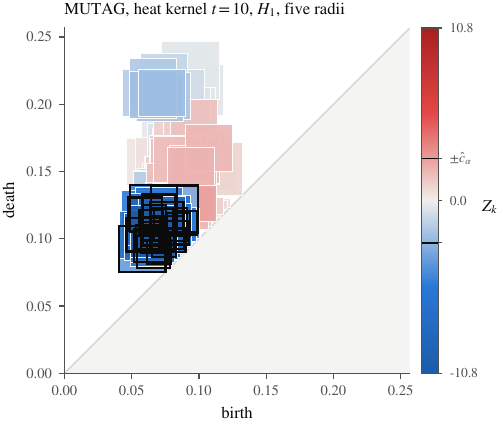}}{\todo{\texttt{\detokenize{figures/map_mutag_rev.pdf}}}}
  \caption{Local inference on MUTAG (mutagenic against nonmutagenic
  compounds; heat-kernel signature, $t=10$, degree one). Class-aware
  landmarks are fitted on a pilot third of $175$ deduplicated molecules,
  each with five radii; squares are filled by the standardized contrast $Z_k$ of \S\ref{sec:method}, and at each center
  the smallest radius whose interval excludes zero is outlined. All $58$
  localized coordinates, at $26$ centers, carry more mass in the
  mutagenic class. Over $20$ random splits the localized set is never
  empty and $97\%$ of a split's localized neighborhoods are matched by
  another split's (\S\ref{sec:exp-mol}).}
  \label{fig:map}
\end{figure}

\section{Local contrasts}
\label{sec:setup}

\paragraph{Diagrams and landmark responses.}
Let $\D{n}$ contain persistence diagrams with at most $n$
points $(b,d)$ satisfying $0\le b<d\le L$. The counting measure 
$N_D$ assigns to each measurable region $A$ the number of points 
of $D$ in $A$, including multiplicities. For $X\sim P$, the
mean measure $\bar\mu_P(A)=\E_P[N_X(A)]$ gives the expected
count in that region \citep{DivolChazal2019}. Its definition
requires no matching of features across diagrams. 
Write $\pi(b,d)=(d-b)/2$. For the one-point diagrams
$\{p\}$ and $\{a\}$, abbreviate their bottleneck distance as
$\dB(p,a)=\min\{\|p-a\|_\infty,\max(\pi(p),\pi(a))\}$.
A landmark family $\nu=\{(p_k,r_k,w_k)\}_{k=1}^{\ell}$
specifies centers in the diagram domain, radii
$0<r_k\le\pi(p_k)$, and weights $w_k>0$.
Centers may appear at several radii.
For an off-diagonal point $a$ and a diagram $D\in\D{n}$,
define the point response and diagram response, respectively, by
\begin{equation}\label{eq:embedding}
  \varphi_k(a)=w_k(r_k-\dB(p_k,a))_+,~
  \Phi_k(D)=\int\varphi_k\,dN_D.
\end{equation}
Together, the diagram responses form a linear
representation \citep{DivolPolonik2019}. 
The point response $\varphi_k$ is $w_k$-Lipschitz in
$\dB$, and $0\le\Phi_k(D)\le nw_kr_k$. 
Within the diagram domain, the positive-response region
$\mathcal B^\circ_k=\{a:\dB(p_k,a)<r_k\}$ is the
intersection of the frame with the open $\ell_\infty$
square of half-side $r_k$ centered at $p_k$.
The radius cap keeps this region off the diagonal.
For diagram laws $P$ and $Q$ on $\D{n}$, define
\begin{equation}\label{eq:estimand}
  \delta_k=\E_P\Phi_k(X)-\E_Q\Phi_k(Y)
  =\int\varphi_k\,d(\bar\mu_P-\bar\mu_Q).
\end{equation}
This contrast measures a difference in weighted local mass.
Moving points within $\mathcal B^\circ_k$ can change $\delta_k$
without changing the expected count. Opposing local changes
can cancel, giving $\delta_k=0$. 
Appendix~\ref{app:transport} gives the padded transport
construction and a lower confidence bound on
mean-measure separation.

\paragraph{Sampling and admissible learning.}
The inference observations are independent samples $X_1,\dots,X_m\sim P$
and $Y_1,\dots,Y_{m'}\sim Q$. Write $N=m+m'$ and
$\Neff=mm'/N$. The response vectors
$U_i=(\Phi_k(X_i))_k$ and $V_j=(\Phi_k(Y_j))_k$
have covariances $\Sigma_P$ and $\Sigma_Q$.
The scaled difference $\sqrt{\Neff}(\bar U-\bar V)$ has
covariance matrix $\Gamma$ and coordinate variances
$\omega_k^2$, given by
\begin{equation}\label{eq:gamma}
  \Gamma=\tfrac{m'}{N}\Sigma_P+\tfrac{m}{N}\Sigma_Q,
  \qquad \omega_k^2=\Gamma_{kk}.
\end{equation}
\begin{assumption}[Fixed family]\label{ass:fixed}
The centers, radii, weights, and retained coordinates are prespecified, 
or are functions of a labeled pilot sample
independent of the inference observations. 
When a pilot is used, probability statements are
conditional on it.
\end{assumption}
A pilot may fit class-aware farthest-point centers \citep{PALACE}, cap
radii, screen coordinates, and estimate scales. 
Normalization factors $\homega_k$ may be
prespecified, pilot-based, or estimated from the
inference observations. Inference-based estimates
require uniform relative consistency in Appendix~\ref{app:studentization},
which also gives groupwise and pooled sample-variance
constructions. With a prespecified family, all observations can be used
for inference with either deterministic normalizing
factors, such as $\homega_k=nw_kr_k$, or inference-based
estimates satisfying this condition. Under the strict null $P=Q$, every $\delta_k=0$. The coordinate mean null $\delta=0$ permits unequal group
covariances and may hold even when $\bar\mu_P\ne\bar\mu_Q$,
because the chosen responses need not distinguish all
differences between mean measures.

\section{Inference and neighborhood choice}
\label{sec:method}

For $0<\alpha<1$, we construct a simultaneous band
with nominal coverage $1-\alpha$ for a fixed landmark
family. Neighborhoods can be examined jointly across
radii or chosen on an independent pilot before inference.

\paragraph{A simultaneous band.}
Let $\hdelta=\bar U-\bar V$ and $Z_k=\sqrt{\Neff}\,\hdelta_k/\homega_k$.
Compute the normalization factors $\homega_k$ once and hold them
fixed across multiplier draws. 
For each bootstrap replicate, draw mutually independent
standard normal multipliers $e_i,e'_j$, independently
of the observations. Use the same multipliers across
all centers and radii to compute
{\small\begin{align}
  Z^*_k&=\frac{\sqrt{\Neff}}{\homega_k}\Bigl(\frac1m\sum_i e_i(U_{ik}-\bar U_k)
-\frac1{m'}\sum_j e'_j(V_{jk}-\bar V_k)\Bigr).\nonumber
\end{align}}
Repeat this step independently $B$ times. Let $\hc$ be the $\lceil B(1-\alpha)\rceil$th smallest
of the resulting values of $\max_k|Z^*_k|$. Define the band and localization set by
\begin{equation}\label{eq:band}
  \mathcal C_k=\bigl[\hdelta_k\pm\hc\,\homega_k/\sqrt{\Neff}\bigr],\qquad
  \Rhat=\{k:|Z_k|>\hc\}.
\end{equation}
The set $\Rhat$ indexes intervals that exclude zero.
For these coordinates, the sign of $\hdelta_k$ gives
the reported direction of the local contrast.
Prespecified or pilot-based normalizing factors need
not estimate standard deviations. Inference-based factors, 
including groupwise and pooled SD estimates, must satisfy 
the uniform relative consistency conditions in 
Appendix~\ref{app:studentization}. The multiplier distribution 
calibrates the normalized contrasts jointly.
For optional step-down testing, $\Rhat_{\mathrm{sd}}$
denotes the Romano--Wolf rejection set computed from
the same bootstrap draws \citep{RomanoWolf2005}. 
Under $P=Q$, label permutation also gives a
finite-sample valid test, with label-dependent
normalization recomputed for each permutation.
It is not generally exact under the mean null.

\paragraph{Several radii and selected reporting.}
Choose centers and base radii $r_0(p)>0$ under
Assumption~\ref{ass:fixed}. For a finite prespecified
grid $\mathcal H \subset (0,\infty)$ of multipliers, use
$r=\min\{hr_0(p),\pi(p)\}$, $h\in\mathcal H$.
Keep identical coordinates created by the radius cap
only once. Shared multipliers account for dependence
among responses. Additional radii may increase the
joint critical value, while sparse responses can make
variance estimation and calibration less reliable.
After computing the joint band, any subset of its
intervals may be displayed. For example, report the
smallest tested radius at each center whose interval
excludes zero. This identifies the finest tested scale
with a detected contrast. It does not rule out
contrasts at smaller radii.
Retain the original intervals, normalizing factors
and joint critical value. Selection preserves their
simultaneous coverage. Recalibrating a new band on
the selected subset requires separate justification.

\paragraph{Choosing a configuration on a pilot.}
Split the independent pilot into fitting and validation
parts. For each radius multiplier $h\in\mathcal H$,
construct the candidate landmark family and its retained
coordinate set $I_h$ on the fitting part.
Discard candidates with no retained coordinates.
On the validation part, estimate the contrasts
$\tilde\delta_{h,k}$ and compute normalization
factors $\tilde\omega_{h,k}$ using the intended final
normalization rule and planned inference allocation.
Let $\tilde c_\alpha(h)$ be the multiplier critical value
computed from the validation covariances, using these factors
and the planned inference allocation (Appendix~\ref{app:budget}).
Select a maximizer of
\begin{equation}\label{eq:score}
  \widehat J(h)=
  \max_{k\in I_h} \left\{
    \frac{|\tilde\delta_{h,k}|}{\tilde\omega_{h,k}}
    -\frac{\tilde c_\alpha(h)}{\sqrt{\Neff}}\right\}.
\end{equation}
Here $\Neff$ uses the planned inference sample sizes.
The score targets the strongest anticipated detection
within a candidate family. All selection quantities
use only the pilot.
For the selected family, retain the pilot normalization
factors or recompute them from the inference observations
under the conditions in Appendix~\ref{app:studentization}.
Pilot selection preserves coverage under the conditions
for the chosen normalization regime.
The selected configuration affects power and spatial
resolution. Comparisons therefore hold the total number
of available diagrams fixed
(Section~\ref{sec:exp-budget}).

\paragraph{Cross-validation and sample splitting.}
Cross-validation may replace the pilot score \eqref{eq:score}, provided
all selection uses only the independent pilot. Alternatively, split
each group into two halves. Use one half to select a neighborhood
family, construct a band on the other at nominal coverage
$1-\alpha/2$, and swap roles. Under the conditions for each band,
a union bound gives joint coverage at least $1-\alpha$ minus the sum
of the two calibration errors. Retain the two original interval
families separately. Appendix~\ref{app:cv-screening} studies this
construction for CV selection of regions and radii.

\section{Theory}
\label{sec:theory}

\begin{theorem}[Simultaneous band]\label{thm:band}
Suppose Assumption~\ref{ass:fixed} holds, and the positive
normalizing factors $\homega_k$ are prespecified or
computed from the independent pilot. Assume that $\sup_D\max_{k\le \ell}|\Phi_k(D)|\le B_\nu$,
$\min_{k \le \ell}\omega_k\ge\sigma_{\min}>0$, and
$\lambda_0\le m/N\le1-\lambda_0$ for some $0<\lambda_0\le1/2$. 
There is $C$ depending only on $B_\nu/\sigma_{\min}$ and
$\lambda_0$ such that, for nominal level $0<\alpha<1$, number of retained coordinates $\ell$, 
and number of bootstrap draws $B$, the band \eqref{eq:band} satisfies
\[
  \Bigl|\Prob\{\delta_k\in\mathcal C_k\ \forall k\}-(1-\alpha)\Bigr|
  \le C\Bigl(\frac{\log^7(2\ell N)}{N}\Bigr)^{1/6}+\frac{1}{B+1}.
\]
The same bound holds for any nonempty coordinate subset
fixed independently of the inference observations,
using its own critical value. The probability is taken 
over the inference observations and bootstrap multipliers, 
conditional on the pilot if one is used.
\end{theorem}

The proof is given in Appendix~\ref{app:proofs}.
The bound may be uninformative at small sample sizes;
Section~\ref{sec:exp-calib} assesses finite-sample calibration. The next result allows normalizing factors estimated
from the inference observations.

\begin{corollary}[Estimated normalization]
\label{cor:estimated-normalization}
Suppose the boundedness, nondegeneracy and allocation conditions of
Theorem~\ref{thm:band} hold with $\log^7(2\ell N)/N\to0$, and the
positive factors $\homega_k$, computed from the inference observations
and held fixed across multiplier draws, satisfy
$\max_{k\le\ell}|\homega_k/\tau_{N,k}-1|\cdot\log(2\ell)\xrightarrow{\Prob}0$
for deterministic reference scales $\tau_{N,k}$ with $\omega_k/\tau_{N,k}$
bounded above and below uniformly in $k$ and $N$
(Proposition~\ref{prop:studentized-band}). Then
the band \eqref{eq:band} satisfies
\[
  \left|
  \Prob\{\delta_k\in\mathcal C_k\ \forall k\}
  -(1-\alpha)
  \right|
  \le \varepsilon_N+\frac{1}{B+1},
\]
where $\varepsilon_N\to0$ as $N\to\infty$. The groupwise and pooled constructions in
Appendix~\ref{app:studentization} satisfy these
conditions under their stated assumptions.
\end{corollary}

The proof is given in Appendix~\ref{app:proof_est_norm}.
The guarantee is asymptotic. In our simulations
(\S\ref{sec:exp-calib}) the pooled choice held its level at every
allocation, while the groupwise choice overrejected whenever one
group was three times the size of the other.

\begin{corollary}[Error control and selection]\label{cor:fwer}
For fixed or pilot-based normalization, let $\varepsilon_{N,\ell}$ denote the full error bound
in Theorem~\ref{thm:band}. For estimated normalization, use $\varepsilon_{N,\ell}=\varepsilon_N+1/(B+1)$ under the conditions of
Proposition~\ref{prop:studentized-band}. Let
$H_0=\{k:\delta_k=0\}$. Then $\Prob\{\Rhat\cap H_0\ne\emptyset\}\le\alpha+\varepsilon_{N,\ell}$
and the same holds for $\Rhat_{\mathrm{sd}}\supseteq\Rhat$. On the simultaneous coverage event, every
data-dependent subset of the original intervals covers its targets.
\end{corollary}

The proof is given in Appendix~\ref{app:fwer-proof}. The next result gives a sufficient condition for
detecting local contrasts.

\begin{proposition}[Detection]\label{prop:detect}
Assume the applicable conditions of Theorem~\ref{thm:band} or
Proposition~\ref{prop:studentized-band}, with error bound $\varepsilon_{N,\ell}$ as defined
in Corollary~\ref{cor:fwer}. Let $\hat\rho^2=\max_k\operatorname{Var}^*(Z^*_k)$, where $\operatorname{Var}^*$ is variance conditional on the inference observations. For constants $\kappa_0,\rho_0>0$ and $\gamma_N\in[0,1]$, suppose
$\Prob\left\{
    \max_k\frac{\homega_k}{\omega_k}>\kappa_0
    \ \text{or}\ \hat\rho>\rho_0
  \right\}\le\gamma_N.$
Then, with probability at least $1-\alpha-\varepsilon_{N,\ell}
  -\gamma_N-e^{-\alpha B/8},$
every coordinate satisfying
\[
  |\delta_k|>
  2\kappa_0\rho_0\omega_k
  \sqrt{\frac{2\log(4\ell/\alpha)}{\Neff}}
\]
belongs to $\Rhat$.
\end{proposition}

The proof is given in Appendix~\ref{app:detect-proof}. The constants account for the choice of normalization: $\kappa_0$ bounds the normalizing factors relative to the true coordinate standard deviations, while $\rho_0$ bounds the standard deviations of the normalized bootstrap coordinates. Sparse responses can weaken calibration, so the logarithmic dependence on $\ell$ does not make arbitrarily small neighborhoods reliable. Proposition~\ref{prop:lower} gives a restricted lower bound for identifying one changed coordinate in a fixed disjoint family from the inference sample alone.

The next corollary compares sufficient detection thresholds for inference on the full family using all observations and on a subfamily selected using a pilot. Pilot selection can reduce the number of tested coordinates, but leaves fewer observations for inference.

\begin{corollary}[The price of a pilot]\label{cor:pilot} 
Let $\nu$ be a prespecified family of $\ell$ coordinates. Compare inference on all observations with a procedure
that uses a fraction $f\in(0,1)$ of each group to select a nonempty subfamily $\nu_h\subseteq\nu$ of $\ell_h$
coordinates. The latter uses the remaining observations for inference. The coordinate response functions are
unchanged. Suppose both procedures satisfy Proposition~\ref{prop:detect} at the same nominal level
$\alpha$, with a common product $\kappa_0\rho_0$. For each retained coordinate, the full-family
sufficient detection threshold is smaller exactly when  $(1-f)\log(4\ell/\alpha)<\log(4\ell_h/\alpha)$.
\end{corollary}

A pilot multiplies $\Neff$ by $1-f$ and leaves $\omega_k$ unchanged for each retained coordinate. For the grid of
\S\ref{sec:exp-budget} ($\ell=497$ coordinates over five radii, $\ell_h=105$ at one, $\alpha=0.05$) the full family sufficient detection threshold is smaller once $f>0.15$, under the common $\kappa_0\rho_0$ assumption. The comparison is between sufficient thresholds from a union bound, not between exact powers: a selector can still gain by selecting a subfamily with a smaller $\kappa_0\rho_0$, and \S\ref{sec:exp-localize-budget} measures both effects. 

\paragraph{A contrast that survives averaging.}
Populations with the same mean counting measure have the same mean responses, even when their diagram laws differ. Compare a diagram containing one point at $a$ with probability 1 with a diagram with two copies of $a$ or no point, each with probability $1/2$. These populations have $\delta=0$. A lower-distortion bound for individual diagrams does not by itself imply a nonzero population contrast. The following model gives a sufficient condition for detecting a localized feature displacement.

\begin{assumption}[Localized feature displacement]\label{ass:displaced}
Write $X=X_0\uplus F_P$ and $Y=Y_0\uplus F_Q$ in $\D{n}$, where $\uplus$ denotes multiset union and
$\E_P N_{X_0}=\E_Q N_{Y_0}$. Each added feature is either absent or a singleton, with presence probability $q\in(0,1]$ in each group. Write $A'$ and $B'$ for its location in $X$ and in $Y$ when present. For fixed off-diagonal points $a,b$ and $\epsilon\ge0$,
$\dB(A',a)\le\epsilon$ under $P$ and $\dB(B',b)\le\epsilon$ under $Q$, almost surely.
\end{assumption}

\begin{theorem}[Local detectability]\label{thm:local}
Under Assumption~\ref{ass:displaced}, suppose a retained coordinate $k$ satisfies $\dB(p_k,a)+\epsilon<r_k \le\dB(p_k,b)-\epsilon.$ Then $\delta_k\ge qw_k\{r_k-\dB(p_k,a)-\epsilon\}>0.$ In particular, if $\dB(p_k,a)\le\eta$ for some
$\eta\ge0$ and $\eta+\epsilon<r_k \le\dB(a,b)-\eta-\epsilon,$ then $\delta_k\ge qw_k(r_k-\eta-\epsilon)>0$.
\end{theorem}

Under the conditions of Proposition~\ref{prop:detect}, a coordinate is detected with the probability stated there whenever
its contrast lower bound exceeds its detection threshold.

Under the radius condition, the added feature gives a positive response under $P$ and zero under $Q$
whenever present. The background contributions agree in expectation, so the displaced feature
determines the contrast. The proof is in Appendix~\ref{app:proof-local}. Background and feature need not be
independent. The interval is sufficient, not necessary, and a finite radius grid must actually meet it; capping or screening can remove every admissible radius. For PLACE or PALACE landmarks \citep{PALACE}, this
guarantee applies when the retained family contains a center and radius satisfying these conditions.

\begin{proposition}[Spatial extent]\label{prop:extent}
Let $S=\operatorname{supp}|\bar\mu_P-\bar\mu_Q|$ off the diagonal and $\dB(x,S)=\inf_{s\in S}\dB(x,s)$. If
$\delta_k\ne0$ then $\dB(p_k,S)<r_k$. 
On the coverage event, every rejected center is within its radius of $S$, and the union $\bigcup_{k\in\mathcal R}\mathcal B^\circ_k$ over any reported
index set $\mathcal{R}\subseteq\Rhat$ with $\max_{k\in\mathcal R}r_k\le r_{\max}$ lies in
$\{x:\dB(x,S)<2r_{\max}\}$.
\end{proposition}

A rejected neighborhood must intersect $S$, but need not be contained in $S$. The rejected
neighborhoods need not cover all of $S$.

\paragraph{What lifetime statistics cannot see.}
Call a two-sample statistic \emph{lifetime-based} if it is a function of
the lifetime multisets $\{d-b:(b,d)\in X_i\}$ and $\{d-b:(b,d)\in Y_j\}$;
the log-rank statistic of \citet{STRAND} in either calibration is one.

\begin{proposition}[Blindness along the diagonal]\label{prop:invariance}
If $Q$ is the law of $\tau(X)$, $X\sim P$, where $\tau$ moves each point by
a possibly random vector $(t,t)$ within the frame, every lifetime-based
test has the same rejection probability under $(P,Q)$ as under $(P,P)$. If
some landmark's open square meets the support of the mean measure of the
moved points under $P$ and almost surely none of their images, its
$\delta_k>0$.
\end{proposition}

\section{Experiments}
\label{sec:experiments}

\subsection{Design at a fixed total budget}
\label{sec:exp-budget}

We assess calibration and localization at nominal level $\alpha=0.05$, holding the total number of diagrams fixed. Each replication draws $T_P$ diagrams from $P$ and $T_Q$ from $Q$; balanced designs have $T_P=T_Q=T$. All procedures use the same generated samples, with any pilot drawn from this budget.

In the displacement experiments, each diagram contains $20$ background points uniform above the diagonal in $[0,1]^2$ and a feature present with probability $q=0.5$. When present, the feature lies at its template plus $\epsilon u$, where $\epsilon=0.02$ and $u$ is uniform on $\{-1,0,1\}^2$. The templates are $a=(0.30,0.50)$ under $P$ and either $a+(0,\Delta)$ or $a+(\Delta,\Delta)$ under $Q$. These shifts change or preserve lifetime, respectively. The common background cancels in expectation, and discrete jitter makes each contrast $\delta_k$ an exact finite sum. 

Setting $\Delta=0$ gives the strict null $P=Q$. Two further designs preserve the mean counting measure while changing response covariances. In the \emph{clustered} design, ten background points have identical marginal distributions over four boxes. They choose boxes independently under $P$ and share one box under $Q$; both groups also contain ten uniform background points. In the \emph{counts} design, $P$ has $20$ background points and $Q$ has $0$ or $40$ with equal probability. Both designs use the same feature process in the two groups.

We compare five landmark procedures. F0 uses the $105$-center grid $\{(i/14,j/14):0\le i<j\le14\}$ with $r_k=\min\{0.05,\pi(p_k)\}$. F1 uses the same centers with $r_{k,h}=\min\{0.05h,\pi(p_k)\}$ for $h\in\{\frac12,\frac1{\sqrt2},1,\sqrt2,2\}$. F2 selects $h$ on a pilot using \eqref{eq:score}. A1 fits $100$ class-aware farthest-point centers on the pilot fitting sample and uses all five radius multipliers; A2 also selects $h$ on the pilot. Identical coordinates created by the radius cap are retained once. F0 and F1 use all diagrams for inference, with $\homega_k=nw_kr_k$. F2, A1 and A2 reserve approximately one third of each group for the pilot, divided into fitting and validation parts.

Comparators are the Westfall--Young permutation maximum on F1 (F1-perm), the multiplier band on a $20\times20$ persistence image (I1), and the procedure of \citet{MoonLazar2023} (ML). Our ML implementation uses a $40\times40$ image, filtering at the $80$th percentile of pooled standard deviation, pooled $t$-tests, and Benjamini--Hochberg FDR adjustment.

For landmark methods, a change is \emph{located} when at least one rejected coordinate has a nonzero population contrast, radius at most $0.10$, and center within $0.05$ of the changed support. Appendix~\ref{app:budget} specifies the localization distance, pixel criterion, and implementation details, including pilot normalization and screening.

\subsection{Calibration}
\label{sec:exp-calib}

\begin{table}[t]
\centering\small
\caption{FWER at a fixed total budget, $1000$ replications per cell.
$T_P/T_Q$: diagrams per class; in the weak nulls $Q$ is the more
variable group. Bold: the $95\%$ Wilson interval lies above $0.05$.}
\label{tab:calib}
\IfFileExists{tables/neigh_calib.tex}{\resizebox{\columnwidth}{!}{\input{tables/neigh_calib.tex}}}{\todo{tab\_neigh\_calib.tex}}
\end{table}

With their default normalization, F0--A2 and I1 have family-wise error rate (FWER) $0.022$--$0.063$ across the ten settings in Table~\ref{tab:calib}. Simultaneous coverage is $0.94$--$0.98$ (Appendix~\ref{app:budget}). Under the weak nulls, the permutation maximum and pooled-$t$ ML implementation overreject when the smaller
group is more variable. Their FWER reaches $0.32$ and $0.26$ in the clustered design, and $0.15$ and $0.11$ in the counts design, respectively. This is the allocation in which pooling can understate
sampling variability. Both comparators are approximately calibrated under $P=Q$ and conservative when the more variable group is larger.

Inference-based normalization (Corollary~\ref{cor:estimated-normalization}) gives different
finite-sample results. For F0 and F1, groupwise SD normalization gives FWER $0.029$--$0.055$ with balanced groups, but $0.29$--$0.48$ at $40/120$ and $120/40$, including under $P=Q$. Pooled SD normalization is conservative, with FWER $0.002$--$0.036$ across all settings (Appendix~\ref{app:budget}).

Normalization also affects power. In a separate balanced design with highly unequal coordinate variances, both SD-based choices locate the change in every run; fixed normalization locates none. In the main displacement design, fixed normalization performs better: at $40$ diagrams per class and $\Delta=0.05$, F0 locates the change in $77\%$ of runs with fixed normalization and $37\%$ with pooled SD normalization.

\subsection{Localization and neighborhood choice}
\label{sec:exp-localize-budget}

\begin{figure*}[t]
  \centering
  \IfFileExists{figures/budget.pdf}{\includegraphics[width=\textwidth]{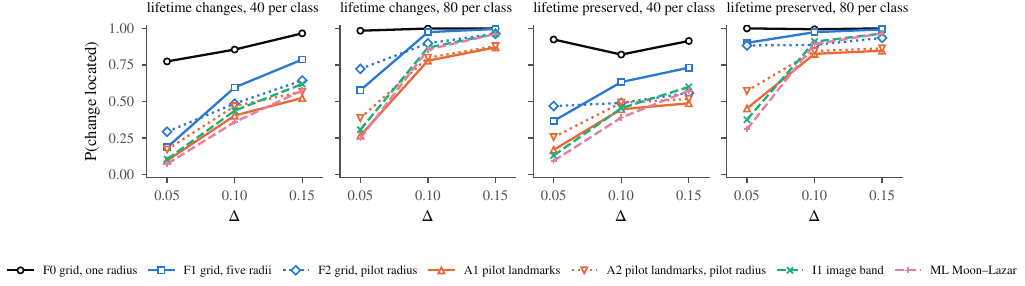}}{\todo{\texttt{\detokenize{figures/budget.pdf}}}}
  \caption{Probability of locating the change at a fixed total budget,
  $1000$ replications per point. The pilot methods (F2, A1, A2) infer on
  two thirds of the diagrams; F0, F1 and the image procedures use all of
  them. $120$ per class and the FWER under alternatives are in
  Appendix~\ref{app:budget}.}
  \label{fig:budget}
\end{figure*}

\begin{figure*}[t]
  \centering
  \IfFileExists{figures/radius.pdf}{\includegraphics[width=\textwidth]{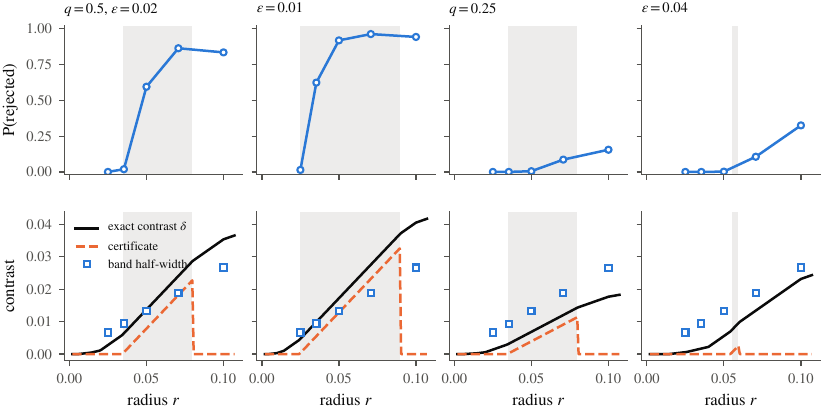}}{}
  \caption{Radius and Theorem~\ref{thm:local}, $80$ diagrams per class,
  $\Delta=0.10$. Top: probability that the coordinate at the grid center
  nearest the template is rejected by F1, at each of its radii. Bottom:
  exact contrast and its certificate on a dense radius sweep, with the
  band's mean half-width. Shaded: the certified interval
  $s+\epsilon<r\le t-\epsilon$.}
  \label{fig:radius}
\end{figure*}

In the main displacement experiments, F0 has the highest localization rates among the compared procedures (Figure~\ref{fig:budget}). At $40$ diagrams per class and $\Delta=0.05$, it locates lifetime-changing and lifetime-preserving shifts in $77\%$ and $92\%$ of runs, respectively.
The corresponding ranges for the pilot procedures are $10$--$29\%$ and $17$--$47\%$. Testing five radii jointly reduces localization for small shifts. For lifetime-changing shifts with $\Delta=0.05$, F1 reaches $19\%$ at $40$ per class and $58\%$ at $80$ per class,
compared with $77\%$ and $98\%$ for F0. The difference is small once $\Delta\ge0.10$ with at least $80$ per class. Under alternatives, the landmark bands have FWER at or below $0.05$ over exactly null
coordinates.

Sensitivity analyses show that pilot size and image settings affect localization (Appendix~\ref{app:budget}). At $80$ per class and $\Delta=0.10$, A1 locates the change in $87\%$, $75\%$ and $56\%$ of runs with pilot fractions of one quarter, one third and one half.
Increasing image resolution from $10\times10$ to $40\times40$ raises localization from $12\%$ to $99.8\%$, matching F0 at the finer setting. These image variants use a one-pixel Gaussian bandwidth, so the comparison changes both resolution and smoothing. The default pilot scale floor is active at every coordinate. Results for other floors and center counts are reported in the appendix.

Theorem~\ref{thm:local} guarantees a nonzero contrast at suitable radii, but this need not give high power (Figure~\ref{fig:radius}). At the grid center nearest the template, radii $0.035$, $0.05$ and $0.071$ all lie in the sufficient interval, approximately $(0.034,0.080]$. Their rejection rates are $2\%$, $60\%$ and $86\%$. Increasing jitter to $\epsilon=0.04$ narrows the interval and reduces detection to $0$--$33\%$. Designs with two features at different scales, a prevalence change, or continuous jitter retain coverage of $0.94$--$0.98$. However, the weaker, wider feature in the two-feature design is detected in only $3$--$17\%$ of runs.

Appendix~\ref{app:cv-screening} examines cross-validated region and
radius selection. Cross-fitted selection gained power for several
nearby changes of the same sign, but lost power for an isolated change.
The gain persisted against a fixed-radius family. The SD-normalized
multiplier bands undercovered, so the positive comparison concerns
testing under exchangeability.

\paragraph{What lifetimes miss.} In the additional comparisons of Appendix~\ref{app:results}, both
calibrations of the log-rank test of \citet{STRAND} remain near their null rejection rates under diagonal shifts, as Proposition~\ref{prop:invariance} implies. The band and an MMD test on
the same landmark embedding have similar global power in these experiments. The band also provides simultaneous confidence intervals for the local contrasts.

\subsection{An exploratory molecular application}
\label{sec:exp-mol}

We analyze MUTAG, DHFR, COX2 and PTC, four molecular graph benchmarks with binary class labels \citep{Morris2020}. After removing duplicates by labeled graph isomorphism and dropping duplicate sets with conflicting labels, $175$, $718$, $465$ and $325$ graphs remain. We compute extended persistence diagrams in degrees zero and one from the heat-kernel signature of the normalized Laplacian \citep{PALACE}. Its values lie in $[0,1]$, giving a fixed frame. F0 and F1 use all
graphs; F2, A1 and A2 reserve a stratified third for the pilot. The focus on heat-kernel time $t=10$ and degree one follows the earlier analysis in \citet{PALACE}. These choices were not prospective, so the application is exploratory.

MUTAG and DHFR yield nonempty maps for every method (global $p\le0.024$; Appendix~\ref{app:mol}). For MUTAG, A1 gives simultaneous intervals excluding zero for $58$ of $342$ coordinates at $26$ centers (Figure~\ref{fig:map}). All indicate greater expected weighted feature mass in the mutagenic class. The map concentrates near birth values $0.06$--$0.09$ and death values $0.09$--$0.12$. Across $20$ random pilot/inference splits, the A1 map is never empty (DHFR map in Figure~\ref{fig:dhfr-rev}), with median signed overlap $0.62$ between splits. COX2 maps are unstable, and PTC maps are empty in $16$--$20$ of the $20$ splits, depending on the method.

With $20$ molecules per class, A1 gives nonempty maps in $94\%$ of subsamples, compared with $28\%$ for the fixed grid. The diagrams occupy only $[0,0.3]^2$ within the $[0,1]$ frame, making the fixed grid coarse in the occupied region. Here, adapting the neighborhoods helps despite the observations reserved for the pilot.

For an exploratory structural interpretation, we match rings from a minimum cycle basis to diagram points. The selected MUTAG region contains $96\%$ of rings in systems of at least three fused rings and $49\%$ of other rings. A different cycle basis gives $95\%$ and $49\%$. Among compounds with $15$--$22$ atoms, mutagenic molecules place $2.3$--$3.2$ rings in this region, compared with $0.8$--$1.3$ for nonmutagenic molecules. These summaries are consistent with the fused-ring indicator of \citet{Debnath1991}. Their interpretation is limited by selection of the region on these data and by the dependence of ring matching on the cycle basis. Degree-one points can also represent branchings of superlevel sets, and rings with symmetric heat-kernel values can produce no point. We do not infer a chemical mechanism.

\section{Discussion}
\label{sec:discussion}

Simultaneous intervals quantify the size and direction of local mean
differences at a stated resolution. In our main fixed-budget
simulations, a prespecified single-radius grid using all observations
localized changes most often. This is a useful starting point when a
meaningful scale can be chosen in advance. When the scale is uncertain,
several radii can be calibrated jointly. Pilot fitting can help when
the fixed grid poorly resolves the occupied region, as in the MUTAG
subsampling results.

Cross-validated region selection can help when several nearby changes
have the same sign and can be identified reliably from the selection
sample. Our additional experiment illustrates this advantage, but also
shows a substantial loss of power for an isolated change
(Appendix~\ref{app:cv-screening}). The multiplier bands had coverage
below nominal there, so the positive comparison concerns power rather
than validated calibration.

The method targets local mean contrasts and can miss differences that
cancel under these responses. Reported neighborhoods need not cover
the whole support of the mean-measure difference, and unselected
regions have not been shown to lack changes. Improving finite-sample
calibration and recalibrating after selection on the inference sample
remain directions for further work. Recovering the full changed
support and incorporating covariates are also natural extensions.

\section*{AI Use Statement}
The authors formulated the research questions and planned the methodological and empirical work. Claude (Anthropic) and ChatGPT (OpenAI) were used to write and run the experiment and analysis code, generate the tables and figures, and draft and edit text. The authors directed the work,
checked the results and arguments, and take responsibility for the paper.

\bibliographystyle{plainnat}
\bibliography{main}

\appendix
\numberwithin{theorem}{section}
\numberwithin{proposition}{section}
\numberwithin{corollary}{section}
\numberwithin{lemma}{section}
\numberwithin{remark}{section}

\onecolumn

\input{Appendix_Part_I_Technical}
\input{Appendix_Part_II_Empirical}

\end{document}

%% file: tables/neigh_calib.tex
\begin{tabular}{l r c c c c c c c c}
\toprule
design & $T_P$/$T_Q$ & F0 & F1 & F2 & A1 & A2 & F1-perm & I1 & ML \\
\midrule
strict null & 40/40 & 0.032 & 0.033 & 0.036 & 0.055 & 0.047 & 0.039 & 0.027 & 0.019 \\
strict null & 80/80 & 0.035 & 0.048 & 0.046 & 0.046 & 0.044 & 0.053 & 0.048 & 0.026 \\
strict null & 120/120 & 0.040 & 0.043 & 0.051 & 0.057 & 0.049 & 0.042 & 0.041 & 0.020 \\
strict null & 40/120 & 0.033 & 0.054 & 0.036 & 0.038 & 0.034 & 0.059 & 0.042 & 0.032 \\
\addlinespace
clustered & 80/80 & 0.063 & 0.054 & 0.059 & 0.052 & 0.052 & 0.051 & 0.057 & 0.038 \\
clustered & 120/40 & 0.047 & 0.041 & 0.062 & 0.052 & 0.057 & \textbf{0.324} & 0.038 & \textbf{0.262} \\
clustered & 40/120 & 0.045 & 0.050 & 0.063 & 0.054 & 0.054 & 0.001 & 0.053 & 0.000 \\
\addlinespace
counts & 80/80 & 0.046 & 0.045 & 0.034 & 0.042 & 0.032 & 0.050 & 0.046 & 0.017 \\
counts & 120/40 & 0.032 & 0.034 & 0.022 & 0.038 & 0.030 & \textbf{0.149} & 0.040 & \textbf{0.105} \\
counts & 40/120 & 0.040 & 0.043 & 0.044 & 0.045 & 0.035 & 0.012 & 0.049 & 0.002 \\
\bottomrule
\end{tabular}

%% file: Appendix_Part_I_Technical.tex
\section*{Part I: Technical results}

\section{Proofs of the main results}
\label{app:proofs}

The proofs follow the order of the main results. The landmark
family is prespecified or learned from an independent pilot.
Whenever a pilot is used, statements concerning the inference
sample are conditional on it.

\subsection{Simultaneous coverage: proof of Theorem~\ref{thm:band}}
\label{app:band-proof}

\begin{proof}
Condition on the independent pilot, if one is used. The landmark
family and the positive normalization factors $\homega_k$ are then
fixed. All probabilities below use this conditional law.
Write $\Prob^*$ for probability conditional also on the inference
observations. Put
\[
  R=\frac{B_\nu}{\sigma_{\min}},\qquad
  H_N=\log(2\ell N),\qquad
  r_N=\left(\frac{H_N^7}{N}\right)^{1/6}.
\]
Constants denoted by $C$ may change between displays and depend
only on $R$ and $\lambda_0$. For $N<4$, the claimed bound follows
by enlarging $C$, since a difference of probabilities is at most
one. Hence assume $N\ge4$.

Define
\[
  \mu_P=\E U_1,\qquad \mu_Q=\E V_1,\qquad
  S=\sqrt{\Neff}(\hdelta-\delta),\qquad
  T=\max_{k\le\ell}\frac{|S_k|}{\homega_k}.
\]
For one multiplier replicate, let
\[
  S^*=\sqrt{\Neff}\left\{
    \frac1m\sum_{i=1}^m e_i(U_i-\bar U)
    -\frac1{m'}\sum_{j=1}^{m'}e'_j(V_j-\bar V)
  \right\},\qquad
  T^*=\max_{k\le\ell}\frac{|S_k^*|}{\homega_k}.
\]
Thus $T^*=\max_k|Z_k^*|$, and the coverage event is
$\{T\le\hc\}$.

\medskip\noindent
\emph{1. Gaussian approximation.}
Write $S=\sum_{i=1}^N\xi_i$, where
\[
  \xi_i=\frac{\sqrt{\Neff}}m(U_i-\mu_P),\quad i\le m,
  \qquad
  \xi_{m+j}=-\frac{\sqrt{\Neff}}{m'}(V_j-\mu_Q),
  \quad j\le m'.
\]
These vectors are independent and centered, and
\[
  \operatorname{Cov}(S)
  =\frac{m'}N\Sigma_P+\frac mN\Sigma_Q=\Gamma.
\]
Set $W_i=\sqrt N\,\xi_i/\sigma_{\min}$. Then
\[
  \frac{S}{\sigma_{\min}}=\frac1{\sqrt N}\sum_{i=1}^N W_i,
  \qquad
  \max_i\|W_i\|_\infty\le K:=\frac{2R}{\sqrt{\lambda_0}},
  \qquad
  \frac1N\sum_{i=1}^N\E W_{ik}^2
  =\frac{\omega_k^2}{\sigma_{\min}^2}\ge1.
\]
The boundedness condition supplies the moment and exponential-moment
conditions (M.2) and (E.1) of \citet{Chernozhukov2017}, with an
envelope depending only on $K$. Proposition~2.1 of that paper
therefore gives, for $G\sim\mathcal N(0,\Gamma)$,
\begin{equation}\label{eq:band-proof-clt}
  \sup_{A\in\mathcal A_\ell}
  |\Prob(S\in A)-\Prob(G\in A)|\le C r_N,
\end{equation}
where $\mathcal A_\ell$ is the class of axis-aligned
hyperrectangles in $\Real^\ell$. If $\ell<3$, duplicate coordinates
when applying the cited result. This changes only the constant.

Let $F$ be the distribution function of
$\max_k|G_k|/\homega_k$. Since $\homega_k>0$ is fixed,
the event $\{T\le t\}$ is a hyperrectangle for $t\ge0$.
For $t<0$, both relevant probabilities are zero. Hence
\begin{equation}\label{eq:band-proof-cdf}
  \sup_{t\in\Real}|\Prob(T\le t)-F(t)|\le C r_N.
\end{equation}
No uniform upper or lower bound on the normalizing factors is needed
for this step beyond their positivity and finiteness.

\medskip\noindent
\emph{2. Within-group covariance estimation.}
Define the empirical covariances with divisors $m$ and $m'$:
\[
  \widehat\Sigma_P=
  \frac1m\sum_{i=1}^m(U_i-\bar U)(U_i-\bar U)^\top,
  \qquad
  \widehat\Sigma_Q=
  \frac1{m'}\sum_{j=1}^{m'}(V_j-\bar V)(V_j-\bar V)^\top.
\]
The independence and normality of the multipliers imply
\[
  S^*\mid\text{inference observations}
  \sim\mathcal N(0,\widehat\Gamma),\qquad
  \widehat\Gamma=
  \frac{m'}N\widehat\Sigma_P+\frac mN\widehat\Sigma_Q.
\]
For a matrix $M$, write $\|M\|_{\max}=\max_{j,k}|M_{jk}|$.
To control the effect of estimating the group means, put
\[
  Q_P=\frac1m\sum_{i=1}^m(U_i-\mu_P)(U_i-\mu_P)^\top.
\]
Define $Q_Q$ analogously. Then
\[
  \widehat\Sigma_P-\Sigma_P
  =(Q_P-\Sigma_P)
  -(\bar U-\mu_P)(\bar U-\mu_P)^\top,
\]
with the same identity for $Q$.

Let $x_N=\log(8\ell^2N)$. Hoeffding's inequality gives,
simultaneously for both groups, with probability at least $1-1/N$,
\[
  \|\bar U-\mu_P\|_\infty
  \le B_\nu\sqrt{\frac{2x_N}{m}},\qquad
  \|Q_P-\Sigma_P\|_{\max}
  \le4B_\nu^2\sqrt{\frac{2x_N}{m}},
\]
and the corresponding bounds with $m'$ and $Q$ in place of $m$
and $P$. Indeed, the mean coordinates have range at most
$2B_\nu$, and each centered product lies in
$[-4B_\nu^2,4B_\nu^2]$. A union bound over the two groups'
$\ell$ means and $\ell^2$ products has failure probability at most
$4(\ell+\ell^2)e^{-x_N}\le1/N$.

Since $m,m'\ge\lambda_0N$, it follows that
\[
  \Delta_N:=
  \frac{\|\widehat\Gamma-\Gamma\|_{\max}}{\sigma_{\min}^2}
  \le d_N:=R^2\left\{
    4\sqrt{\frac{2x_N}{\lambda_0N}}
    +\frac{2x_N}{\lambda_0N}
  \right\}
\]
with probability at least $1-1/N$. Write
$\mathcal E_N=\{\Delta_N\le d_N\}$. Since $x_N\le3H_N$,
\begin{equation}\label{eq:band-proof-cov}
  \Prob(\mathcal E_N^c)\le N^{-1},\qquad
  d_N\le C\left\{
    \sqrt{\frac{H_N}{N}}+\frac{H_N}{N}
  \right\}.
\end{equation}

It remains to consider $r_N$ smaller than a sufficiently small
constant depending only on $R$ and $\lambda_0$. For larger
$r_N$, the theorem again follows by enlarging $C$.
Since $H_N\ge1$, we have $H_N/N\le r_N^6$.
Thus, in the remaining case, we may assume
$H_N/N\le1$ and $d_N\le1/2$. On $\mathcal E_N$,
\begin{equation}\label{eq:band-proof-positive}
  \widehat\Gamma_{kk}
  \ge\Gamma_{kk}-\sigma_{\min}^2d_N
  \ge\frac{\sigma_{\min}^2}{2},\qquad k\le\ell.
\end{equation}

\medskip\noindent
\emph{3. Conditional Gaussian comparison.}
Let $F^*(t)=\Prob^*(T^*\le t)$. Apply the Gaussian covariance
comparison in the proof of Theorem~4.1 and Remark~4.1 of
\citet{Chernozhukov2017} to $G/\sigma_{\min}$ and
$S^*/\sigma_{\min}$. Their covariance discrepancy is $\Delta_N$,
and the marginal variances of the reference Gaussian are at
least one. On $\mathcal E_N$, this yields
\[
  \sup_{A\in\mathcal A_\ell}
  |\Prob^*(S^*\in A)-\Prob(G\in A)|
  \le C d_N^{1/3}\log^{2/3}(2\ell).
\]
This comparison uses the conditional Gaussian law of $S^*$;
the covariance concentration for our within-group centering
was established separately in Step~2. Coordinate duplication
again handles $\ell<3$.

Applying the comparison to the rectangles defining the maximum gives,
on $\mathcal E_N$,
\begin{equation}\label{eq:band-proof-bootstrap-cdf}
  \sup_{t\in\Real}|F^*(t)-F(t)|
  \le s_N:=C d_N^{1/3}\log^{2/3}(2\ell)
  \le C\left(\frac{H_N^5}{N}\right)^{1/6}
  \le C r_N.
\end{equation}
Here we used $H_N/N\le1$ in the middle inequality.

The distribution function $F$ is continuous. For each $t\ge0$,
the event that the Gaussian maximum equals $t$ is contained in
the union of the events $\{|G_k|=t\homega_k\}$, each of probability
zero. The same argument and \eqref{eq:band-proof-positive}
show that $F^*$ is continuous on $\mathcal E_N$.
Neither statement requires an invertible covariance matrix.
Moreover, $F$ is strictly increasing on $[0,\infty)$: on the
nonzero linear support of $G$, the map
$g\mapsto\max_k|g_k|/\homega_k$ is a norm, and the Gaussian law
assigns positive probability to every nonempty open annulus.

\medskip\noindent
\emph{4. Calibration at a data-dependent critical value.}
Define $W=F(T)$ and $V=F^*(T)$.
For $0<v<1$, continuity and strict increase give
\[
  \Prob(W\le v)=\Prob\{T\le F^{-1}(v)\}.
\]
Here $F^{-1}(v)$ is the unique nonnegative $t$ satisfying $F(t)=v$.
Equation~\eqref{eq:band-proof-cdf}, including its value at zero,
therefore implies
\[
  \sup_{0\le v\le1}|\Prob(W\le v)-v|\le C r_N.
\]
On $\mathcal E_N$, \eqref{eq:band-proof-bootstrap-cdf} gives
$|V-W|\le s_N$. Hence, for $0\le v\le1$,
\[
  \Prob(W\le v-s_N)-N^{-1}
  \le\Prob(V\le v)
  \le\Prob(W\le v+s_N)+N^{-1}.
\]
Consequently,
\begin{equation}\label{eq:band-proof-uniform}
  \sup_{0\le v\le1}|\Prob(V\le v)-v|
  \le C r_N+s_N+N^{-1}\le C r_N.
\end{equation}
This argument allows $T$ and $F^*$ to depend on the same observations.

\medskip\noindent
\emph{5. Finitely many multiplier draws.}
Let $T_1^*,\ldots,T_B^*$ be the independent multiplier maxima
conditional on the inference observations. Write
$j=\lceil B(1-\alpha)\rceil$, so $1\le j\le B$, and let
$\hc$ be their $j$th order statistic. On $\mathcal E_N$,
continuity of $F^*$ gives
\[
  \Prob^*(T\le\hc)=g_{B,j}(V),\qquad
  g_{B,j}(v)=
  \sum_{q=0}^{j-1}\binom Bq v^q(1-v)^{B-q}.
\]
Indeed, $T\le\hc$ holds exactly when at most $j-1$ of the
draws are strictly below $T$. Since both sides lie in $[0,1]$
on the complementary event,
\[
  \left|\Prob(T\le\hc)-\E g_{B,j}(V)\right|
  \le\Prob(\mathcal E_N^c)\le N^{-1}.
\]
The function $g_{B,j}$ is decreasing, with
$g_{B,j}(0)=1$ and $g_{B,j}(1)=0$. Integration by parts and
\eqref{eq:band-proof-uniform} yield
\[
  \left|\E g_{B,j}(V)-\int_0^1g_{B,j}(v)\,dv\right|
  \le
  \sup_{0\le v\le1}|\Prob(V\le v)-v|
  \le C r_N.
\]
Each summand in the integral has value $1/(B+1)$, so
\[
  \int_0^1g_{B,j}(v)\,dv=\frac{j}{B+1},\qquad
  \left|\Prob(T\le\hc)-\frac{j}{B+1}\right|\le C r_N.
\]
Finally, $j=\lceil B(1-\alpha)\rceil\in[B(1-\alpha),B(1-\alpha)+1)$, so
$j-(1-\alpha)(B+1)\in[-(1-\alpha),\alpha)$ and
\[
  \left|\frac{j}{B+1}-(1-\alpha)\right|
  \le\frac{1}{B+1}.
\]
Since $\{T\le\hc\}$ is the simultaneous coverage event,
the triangle inequality proves the displayed bound in the theorem.
The approximation constant is independent of $B$ and $\alpha$.

For a nonempty subset $I\subseteq\{1,\ldots,\ell\}$ fixed
conditional on the pilot, apply the same argument to
$(U_{ik})_{k\in I}$ and $(V_{jk})_{k\in I}$.
The envelope, allocation and variance bounds still hold.
Using the subset's own multiplier maximum gives the same result
with $|I|$ in place of $\ell$, which is no larger than the stated bound.
\end{proof}

\subsection{Estimated normalization}
\label{app:studentization}
\label{app:proof_est_norm}

Theorem~\ref{thm:band} treats normalization factors fixed
independently of the inference observations. We now allow
estimated factors that approach deterministic reference scales.
These scales need not equal the coordinate SDs $\omega_k$.
The landmark family must still satisfy
Assumption~\ref{ass:fixed}. All probability statements below
condition on the independent pilot, if one is used.

\begin{proposition}[Estimated normalization]
\label{prop:studentized-band}
Let $\ell=\ell_N\ge1$ and $\lambda_N=m/N$. Suppose that,
for constants $M,\sigma_0>0$ and $0<\lambda_0\le1/2$
independent of $N$,
\[
  \sup_D\max_{k\le\ell}|\Phi_k(D)|\le M,\qquad
  \min_{k\le\ell}\omega_k\ge\sigma_0,\qquad
  \lambda_0\le\lambda_N\le1-\lambda_0,
  \qquad \frac{\log^7(2\ell N)}{N}\longrightarrow0.
\]
Let $\tau_{N,k}>0$ be deterministic reference scales satisfying
\[
  0<c_0\le\frac{\omega_k}{\tau_{N,k}}\le C_0<\infty
  \qquad\text{for every }k\le\ell,
\]
where $c_0,C_0$ do not depend on $N$.
Suppose the positive factors $\homega_k$ satisfy
\begin{equation}\label{eq:relative-normalization}
  a_N=\max_{k\le\ell}
  \left|\frac{\homega_k}{\tau_{N,k}}-1\right|,
  \qquad a_N\log(2\ell)\xrightarrow{\Prob}0.
\end{equation}
Compute the band \eqref{eq:band} with these factors held
fixed across Gaussian multiplier draws. For fixed
$0<\alpha<1$ and any sequence of positive integers $B=B_N$,
there is $\varepsilon_N\to0$ such that
\[
  \left|\Prob\{\delta_k\in\mathcal C_k\text{ for all }k\le\ell\}
      -(1-\alpha)\right|
  \le\varepsilon_N+\frac{1}{B+1}.
\]
In particular, coverage converges to $1-\alpha$ if $B\to\infty$.
The same conclusion holds for any nonempty coordinate subset
fixed independently of the inference observations, using its
own critical value. Any data-dependent selection of the
original intervals retains their joint coverage guarantee.
\end{proposition}

\begin{proof}
Put $L_N=\log(2\ell)$ and define
\[
  S_k=\sqrt{\Neff}(\hdelta_k-\delta_k),\qquad
  S_k^*=\homega_k Z_k^*,\qquad
  T_\tau=\max_k\frac{|S_k|}{\tau_{N,k}},\qquad
  T_\tau^*=\max_k\frac{|S_k^*|}{\tau_{N,k}}.
\]
Let $G\sim\mathcal N(0,\Gamma)$ and let $F_N$ be the
distribution function of $\max_k|G_k|/\tau_{N,k}$.
The Gaussian approximation and covariance comparison in the
proof of Theorem~\ref{thm:band}, applied to the deterministic
scales $\tau_{N,k}$, give
\[
  \sup_t|\Prob(T_\tau\le t)-F_N(t)|\longrightarrow0,
  \qquad
  \sup_t|\Prob^*(T_\tau^*\le t)-F_N(t)|
  \xrightarrow{\Prob}0.
\]
Here $\Prob^*$ conditions on the inference observations.
These approximations follow from
\citet[Proposition~2.1 and Theorem~4.1]{Chernozhukov2017}.
The scale bounds keep the normalized coordinate variances
bounded above and away from zero. They also imply
$T_\tau=O_{\Prob}(\sqrt{L_N})$ and the corresponding
conditional bound for $T_\tau^*$, in probability.

Write $T=\max_k|S_k|/\homega_k$ and
$T^*=\max_k|Z_k^*|$. On $\{a_N<1/2\}$,
\[
  \frac{T_\tau}{1+a_N}\le T\le\frac{T_\tau}{1-a_N},
  \qquad
  \frac{T_\tau^*}{1+a_N}\le T^*
  \le\frac{T_\tau^*}{1-a_N}.
\]
Thus $|T-T_\tau|=o_{\Prob}(L_N^{-1/2})$, with the
analogous conditional statement for $T^*$.
Applying Nazarov's inequality to the $2\ell$ Gaussian
coordinates $G_k/\tau_{N,k}$ and $-G_k/\tau_{N,k}$ gives
\[
  \sup_t\{F_N(t+u)-F_N(t)\}\le C u\sqrt{L_N},
  \qquad u>0;
\]
see \citet[Lemma~A.1]{Chernozhukov2017}.
The replacement errors and this inequality therefore yield
\[
  \sup_t|\Prob(T\le t)-F_N(t)|\longrightarrow0,
  \qquad
  \sup_t|\Prob^*(T^*\le t)-F_N(t)|
  \xrightarrow{\Prob}0.
\]

Let $F_N^*$ be the conditional distribution function of $T^*$.
The two preceding approximations imply that
$V_N=F_N^*(T)$ approaches the uniform distribution on $[0,1]$
in Kolmogorov distance, by the argument in Step~4 of
Appendix~\ref{app:band-proof}. Covariance consistency and the
variance lower bound also imply that $F_N^*$ is continuous
with probability tending to one. Step~5 of that proof therefore
applies, with $j_B=\lceil B(1-\alpha)\rceil$, and gives
\[
  \Prob\{T\le\hc\}=\frac{j_B}{B+1}+o(1),
  \qquad
  \left|\frac{j_B}{B+1}-(1-\alpha)\right|
  \le\frac{1}{B+1}.
\]
The $o(1)$ term is uniform over $B_N$: the binomial-tail
function used in Step~5 has total variation at most one for
every $B$ and $j_B$. This proves the coverage claim. The same
argument applies to a coordinate subset fixed independently of
the inference observations. Coverage of any selected original
intervals follows from coverage of the full family.
\end{proof}

\paragraph{Two estimated choices.}
For $m,m'\ge2$, write
\[
  s_{P,k}^2=\frac1{m-1}\sum_{i=1}^m(U_{ik}-\bar U_k)^2,
  \qquad
  s_{Q,k}^2=\frac1{m'-1}\sum_{j=1}^{m'}(V_{jk}-\bar V_k)^2.
\]
Let $\eta_N>0$ be deterministic with $\eta_N\to0$.
The groupwise choice is
\begin{equation}\label{eq:groupwise-normalization}
  \homega_k^{\mathrm{grp}}
  =\left[\max\left\{
    (1-\lambda_N)s_{P,k}^2+\lambda_N s_{Q,k}^2,
    \eta_N^2\right\}\right]^{1/2},
\end{equation}
whose reference scale is $\tau_{N,k}=\omega_k$.

Alternatively, ignore group labels when computing the
normalizer. With $\bar W_k=\lambda_N\bar U_k+(1-\lambda_N)\bar V_k$,
define
\begin{align}
  s_{\mathrm{pool},k}^2
  &=\frac{\sum_{i=1}^m(U_{ik}-\bar W_k)^2
          +\sum_{j=1}^{m'}(V_{jk}-\bar W_k)^2}{N-1},
          \label{eq:pooled-variance}\\
  \homega_k^{\mathrm{pool}}
  &=\left[\max\{s_{\mathrm{pool},k}^2,\eta_N^2\}\right]^{1/2}.
          \label{eq:pooled-normalization}
\end{align}
Writing $\sigma_{P,k}^2=\operatorname{Var}_P\{\Phi_k(X)\}$
and $\sigma_{Q,k}^2=\operatorname{Var}_Q\{\Phi_k(Y)\}$,
the corresponding reference scale satisfies
\begin{equation}\label{eq:pooled-reference-scale}
  \tau_{N,k}^2
  =\lambda_N\sigma_{P,k}^2
   +(1-\lambda_N)\sigma_{Q,k}^2
   +\lambda_N(1-\lambda_N)\delta_k^2.
\end{equation}
The between-group term is present because the pooled variance
is centered at the grand mean. In fact,
\[
  (N-1)s_{\mathrm{pool},k}^2
  =(m-1)s_{P,k}^2+(m'-1)s_{Q,k}^2
    +\Neff\,\hdelta_k^2.
\]
This is not the pooled within-group variance.

Under the boundedness, nondegeneracy and allocation conditions
of Proposition~\ref{prop:studentized-band}, both constructions
satisfy its scale conditions. For the pooled choice,
$\lambda_N\sigma_{P,k}^2+(1-\lambda_N)\sigma_{Q,k}^2$
is bounded above and below by fixed positive multiples of
$\omega_k^2$. The bounded contrast $|\delta_k|\le2M$ and
$\omega_k\ge\sigma_0$ preserve this comparability after
adding the last term in \eqref{eq:pooled-reference-scale}.

Concentration of the bounded first and second sample moments,
followed by a union bound over coordinates, gives for either
choice
\[
  \max_{k\le\ell}
  \left|\frac{\homega_k}{\tau_{N,k}}-1\right|
  =O_{\Prob}\!\left(\sqrt{\frac{\log(2\ell)}{N}}\right).
\]
The reference scales are uniformly bounded away from zero,
so the vanishing floor is inactive with probability tending
to one. The dimension condition implies
$\log^3(2\ell)/N\to0$, which verifies
\eqref{eq:relative-normalization}. For fixed $\ell$,
ordinary relative consistency suffices.

Corollary~\ref{cor:estimated-normalization} follows from
Proposition~\ref{prop:studentized-band} and the preceding
verification of the two normalization constructions.

\paragraph{Calibration and scope.}
Both choices retain the separately centered multiplier
sum in Section~\ref{sec:method}. Its unnormalized conditional
covariance is
\[
  \widehat\Gamma
  =(1-\lambda_N)\widehat\Sigma_P
    +\lambda_N\widehat\Sigma_Q,
\]
where the within-group empirical covariances use divisors
$m$ and $m'$, respectively. Pooled normalization therefore
does not impose equal group covariances. Its factor is held
fixed across multiplier draws, just as for the groupwise choice.

The subset conclusion gives the error-control and step-down
guarantees of Corollary~\ref{cor:fwer}, with FWER at most
$\alpha+\varepsilon_N+1/(B+1)$.
Appendix~\ref{app:fwer-proof} gives the argument for both
normalization regimes.

These are asymptotic guarantees. The allocation condition
allows unequal group sizes, but finite-sample accuracy still
depends on the smaller group, response variability and the
number of coordinates. The floor ensures positive factors;
it does not by itself ensure accurate calibration.

\subsection{Error control and selected reporting: proof of Corollary~\ref{cor:fwer}}
\label{app:fwer-proof}

\begin{proof}
Condition on the independent pilot, if one is used.
Write
\[
  \mathcal E=
  \{\delta_k\in\mathcal C_k\text{ for all }k\}.
\]
On $\mathcal E$, every $k\in H_0$ satisfies
$|Z_k|\le\hc$. Thus no true null is rejected by
$\Rhat$, and
\[
  \Prob\{\Rhat\cap H_0\ne\emptyset\}
  \le\Prob(\mathcal E^c)
  \le\alpha+\varepsilon_{N,\ell}.
\]

For the step-down procedure, the result is immediate
if $H_0=\varnothing$. Otherwise, for each nonempty
index set $K$, let $\hc(K)$ be the
$\lceil B(1-\alpha)\rceil$th order statistic of
\[
  \left\{
    \max_{k\in K}|Z_k^{*(b)}|:b=1,\ldots,B
  \right\},
\]
where $Z_k^{*(b)}$ is the response in multiplier
replicate $b$. Use the same draws and normalizing
factors at every step. The procedure starts from
$K=\{1,\ldots,\ell\}$, rejects every $k\in K$ with
$|Z_k|>\hc(K)$, removes the rejected indices from $K$,
and repeats until no further rejection occurs;
$\Rhat_{\mathrm{sd}}$ is the set of all rejected indices. For $I\subseteq K$,
the bootstrap maximum over $I$ is no larger than
that over $K$ in every replicate. Hence
\[
  \hc(I)\le\hc(K).
\]

If a false rejection occurs, consider the first
step at which this happens. Its remaining index
set $K$ contains $H_0$. Some $k\in H_0$ then
satisfies
\[
  |Z_k|>\hc(K)\ge\hc(H_0).
\]
Consequently,
\[
  \Prob\{\Rhat_{\mathrm{sd}}\cap H_0\ne\emptyset\}
  \le
  \Prob\left\{
    \max_{k\in H_0}|Z_k|>\hc(H_0)
  \right\}
  \le\alpha+\varepsilon_{N,\ell}.
\]
The last inequality follows from the fixed-subset
conclusion of Theorem~\ref{thm:band} or
Proposition~\ref{prop:studentized-band}, as applicable.
The set $H_0$ is fixed conditional on the pilot.
For estimated normalization, choose $\varepsilon_N$
to dominate the vanishing approximation errors
for both the full family and $H_0$.

Critical values cannot increase as coordinates
are removed. Therefore every coordinate rejected
by $\Rhat$ is also rejected by the step-down
procedure, giving
$\Rhat\subseteq\Rhat_{\mathrm{sd}}$.

Finally, for any data-dependent subset
$\widehat I\subseteq\{1,\ldots,\ell\}$,
\[
  \mathcal E\subseteq
  \{\delta_k\in\mathcal C_k
    \text{ for all }k\in\widehat I\}.
\]
This proves the selected-reporting claim,
provided the original intervals are retained.
\end{proof}

\subsection{Detection: proof of Proposition~\ref{prop:detect}}
\label{app:detect-proof}

\begin{proof}
Condition on the independent pilot, if one is used.
Write $\Prob^*$ and $\E^*$ for probability and
expectation over the multiplier draws conditional
on the inference observations. Put
\[
  a_\alpha=\sqrt{2\log(4\ell/\alpha)},\qquad
  \bar c=\hat\rho a_\alpha,\qquad
  T^*=\max_k|Z_k^*|.
\]
Each $Z_k^*$ is conditionally centered Gaussian
with variance at most $\hat\rho^2$.
If $\hat\rho>0$, the Gaussian tail bound and
a union bound give
\[
  \Prob^*(T^*>\bar c)
  \le
  2\ell\exp\left\{
    -\frac{\bar c^2}{2\hat\rho^2}
  \right\}
  =\frac{\alpha}{2}.
\]
If $\hat\rho=0$, all multiplier responses are zero,
so this probability is zero.

Let $K_B$ count the bootstrap maxima exceeding
$\bar c$. Conditional on the observations,
\[
  K_B\sim\operatorname{Bin}(B,p_*),\qquad
  p_*=\Prob^*(T^*>\bar c)\le\alpha/2.
\]
Write $j=\lceil B(1-\alpha)\rceil$.
Since $\hc$ is the $j$th order statistic,
\[
  \{\hc>\bar c\}
  =
  \{K_B\ge B-j+1\}
  =
  \{K_B\ge\lfloor\alpha B\rfloor+1\}
  \subseteq\{K_B>\alpha B\}.
\]
Also,
\[
  \E^*[2^{K_B}]
  =(1+p_*)^B
  \le e^{\alpha B/2}.
\]
Markov's inequality therefore yields
\[
  \Prob^*(\hc>\bar c)
  \le
  2^{-\alpha B}\E^*[2^{K_B}]
  \le
  e^{-\alpha B(\log2-1/2)}
  \le e^{-\alpha B/8}.
\]
Taking expectations gives the same unconditional
bound, conditional on the pilot.

Define
\[
  \mathcal E=
  \{\delta_k\in\mathcal C_k\text{ for all }k\},
  \qquad
  \mathcal G=
  \left\{
    \max_k\frac{\homega_k}{\omega_k}\le\kappa_0,
    \ \hat\rho\le\rho_0
  \right\}.
\]
The coverage result and the assumed normalization
bound give
\[
  \Prob(\mathcal E^c)
  \le\alpha+\varepsilon_{N,\ell},
  \qquad
  \Prob(\mathcal G^c)\le\gamma_N.
\]

On $\mathcal G$, every coordinate satisfying the
stated signal condition obeys
\[
  \frac{\sqrt{\Neff}|\delta_k|}{\homega_k}
  >
  2\rho_0a_\alpha
  \ge2\bar c.
\]
On $\mathcal E$, the band gives
\[
  \frac{\sqrt{\Neff}|\hdelta_k-\delta_k|}
       {\homega_k}
  \le\hc.
\]
Consequently, on
$\mathcal E\cap\mathcal G\cap\{\hc\le\bar c\}$,
\[
  |Z_k|
  \ge
  \frac{\sqrt{\Neff}|\delta_k|}{\homega_k}-\hc
  >
  2\bar c-\hc
  \ge\hc.
\]
Thus every such coordinate belongs to $\Rhat$
on the same event. A union bound over the three
complementary events proves the result.
\end{proof}

\subsection{The pilot comparison: proof of Corollary~\ref{cor:pilot}}
\label{app:pilot-proof}

\begin{proof}
Let $T_P,T_Q$ be the total numbers of available observations
in the two groups. The effective sample sizes are
\[
  \Neff^{\mathrm{full}}
  =\frac{T_PT_Q}{T_P+T_Q},\qquad
  \Neff^{\mathrm{pilot}}
  =\frac{(1-f)T_P(1-f)T_Q}{(1-f)(T_P+T_Q)}
  =(1-f)\Neff^{\mathrm{full}}.
\]
Here the proportional split is understood to give
integer sample sizes.

For each retained coordinate, the response function
and its population variances are unchanged.
The split also preserves the group proportions.
Hence $\omega_k$ is the same for both procedures.

Let $t_{\mathrm{full},k}$ and $t_{\mathrm{pilot},k}$
denote the sufficient detection thresholds in
Proposition~\ref{prop:detect}. With a common product
$\kappa_0\rho_0$,
\[
  \frac{t_{\mathrm{pilot},k}}{t_{\mathrm{full},k}}
  =
  \sqrt{
    \frac{\log(4\ell_h/\alpha)}
    {(1-f)\log(4\ell/\alpha)}
  }.
\]
The claimed comparison follows by determining when
this ratio exceeds one.

The comparison holds for each realized selected
subfamily and its retained coordinates. It orders
the sufficient thresholds, but does not order
the actual detection probabilities.
The error terms in the two detection guarantees
may also differ.
\end{proof}

\subsection{Local detectability: proof of Theorem~\ref{thm:local}}
\label{app:proof-local}
\begin{proof}
Fix a retained coordinate $k$ satisfying
\[
  \dB(p_k,a)+\epsilon<r_k
  \le\dB(p_k,b)-\epsilon.
\]
Conditional on the added feature being present
under $P$, the triangle inequality gives
\[
  \dB(p_k,A')
  \le \dB(p_k,a)+\dB(a,A')
  \le \dB(p_k,a)+\epsilon<r_k
\]
almost surely. Hence
\[
  \varphi_k(A')
  \ge w_k\{r_k-\dB(p_k,a)-\epsilon\}>0.
\]
Under $Q$, conditional on presence, the reverse
triangle inequality gives
\[
  \dB(p_k,B')
  \ge \dB(p_k,b)-\dB(b,B')
  \ge \dB(p_k,b)-\epsilon
  \ge r_k.
\]
Thus $\varphi_k(B')=0$ almost surely.

The background mean measures are equal, so
\[
  \E_P\Phi_k(X_0)
  =\int\varphi_k\,d(\E_P N_{X_0})
  =\int\varphi_k\,d(\E_Q N_{Y_0})
  =\E_Q\Phi_k(Y_0).
\]
All expectations are finite because the responses
are bounded. By additivity and the common presence
probability $q$,
\begin{align*}
  \delta_k
  &=
  q\,\E_P\{\varphi_k(A')\mid F_P\ne\varnothing\}\\
  &\quad-
  q\,\E_Q\{\varphi_k(B')\mid F_Q\ne\varnothing\}\\
  &\ge qw_k\{r_k-\dB(p_k,a)-\epsilon\}>0.
\end{align*}
No independence between the background and the
added feature is needed.

For the second claim, suppose $\dB(p_k,a)\le\eta$.
The reverse triangle inequality yields
\[
  \dB(p_k,b)
  \ge \dB(a,b)-\dB(p_k,a)
  \ge \dB(a,b)-\eta.
\]
Therefore,
$\eta+\epsilon<r_k\le\dB(a,b)-\eta-\epsilon$
implies the first radius condition. Applying the
preceding bound gives
\[
  \delta_k
  \ge qw_k\{r_k-\dB(p_k,a)-\epsilon\}
  \ge qw_k(r_k-\eta-\epsilon)>0.
\]

Under the conditions of
Proposition~\ref{prop:detect}, either lower bound,
when applicable and exceeding the corresponding
detection threshold, gives the stated detection
guarantee.
\end{proof}

\subsection{Spatial extent: proof of Proposition~\ref{prop:extent}}
\label{app:extent-proof}

\begin{proof}
Let $\zeta=\bar\mu_P-\bar\mu_Q$, a finite signed Borel measure,
and let $S=\operatorname{supp}|\zeta|$. The response $\varphi_k$
vanishes outside $\mathcal B_k^\circ$. If
$\mathcal B_k^\circ\cap S=\emptyset$, then
$|\zeta|(\mathcal B_k^\circ)=0$, so
$\delta_k=\int\varphi_k\,d\zeta=0$. Thus $\delta_k\ne0$
implies that some $s_k\in\mathcal B_k^\circ\cap S$ satisfies
\[
  \dB(p_k,S)\le\dB(p_k,s_k)<r_k.
\]
On the simultaneous coverage event, a rejected coordinate has
$0\notin\mathcal C_k$ and therefore $\delta_k\ne0$. For any
$x\in\mathcal B_k^\circ$,
\[
  \dB(x,S)\le\dB(x,p_k)+\dB(p_k,s_k)
  <2r_k\le2r_{\max}.
\]
The same inclusion holds for any subset of the rejected
neighborhoods reported afterwards. It does not imply that the
reported neighborhoods cover $S$.
\end{proof}

\subsection{Lifetime invariance: proof of Proposition~\ref{prop:invariance}}
\label{app:invariance-proof}

\begin{proof}
For $a=(b,d)$, the point $a+(t,t)=(b+t,d+t)$ has lifetime $d-b$.
A diagram and its image under $\tau$ therefore have identical
lifetime multisets. Consequently, the joint law of the lifetime
data under $(P,Q)$ equals its law under $(P,P)$. Every statistic
based only on these data, and hence every lifetime-based test,
has the same distribution under the two pairs of populations.

For the second claim, use the coupling $Y=\tau(X)$ with $X\sim P$.
Points left in place contribute equally to the two expected
responses. By assumption, the images of moved points contribute
zero under $Q$. Under $P$, the open response square meets the
support of the mean measure of the moved points. This gives
positive expected mass where $\varphi_k>0$, and hence a strictly
positive expected response. It follows that
$\delta_k=\E_P\Phi_k(X)-\E_Q\Phi_k(Y)>0$.
\end{proof}

\section{Additional technical results}
\label{app:technical}

We give a transport interpretation of the band, a lower bound
for identifying a changed coordinate, and the bias calculation
for empirical separation.

\subsection{Transport separation of mean counting measures}
\label{app:transport}

We use the padded transport construction of \citet{PaperIII}.
Throughout this subsection, the landmark family is fixed.
If it was learned from a pilot, all statements condition on
that pilot.

\paragraph{Padded counting measures.}
Write $\mathbb H_L=\{(b,d):0\le b<d\le L\}$ and introduce an absence state $\varnothing$. Equip
$\mathbb H_L^\varnothing=\mathbb H_L\cup\{\varnothing\}$ with the ground metric
\begin{align}
  d_0(a,b) &=\min\{\|a-b\|_\infty,\max(\pi(a),\pi(b))\}, && a,b\in\mathbb H_L,
  \label{eq:transport-ground}\\
  d_0(a,\varnothing)&=d_0(\varnothing,a)=\pi(a), && a\in\mathbb H_L,\nonumber
\end{align}
and set $d_0(\varnothing,\varnothing)=0$. This is the bottleneck metric on singleton diagrams and
the empty diagram. For two off-diagonal points, its two branches represent direct matching and deletion of both points, respectively.

For $D\in\D{n}$, define
\begin{equation}\label{eq:padded-counting}
  \widetilde N_D =N_D+(n-|D|)\delta_{\varnothing},
\end{equation}
where $|D|$ counts multiplicities and $\delta_{\varnothing}$ is the unit point mass at the
absence state. Every padded counting measure has mass $n$. Its population mean is
\begin{equation}\label{eq:padded-population}
  \widetilde{\bar\mu}_P :=\E_P\widetilde N_X
  =\bar\mu_P+\{n-\bar\mu_P(\mathbb H_L)\}\delta_{\varnothing},
\end{equation}
and $\widetilde{\bar\mu}_Q$ is defined similarly. The restriction of $\widetilde{\bar\mu}_P$ to
$\mathbb H_L$ is the mean measure $\bar\mu_P$ from Section~\ref{sec:setup}.

For finite nonnegative measures $\eta,\zeta$ of the same positive mass, let $\Pi(\eta,\zeta)$ be their couplings. Define
\begin{equation}\label{eq:padded-transport-distance}
  \Winf(\eta,\zeta):=\inf_{\gamma\in\Pi(\eta,\zeta)} \|d_0\|_{L^\infty(\gamma)}.
\end{equation}
Here the norm is the essential supremum of the ground
distance under the coupling. Multiplying both measures
by the same positive constant leaves $\Winf$ unchanged.

\begin{lemma}[Padded transport representation]
\label{lem:padded-transport}
For every $D,E\in\D{n}$,
\[
  \dB(D,E)
  =\Winf(\widetilde N_D,\widetilde N_E).
\]
\end{lemma}

\begin{proof}
A partial diagram matching of cost at most $t$ extends
to a coupling of the padded measures. Retain the directly
matched pairs. Pair the remaining off-diagonal points
with one another where possible, then with absence
states, and pair the unused absence states.
Every pair has $d_0$-cost at most $t$. This proves
$\Winf(\widetilde N_D,\widetilde N_E)\le\dB(D,E)$.

Conversely, split repeated atoms and absence states into
$n$ labeled unit-mass slots on each side. A coupling
of cost at most $t$ gives a doubly stochastic matrix
supported on pairs with $d_0$-distance at most $t$.
A Birkhoff--von Neumann decomposition supplies a
permutation with the same support restriction.
For each paired off-diagonal pair, retain the direct
match if its $\ell_\infty$-distance is at most $t$.
Otherwise, the deletion branch of
\eqref{eq:transport-ground} allows both points to be
sent to the diagonal at cost at most $t$.
Pairs involving an absence state also correspond to
deletions. The resulting partial matching has
bottleneck cost at most $t$. Taking infima proves
the reverse inequality.
\end{proof}

\paragraph{Response bounds on the padded space.}
Extend each point response by $\varphi_k(\varnothing)=0$. Since
$r_k\le\pi(p_k)=d_0(p_k,\varnothing)$, the extension satisfies
\[
  \varphi_k(x)=w_k(r_k-d_0(p_k,x))_+,
  \qquad x\in\mathbb H_L^\varnothing.
\]
The triangle inequality and the Lipschitz property
of the positive-part function give
\begin{equation}\label{eq:padded-response-lipschitz}
  |\varphi_k(x)-\varphi_k(y)|
  \le w_k d_0(x,y).
\end{equation}
Padding therefore preserves both the diagram responses
and the population contrasts:
\begin{equation}\label{eq:padded-response-integrals}
  \Phi_k(D)=\int\varphi_k\,d\widetilde N_D,
  \qquad
  \delta_k=\int\varphi_k\,
  d(\widetilde{\bar\mu}_P-\widetilde{\bar\mu}_Q).
\end{equation}
The second identity follows by boundedness and
linearity of expectation.

\begin{corollary}[Transport separation of mean measures]
\label{cor:units}
Let
$\mathcal E=\{\delta_k\in\mathcal C_k
\text{ for all }k\}$ be the simultaneous coverage
event for the band \eqref{eq:band}. On $\mathcal E$,
\begin{equation}\label{eq:transport-certificate}
  \Winf(\widetilde{\bar\mu}_P,\widetilde{\bar\mu}_Q)
  \ge
  \max_{1\le k\le\ell}
  \frac{
    \bigl(
      |\hdelta_k|
      -\hc\,\homega_k/\sqrt{\Neff}
    \bigr)_+
  }{nw_k}.
\end{equation}
Consequently, the probability that this lower bound
holds is at least $\Prob(\mathcal E)$, with the
coverage guarantee given by Theorem~\ref{thm:band} or
Corollary~\ref{cor:estimated-normalization}, as applicable.
\end{corollary}

\begin{proof}
For any coupling
$\gamma\in\Pi(\widetilde{\bar\mu}_P,
              \widetilde{\bar\mu}_Q)$,
the marginal identities and
\eqref{eq:padded-response-lipschitz} imply
\begin{align*}
  |\delta_k|
  &=
  \left|
    \int\{\varphi_k(x)-\varphi_k(y)\}\,d\gamma(x,y)
  \right|\\
  &\le w_k\int d_0(x,y)\,d\gamma(x,y)\\
  &\le nw_k\|d_0\|_{L^\infty(\gamma)}.
\end{align*}
The last step uses the coupling's total mass $n$.
Taking the infimum over couplings gives
\[
  |\delta_k|
  \le nw_k
  \Winf(\widetilde{\bar\mu}_P,\widetilde{\bar\mu}_Q).
\]
On $\mathcal E$, the interval for coordinate $k$
also gives
\[
  |\delta_k|
  \ge
  \bigl(
    |\hdelta_k|-\hc\,\homega_k/\sqrt{\Neff}
  \bigr)_+.
\]
Combine these inequalities and maximize over $k$.
\end{proof}

This certificate concerns separation of the padded
mean counting measures. It uses the upper Lipschitz
bound and requires no lower-distortion assumption.
A positive certificate establishes separation in
the transport metric. A zero certificate does not
establish equality of the mean measures or the
diagram laws.

\subsection{A lower bound for identifying one changed coordinate}
\label{app:lower-bound}
\begin{proposition}[Lower bound for localization]\label{prop:lower}
Let $\nu$ have $\ell\le n$ landmarks with pairwise disjoint balls and
centers $p_k$ in the frame above the diagonal, and let $q\in(0,\tfrac12]$
and $0<h\le\min\{q,(1-q)/2\}$. For $j\in\{1,\dots,\ell\}$ let $P_j$ be the
law of the diagram that contains the center $p_k$ independently with
probability $q$ for $k\neq j$ and with probability $q+h$ for $k=j$, and
let $Q=P_0$ be the law with probability $q$ at every center. Then
the contrast vector, defined by
$\delta^{(j)}_k=\E_{P_j}\Phi_k(X)-\E_Q\Phi_k(Y)$,
is supported on coordinate $j$, with
$h/\sqrt{2q(1-q)}\le\delta^{(j)}_j/\omega_j\le\sqrt2\,h/\sqrt{q(1-q)}$
for every $m,m'$, and any estimator $\hat\jmath$ of the localized
coordinate from $m$ draws of $P_j$ and $m'$ draws of $Q$ satisfies, for
$\ell\ge8$,
\[
  \max_{j\le\ell}\Prob_j\{\hat\jmath\neq j\}\ge\tfrac12
  \quad\text{whenever}\quad
  \frac{\delta^{(j)}_j}{\omega_j}\le\frac18\sqrt{\frac{\log\ell}{m}} .
\]
\end{proposition}

\begin{proof} A point at the center $p_k$ lies in the $k$-th ball only, with response
$a_k=w_kr_k$ by \eqref{eq:embedding}. Under $P_j$ the coordinates are
therefore $\Phi_k(X)=a_k\zeta_k$ with independent
$\zeta_k\sim\mathrm{Bern}(q)$ for $k\neq j$ and
$\zeta_j\sim\mathrm{Bern}(q+h)$; each diagram carries at most one point
per ball, so at most $\ell\le n$ points, and lies in $\D{n}$. Hence
$\delta^{(j)}_k=0$ for $k\neq j$, $\delta^{(j)}_j=h\,a_j$, and
$\omega_j^2=a_j^2\{\tfrac{m'}{N}(q+h)(1-q-h)+\tfrac mN q(1-q)\}$. For
$h\le\min\{q,(1-q)/2\}$, $(q+h)(1-q-h)\in[\tfrac12q(1-q),\,2q(1-q)]$, so
$\omega_j^2\in[\tfrac12,2]\,a_j^2q(1-q)$ for every $m,m'$, which gives the
stated bounds on $\delta^{(j)}_j/\omega_j$. The $Q$-sample has the same
law under every $j$ and can be discarded. For $j\neq j'$ the laws $P_j$
and $P_{j'}$ differ in coordinates $j$ and $j'$ only, so
\[
  \mathrm{KL}(P_j^{\otimes m}\,\|\,P_{j'}^{\otimes m})
  = m\bigl[\mathrm{KL}(\mathrm{Bern}(q+h)\|\mathrm{Bern}(q))
          +\mathrm{KL}(\mathrm{Bern}(q)\|\mathrm{Bern}(q+h))\bigr]
  \le \frac{2mh^2}{q(1-q-h)}\le\frac{4mh^2}{q(1-q)},
\]
by $\mathrm{KL}(\mathrm{Bern}(a)\|\mathrm{Bern}(b))\le(a-b)^2/(b(1-b))$.
Fano's inequality in the form of \citet[proof of Lemma~3]{Yu1997}, for
$\ell$ hypotheses with pairwise divergence at most $\beta$, gives
$\max_j\Prob_j\{\hat\jmath\neq j\}\ge1-(\beta+\log2)/\log\ell$, which is
at least $\tfrac12$ as soon as $\beta\le\tfrac12\log\ell-\log2$. For
$\ell\ge8$, $\log2\le\tfrac13\log\ell$, so $\beta\le\tfrac16\log\ell$
suffices, and with $\beta=4mh^2/(q(1-q))$ this holds when
$h/\sqrt{q(1-q)}\le\sqrt{\log\ell/(24m)}$. The hypothesis
$\delta^{(j)}_j/\omega_j\le\tfrac18\sqrt{\log\ell/m}$ and the lower bound
$\delta^{(j)}_j/\omega_j\ge h/\sqrt{2q(1-q)}$ give
$h/\sqrt{q(1-q)}\le\sqrt{\log\ell/(32m)}$, which is enough.

\end{proof}

This construction uses diagrams with at most $\ell\le n$ points
and a fixed family of disjoint neighborhoods, so it lies within
the model of Section~\ref{sec:setup}. The $Q$-sample carries no
information about the changed coordinate, which is why the
bound depends on $m$ alone. With balanced groups and occupancy
$q$ bounded away from zero, it has the same
$\sqrt{\log\ell/m}$ dependence as the sufficient threshold in
Proposition~\ref{prop:detect}, up to constants. The construction
is informative only when $h\le q$ can be of order
$\sqrt{q\log\ell/m}$, which requires
$q\gtrsim\log\ell/m$.

The result concerns identification of one changed coordinate
from the inference sample. It does not establish minimax
recovery of a spatial support with overlapping or learned
neighborhoods. If a pilot drawn from $P_j$ is available to the
estimator, its information must also enter the lower-bound
experiment. Conditioning an upper bound on the pilot does not
remove that information.

\subsection{Bias of the plug-in separation}
\label{app:bias}

Let $U_i$ and $V_j$ be independent samples of feature vectors
with finite second moments. Write $\mu_P=\E U_1$,
$\mu_Q=\E V_1$, $\delta=\mu_P-\mu_Q$, and let their covariance
matrices be $\Sigma_P$ and $\Sigma_Q$. Independence gives
\[
  \E\|\bar U\|^2
  =\frac1m\E\|U_1\|^2
   +\left(1-\frac1m\right)\|\mu_P\|^2
  =\|\mu_P\|^2+\frac{\tr\Sigma_P}{m},
\]
with the analogous identity for $\bar V$, and
$\E\langle\bar U,\bar V\rangle=\langle\mu_P,\mu_Q\rangle$.
Therefore
\[
  \E\|\bar U-\bar V\|^2
  =\|\delta\|^2
   +\frac{\tr\Sigma_P}{m}
   +\frac{\tr\Sigma_Q}{m'}.
\]
This is the bias of the V-statistic MMD estimator with the
linear kernel on $\Phi$ \citep{Gretton2012}. It explains why
maximizing empirical separation over configurations can reward
variance as well as a population difference.

For a fixed family $\nu$, the band instead uses the
coordinatewise estimates $\hdelta_k$, which are unbiased.
Fitting or screening a new family on the inference observations
requires additional justification. Selecting which original
intervals to display after simultaneous calibration preserves
their coverage guarantee; see Corollary~\ref{cor:fwer} and
Remark~\ref{rem:pilot}.

%% file: Appendix_Part_II_Empirical.tex
\clearpage
\section*{Part II: Implementation and empirical studies}

\section{Implementation details}
\label{app:implementation}

The diagram is the sampling unit throughout the proposed analysis.
This section describes calibration, pilot selection, and the
comparators. Appendix~\ref{app:budget} gives the fixed-budget
simulations, Appendix~\ref{app:mol} the molecular analyses, and
Appendix~\ref{app:results} the additional simulation studies.

\subsection{Calibration and resampling}
\label{app:calibration-implementation}

The Gaussian multiplier band uses the separately centered group
responses in Section~\ref{sec:method}. Each multiplier is shared
across all coordinates of a diagram. Normalizing factors are
held fixed across draws. Fixed and pilot-based factors are
covered by Theorem~\ref{thm:band}; inference-based factors are
covered under the conditions in Appendix~\ref{app:studentization}.
Step-down testing uses the same draws and factors at every step.

For label permutation, entire diagrams are reassigned between
groups while preserving group sizes. Label-dependent
normalization is recomputed for each allocation. This gives a
test under the exchangeable null $P=Q$. The weaker coordinate
mean null can hold with unequal group laws and does not, by
itself, justify permutation calibration. The numbers of
replications and resampling draws are stated with each study.

\subsection{Pilot selection and image comparators at a fixed budget}
\label{app:pilot-implementation}

The pilot procedures in Appendix~\ref{app:budget} reserve
$\lfloor T_P/3\rfloor$ and $\lfloor T_Q/3\rfloor$ observations
from the two groups. Each pilot group is split into fitting and
validation samples of approximately equal size. The adaptive
family uses $100$ class-aware farthest-point centers fitted on
the fitting sample. Coordinates active in fewer than $5\%$ of
fitting diagrams are removed. Its validation normalization uses
the planned inference allocation and the floor
$0.05\,nw_kr_k$:
\[
  \homega_k=\max\{\tilde\omega_k,\,0.05\,nw_kr_k\}.
\]
Each adaptive center uses the base radius $0.05$, so its
radii are $\min\{0.05h,\pi(p)\}$ for the same multipliers $h$
as F1.

For each candidate radius multiplier $h$, let
$\widetilde\Sigma_{P,h}$ and $\widetilde\Sigma_{Q,h}$ be the
within-group empirical covariance matrices of the retained
validation responses. With planned inference sizes $m,m'$ and
$N=m+m'$, set
\[
  \widetilde\Gamma_h
  =\frac{m'}{N}\widetilde\Sigma_{P,h}
   +\frac{m}{N}\widetilde\Sigma_{Q,h}.
\]
Conditional on the pilot, the unnormalized Gaussian multiplier
vector has law $G_h^*\sim\mathcal N(0,\widetilde\Gamma_h)$.
We estimate $\tilde c_\alpha(h)$ from $1000$ independent draws
of
\[
  \max_{k\in I_h}\frac{|G_{h,k}^*|}{\tilde\omega_{h,k}},
\]
where $\tilde\omega_{h,k}$ denotes the normalization factor
used for that candidate, including its applicable floor.
The validation observations estimate the covariances; the
planned inference allocation supplies their weights and
$\Neff$ in the score \eqref{eq:score}. Ties favor fewer
coordinates and then the smaller radius.

The fixed-budget permutation maxima use $1000$ permutations,
counting the observed labeling. The per-coordinate permutation
tests followed by Holm adjustment use $10{,}000$ permutations,
because Holm's smallest threshold, $\alpha/\ell$, is below
$1/1001$.

The default persistence images use constant point weights and
an untruncated Gaussian kernel with standard deviation $0.05$
in birth--persistence coordinates. This is one pixel for I1
on its $20\times20$ grid and two pixels for ML on its
$40\times40$ grid. The lower triangle of pixels is retained.
I1 applies the multiplier band. The ML implementation follows
\citet{MoonLazar2023}: filtering at the $80$th percentile of
pooled standard deviation, pooled $t$-tests, and
Benjamini--Hochberg adjustment. Its target is FDR.
The resolution variants in study S5 instead use a one-pixel
kernel at every resolution, so they change both pixel size and
smoothing. The older image comparator in
Appendix~\ref{sec:exp-baseline} uses different weights and a
truncated kernel, as specified there.

\subsection{The survival-based comparator}
\label{app:strand}

We implement the survival-based comparator from the description
in \citet{STRAND}. Each off-diagonal point contributes one
fully observed event time, its lifetime $d-b$. Lifetimes are
pooled within groups and compared by a log-rank statistic.
There is no censoring. At each distinct lifetime $t$, the risk
set contains the pooled points with lifetime at least $t$,
and the events are those with lifetime exactly $t$. We use
$(O_A-E_A)^2/V$, with the usual hypergeometric variance, and
compare two calibrations.

\paragraph{Pooled asymptotic calibration.}
The statistic is compared with a $\chi^2_1$ distribution,
treating pooled lifetimes as independent, as in
\citet{STRAND}. The effective sample size is then a feature
count. Dependence among features of the same diagram is not
accounted for. In the neuroimaging example of \citet{PaperIII},
this pooled $p$-value can be two orders of magnitude too small.

\paragraph{Diagram-level permutation calibration.}
We permute diagram labels and recompute the statistic, retaining
the diagram as the sampling unit. Permutation validity holds
under $P=Q$. In Appendix~\ref{sec:exp-gap}, lifetimes are
computed once per replication and only diagram labels are
permuted. Both calibrations depend solely on the lifetime
multisets, so Proposition~\ref{prop:invariance} applies to both.

\subsection{Selection and the inference sample}
\label{app:selection-implementation}

\begin{remark}[Family selection and selected reporting]
\label{rem:pilot}
Selecting a new family on the inference sample, for example by
retaining the largest $|Z_k|$ values, changes the distribution
of the retained statistics. Recalibrating only that family on
the same observations is not covered by the fixed-family
argument. In Appendix~\ref{sec:exp-arms}, this procedure raises
FWER to as much as five times its nominal level. Maximizing
empirical separation can also favor configurations with larger
variance, as the bias calculation in Appendix~\ref{app:bias}
shows.

An independent pilot separates this choice from inference, and
the cross-fitting construction of Section~\ref{sec:method} does so
with every observation used once for inference.
This is different from selecting intervals for display after
calibrating a band over the full family. Such selected
reporting retains the original simultaneous guarantee, as
shown in Corollary~\ref{cor:fwer}.
\end{remark}

The older cross-fitting experiment appears in
Appendix~\ref{app:arms}. The additional CV experiment in
Appendix~\ref{app:cv-screening} specifies the inner selection
folds, the outer inference samples, and the calibration errors
for its two bands.

\section{Fixed-budget studies}
\label{app:budget}

These studies support Section~\ref{sec:experiments}. Each
replication draws a fixed total number of diagrams from each
group, and all procedures use the same generated samples.
Observations reserved for a pilot are included in that budget.
The studies are S1 (calibration under strict and weak nulls),
S2 (lifetime-changing and lifetime-preserving displacements),
S3 (radius against the displaced-feature certificate),
S4 (two features at different scales, a prevalence change,
and continuous jitter) and S5 (sensitivity at $80$ per class).
The CV study at the end of this section uses a separate
generator and states its own allocation and calibration.

\subsection{Simulation design and evaluation}
\label{app:budget-design}

\paragraph{Main displacement designs.}
Background points are uniform on the triangle
$0\le b<d\le1$, generated by sorting two independent uniforms.
Each diagram has $20$ background points and a feature present
with probability $q=0.5$. Its template is $a=(0.30,0.50)$ under
$P$ and either $b=a+(0,\Delta)$ or $b=a+(\Delta,\Delta)$ under
$Q$. The first shift changes lifetime; the second preserves
it. When present, the feature has jitter $\epsilon u$, with
$\epsilon=0.02$ and $u$ uniform on $\{-1,0,1\}^2$.

The clustered weak null uses four boxes of half-width $0.025$
centered at $(0.15,0.40)$, $(0.35,0.65)$, $(0.60,0.80)$, and
$(0.15,0.80)$. Under $P$, each of ten points independently
chooses a box. Under $Q$, all ten choose the same box within
a diagram. Both groups also have ten uniform background
points and the same feature process at $a$. In the counts
weak null, $P$ has $20$ background points and $Q$ has $0$ or
$40$ with equal probability, again with a common feature
process. These designs preserve the mean counting measure
while changing response covariances.

\paragraph{Families and allocation.}
F0 uses the $105$-center grid
$\{(i/14,j/14):0\le i<j\le14\}$ with
$r_k=\min\{0.05,\pi(p_k)\}$. F1 uses the same centers at
$r_{k,h}=\min\{0.05h,\pi(p_k)\}$ for
$h\in\{1/2,1/\sqrt2,1,\sqrt2,2\}$. Identical coordinates
created by the radius cap are retained once. F0 and F1 use
all observations for inference, with
$\homega_k=nw_kr_k$.

F2 chooses a radius multiplier on a pilot. A1 fits class-aware
farthest-point centers and uses all five multipliers; A2 also
selects the multiplier. These procedures reserve approximately
one third of each group for the pilot. Fitting, screening,
validation normalization, and the score calculation are
specified in Appendix~\ref{app:pilot-implementation}. The
comparators are F1-perm, I1, and ML, with image and permutation
settings given there.

\paragraph{Population targets and null coordinates.}
The common background cancels in the mean contrast. For any
additive map in the discrete-jitter design,
\[
  \delta_k=\frac q9
  \sum_{u\in\{-1,0,1\}^2}
  \{\varphi_k(a+\epsilon u)-\varphi_k(b+\epsilon u)\}.
\]
A $200{,}000$-draw Monte Carlo check gave a largest absolute
error of $4.5\times10^{-5}$ for targets up to $0.035$.
A landmark coordinate is classified as null when
$|\delta_k|\le10^{-12}r_k$, to allow for floating-point
cancellation. An image pixel is classified as null only when
$\delta_k=0$ exactly, because the Gaussian kernel is
untruncated. Under alternatives, its nonzero tails can leave
few or no exactly null pixels; FWER over those pixels must be
interpreted with that null set in mind.

For continuous jitter, $u$ is uniform on $[-1,1]^2$.
Targets use a $201^2$-node midpoint rule, agreeing with
$401^2$ nodes to $10^{-7}$. Null landmark coordinates are those
whose response squares avoid every changed support box.

For landmark methods, a change is located when at least one
rejected coordinate has nonzero population contrast, radius
at most $0.10$, and a center within $0.05$ of the changed
support, in the one-point bottleneck distance $\dB$ to the
support points $a+\epsilon u$ and $b+\epsilon u$ of the
displaced feature. For image methods, each pixel is treated as
a coordinate whose center is the pixel midpoint and whose
radius is half the pixel width, and the same rule applies.
Simultaneous coverage refers to all coordinates of
the calibrated band; FWER concerns rejection of at least one
null coordinate.

\subsection{Calibration and normalization}
\label{app:budget-calibration}

Coverage of every band lies between $0.94$ and $0.98$ in
studies S1 and S4. S1 coverage is reported below; the main-text
table reports the corresponding null rejection rates.
With no feature present, every procedure holds its nominal
level in the reported simulations.

\begin{table}[!htbp]\centering\small
\caption{S1: simultaneous coverage of the bands.}
\IfFileExists{tables/neigh_s1_cover.tex}{\input{tables/neigh_s1_cover.tex}}{}
\end{table}

\paragraph{Estimated normalization.}
The normalization study reruns F0 and F1 on the diagrams from
S1 and from S2 at $40$ per class, paired by replication.
It replaces the fixed factors by the groupwise estimate
\eqref{eq:groupwise-normalization}, the pooled estimate
\eqref{eq:pooled-normalization}, or one of two hybrids.
The first two use $\eta_N=10^{-6}nw_kr_k$; the hybrids floor
the groupwise SD at $\kappa nw_kr_k$ for
$\kappa\in\{0.01,0.005\}$. Each cell uses $1000$ replications.

A further design introduces substantial variance differences
across coordinates. Under its null, the two groups have the
same law: $20$ uniform background points; an additional $20$
points in a box of half-width $0.05$ centered at $(0.60,0.85)$
with probability $1/2$; and a feature always present at the
template ($q=1$). The unequal variances here are across
coordinates, rather than between groups under the null.
Under fixed normalization, the box coordinates dominate the
maximum and its critical value. The associated change at the
template is never located (Table~\ref{tab:norm-located}). The hybrid floors do not repair
the groupwise estimate under imbalance
(Table~\ref{tab:norm-fwer}).

\begin{table}[!htbp]\centering\small
\caption{Estimated normalization: FWER at the null, $1000$ replications
per cell. Fixed: $nw_kr_k$; group: groupwise SD; pooled: pooled SD;
hyb.: groupwise SD floored at $1\%$ or $0.5\%$ of $nw_kr_k$. Bold: the
$95\%$ Wilson interval lies above $0.05$.}
\label{tab:norm-fwer}
\IfFileExists{tables/neigh_norm_fwer.tex}{\resizebox{\textwidth}{!}{\input{tables/neigh_norm_fwer.tex}}}{}
\end{table}

\begin{table}[!htbp]\centering\small
\caption{Estimated normalization: probability of locating the change,
$T$ diagrams per class, columns as in Table~\ref{tab:norm-fwer}.}
\label{tab:norm-located}
\IfFileExists{tables/neigh_norm_located.tex}{\resizebox{\textwidth}{!}{\input{tables/neigh_norm_located.tex}}}{}
\end{table}

\subsection{Localization and geometric checks}
\label{app:budget-localization}

The following tables give localization and false-rejection
rates omitted from the main text. Under alternatives, FWER
over exactly null coordinates is at most $0.044$ for the bands
in S2 and $0.048$ in S3--S5. In the corresponding FWER
comparisons, ML stays at or below $0.04$ except for the
one-pixel kernel on the $40\times40$ image in S5, where its
rate is $0.158$. Its nominal target is FDR, so these values
are not comparisons of methods with identical error targets.

\begin{table}[!htbp]\centering\small
\caption{S2: probability of locating the change.}
\IfFileExists{tables/neigh_s2.tex}{\resizebox{\textwidth}{!}{\input{tables/neigh_s2.tex}}}{}
\end{table}

\begin{table}[!htbp]\centering\small
\caption{S2: FWER over exact nulls under alternatives.}
\IfFileExists{tables/neigh_s2_fwer.tex}{\resizebox{\textwidth}{!}{\input{tables/neigh_s2_fwer.tex}}}{}
\end{table}

\begin{table}[!htbp]\centering\small
\caption{S3 and S4: probability of locating the change (for two features,
any component; per-component rates are in the replication files).}
\IfFileExists{tables/neigh_s3.tex}{\resizebox{\textwidth}{!}{\input{tables/neigh_s3.tex}}}{}

\medskip
\IfFileExists{tables/neigh_s4.tex}{\resizebox{\textwidth}{!}{\input{tables/neigh_s4.tex}}}{}
\end{table}

\paragraph{Spatial extent and sufficient radii.}
We check grids of $55$, $105$, and $210$ centers, both shift
types, $\Delta\in\{0.05,0.10,0.15\}$, and
$\epsilon\in\{0.01,0.02,0.04\}$. Each of the $2395$
coordinates with nonzero exact contrast has its center within
its radius of the changed support. Every sampled point of
its square lies within twice its radius, consistent with
Proposition~\ref{prop:extent}.

For the $210$-center grid, we also examine $60$ radii in
$[0.002,0.15]$, lifetime-changing shifts from $a$, the same
$\Delta$ and $\epsilon$ values, and $q\in\{0.25,0.5,1\}$.
In all $1332$ certified center--radius cases, the lower bound
in Theorem~\ref{thm:local} is no larger than the exact contrast.
These are numerical checks of the geometric statements, not
additional coverage guarantees. The script
\texttt{experiments/check\_geometry.py} reproduces both checks
and the Monte Carlo comparison of population targets.

\subsection{Sensitivity analyses}
\label{app:budget-sensitivity}

Study S5 uses $80$ diagrams per class and $500$ replications
per cell. Its variants share generated diagrams, so results
across variants are paired. In the shared null cell, F1 has
FWER $0.068$, compared with $0.048$ over $1000$ independent
runs at the same sample size in S1. Moving the template off
the grid to $(0.31,0.51)$ lowers F0's localization to
$0.93$--$0.99$.

The per-coordinate permutation references have null rejection
rates of $0.036$--$0.044$ on the grid and $0.026$--$0.028$ on
the image. At $\Delta=0.10$, they locate the change in $97\%$
and $64$--$66\%$ of runs, respectively. The S5 tables include
the remaining sensitivity variants. Their image-resolution
comparisons use the one-pixel bandwidth described in
Appendix~\ref{app:pilot-implementation}.

\begin{table}[!htbp]\centering\small
\caption{S5 sensitivity at $80$ per class, $500$ replications: probability
of locating the change (top) and FWER over exact nulls (bottom), by
variant and $\Delta$.}
\IfFileExists{tables/neigh_s5.tex}{\resizebox{\textwidth}{!}{\input{tables/neigh_s5.tex}}}{}

\medskip
\IfFileExists{tables/neigh_s5_fwer.tex}{\resizebox{\textwidth}{!}{\input{tables/neigh_s5_fwer.tex}}}{}
\end{table}

\subsection{Cross-validated region and radius selection}
\label{app:cv-screening}

This experiment examines whether spatial selection can compensate for
the observations reserved for selection. We use $100$ prespecified
centers in $25$ disjoint regions, each containing a $2\times2$ block of
four centers. Every center has radii
$0.002$, $0.004$, $0.008$, $0.012$ and $0.018$, with weights $w_k=1/r_k$.
The radius cap holds throughout. Each diagram has independent
$\operatorname{Binomial}(8,0.35)$ multiplicity at each center and
$\operatorname{Binomial}(8,0.5)$ multiplicities at four horizontal
offsets, $0.003$, $0.006$, $0.010$ and $0.015$. Supports at different
centers are disjoint. Under the null, both groups have the same law.
Under the alternatives, the central probability in the second group
increases to $0.415$ at either one center or all four centers of one
region. The affected mean contrasts are $-0.52$ at every radius.
Larger radii add background variation, so the smallest radius has the
highest population signal-to-noise ratio. The four-center alternative
has four times the total mean change of the one-center alternative.

\paragraph{Selection score.}
Let $G_j$ contain the four centers of candidate region $j$, and let
$\Phi_{c,h}$ denote the response at center $c$ and radius $r_h$.
For selection only, use the regional average
$A_{j,h}(D)=\frac14\sum_{c\in G_j}\Phi_{c,h}(D)$.
The outer selection sample contains $m_S$ diagrams per group.
Partition these observations into three folds within each group, with
matching fold sizes $m_v$. All splits are independent of the observed
responses. Let $\widehat\Delta_{j,h}^{(-v)}$ and
$\widehat\Delta_{j,h}^{(v)}$ be the training and validation mean
contrasts of $A_{j,h}$, respectively. Compute the training pooled SD
$\widehat s_{j,h}^{(-v)}$ as the square root of the average of the two
unbiased group variances. Set $\widehat u_{j,h}^{(-v)}=1$ when the
training contrast is nonnegative and $-1$ otherwise. The score is
\[
  \widehat S(j,h)=\sum_{v=1}^{3}\frac{m_v}{m_S}
  \frac{\widehat u_{j,h}^{(-v)}\widehat\Delta_{j,h}^{(v)}}
       {\max\{\widehat s_{j,h}^{(-v)},10^{-8}\}}.
\]
Select a maximizing pair $(\widehat j,\widehat h)$, breaking ties in
a fixed candidate order. Every quantity in this score uses only the
outer selection sample. The final inference sample enters neither
training nor validation. Region/radius selection retains the four
individual contrasts indexed by $G_{\widehat j}$ at $r_{\widehat h}$;
the regional average itself is not tested. The radius-only comparison
retains every center at this same selected radius. It is an ablation
of spatial selection, not a separately optimized radius-only score.

\paragraph{Sample allocation and validity.}
We use $\alpha=0.05$. All methods receive $120$ diagrams per group. The pilot methods use
$40$ per group for the inner CV and the remaining $80$ for inference
at level $\alpha$. The cross-fitted methods split each group into
halves of $60$. One half supplies all three inner CV folds; the
other supplies inference at level $\alpha/2$. Swapping the halves
gives a second selected family and a second band. The selection SDs
are used only in the score; individual-coordinate SDs for calibration
are computed from the relevant inference sample.

Conditional on its own selection half, each inference family is fixed
and its inference observations remain independent of the selection
data. Let $E_f$ be failure of simultaneous coverage in outer direction
$f$. Under the conditions for the band, suppose, for deterministic
error bounds $\varepsilon_f$, that
$\Pr(E_f\mid\text{selection half }f)\le\alpha/2+\varepsilon_f$.
Taking expectations and applying a union bound gives
\[
  \Pr(E_1\cup E_2)\le\alpha+\varepsilon_1+\varepsilon_2.
\]
Thus the two families have joint coverage, and their union of
discoveries has family-wise error control, up to these calibration
errors. Independence of $E_1$ and $E_2$ is not required. We retain both
original interval families with their fold identities, including when
a coordinate is selected twice. We do not construct a pooled
full-sample interval after selection.

The full-sample comparison tests all centers at all radii jointly.
An oracle knows the changed region but tests all five radii there.

We use $1000$ paired replications and $999$ resampling draws.
The multiplier bands use the groupwise SD
$\{(s_{P,k}^2+s_{Q,k}^2)/2\}^{1/2}$, floored at $10^{-8}$.
The selected family is held fixed throughout inference resampling.
The script \texttt{localized\_cv\_simulation.py} in the supplementary
material writes every cell; Table~\ref{tab:cv-screening} is generated
from its output.

\begin{table}[t]
\centering
\small
\caption{Additional selection experiment with $25$ regions and $120$
diagrams per group. Entries are percentages from $1000$ replications.
Power is the probability that the multiplier band rejects at least one
nonzero contrast. Columns labeled $1$ and $4$ refer to the number of
changed centers. Coverage concerns all original intervals reported by a
method, including both cross-fitted families.}
\label{tab:cv-screening}
\IfFileExists{tables/cv_selection.tex}{\input{tables/cv_selection.tex}}{}
\end{table}

For four changed centers (Table~\ref{tab:cv-screening}), cross-fitted
region and radius selection gains $12.8$ percentage points over joint
testing of all radii, with paired $95\%$ Monte Carlo interval
$[9.9,15.7]$ points; pilot region and radius selection gains $1.0$
point, $[-2.3,4.3]$. The correct region is selected in $93.7\%$ of
the outer selection halves, compared with $11.7\%$ for an isolated
change, where every selection method loses power ($3.2$--$10.3\%$
against $25.9\%$). Thus the gain
depends on reliable regional selection and does not extend to every
spatially concentrated alternative. The regions are prespecified and
the score favors changes with the same sign. Opposing changes can
cancel in the selection score. At $240$ diagrams per group, joint
testing reaches $99.3\%$ for four changed centers and $68.9\%$ for
one, and the selection methods $89$--$100\%$ and $9$--$30\%$.
Multiplier coverage is below its
nominal level in several cells ($91.4$--$94.7\%$), so this example
supports a possible selection advantage for testing rather than
validated band calibration in this larger design. The Monte Carlo intervals quantify simulation
error within this design, not uncertainty about performance in other
settings.

\clearpage
\section{Molecular application}
\label{app:mol}

The primary analysis is the deduplicated reanalysis reported in
Section~\ref{sec:exp-mol}. The earlier analysis is retained at
the end of this section with its own preprocessing and
landmark settings. The two sets of numerical results should
be read under their respective specifications.

\subsection{Data and preprocessing}
\label{app:mol-data}

The supplementary manifest records TUDataset files for MUTAG,
DHFR, COX2, and PTC\_MR, together with checksums, class counts,
and graph sizes. Duplicates are exact graph isomorphisms
preserving node labels. MUTAG has $11$ duplicate groups
covering $24$ graphs; DHFR has $35$ covering $72$, including
one group with conflicting outcome labels. COX2 has $2$
covering $4$, and PTC has $15$ covering $31$, including three
with conflicting labels. The primary analysis retains one
graph from each duplicate group and drops conflicting groups.
Versions using all graphs are retained in the replication
files.

The heat-kernel signature uses the full eigendecomposition of
the normalized Laplacian. An edge takes the larger of its
endpoint values. GUDHI's extended persistence output combines
ordinary, relative, and extended pairs using absolute values;
no additional shift is needed. The primary heat-kernel time
is $t=10$ and the primary degree is one. These choices
continue the earlier analysis and were not specified
prospectively. Degree zero is secondary, and COX2 and PTC
provide supporting analyses.

\subsection{Families, splits, and supplementary results}
\label{app:mol-procedures}

The fixed family uses the grid of Section~\ref{sec:exp-budget}
on the fixed $[0,1]$ frame. The adaptive family uses $100$
farthest-point centers, their fitted radii as base radii, a
$5\%$ activity screen on the fitting half, and validation
scales with floor $0.05\,n_{\mathrm{fit}}r_k$. Primary maps
use the original stratified split and $B=5000$. The $20$
stability splits and the learning curves use $B=2000$
(Tables~\ref{tab:chem-learning} and~\ref{tab:mutag-structure}).

\begin{table}[!htbp]\centering\small
\caption{Primary analyses: localized coordinates and global $p$-values
($n_0/n_1$ diagrams per class used by F0 and F1; the pilot methods infer
on two thirds).}
\IfFileExists{tables/chem_primary.tex}{\resizebox{\textwidth}{!}{\input{tables/chem_primary.tex}}}{}
\end{table}

\begin{table}[!htbp]\centering\small
\caption{Stability over $20$ random pilot/inference splits: empty maps,
median signed mesh overlap between pairs of nonempty maps, mean fraction
of a split's localized neighborhoods matched in sign within their radius
by another split, median number localized.}
\IfFileExists{tables/chem_stability.tex}{\input{tables/chem_stability.tex}}{}
\end{table}

\begin{table}[!htbp]\centering\small
\caption{Learning curves: $50$ class-balanced subsamples per size,
degree one.}
\label{tab:chem-learning}
\IfFileExists{tables/chem_curves.tex}{\input{tables/chem_curves.tex}}{}
\end{table}

\begin{figure}[!htbp]
  \centering
  \IfFileExists{figures/map_dhfr_rev.pdf}{\includegraphics[width=0.62\textwidth]{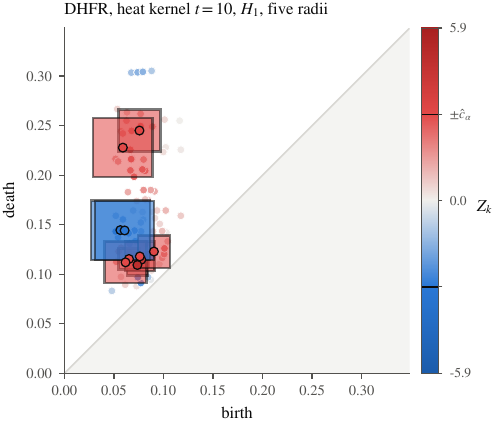}}{}
  \caption{DHFR, heat-kernel signature $t=10$, degree one, A1 on the
  primary split: $20$ localized coordinates at $10$ centers.}
  \label{fig:dhfr-rev}
\end{figure}

\subsection{Mapping localized regions to molecular structure}
\label{app:mol-structure}

Rings are the cycles of a minimum cycle basis. They are
matched to degree-one points by an optimal assignment using
the lowest and highest heat-kernel values over the ring.
Fused systems consist of rings sharing a bond.

Checks on toy graphs show that a single symmetric ring gives
no persistence point because its heat-kernel values coincide.
Two and three linearly fused rings give two and three matched
rings, respectively, together with branchings. Rings sharing
only one atom are not classified as fused. Among the $175$
deduplicated MUTAG graphs, all $486$ basis rings match a
point using the minimum cycle basis, compared with $452$
using NetworkX's cycle basis. Thus the structural reading
depends on the cycle basis as well as on the filtration.
The selected-region summaries are exploratory.

\begin{table}[!htbp]\centering\small
\caption{MUTAG: rings in the localized region of the primary maps
(exploratory).}
\label{tab:mutag-structure}
\IfFileExists{tables/mutag_structure.tex}{\input{tables/mutag_structure.tex}}{}
\end{table}

\subsection{Earlier molecular analysis}
\label{sec:exp-chem}

This analysis predates the deduplicated reanalysis above. We
retain its reported results and preprocessing separately.
MUTAG, COX2, DHFR, and PTC \citep{Morris2020} are binary
activity classifications of small molecules. Following
\citet{PALACE}, graphs were filtered using node labels,
Ollivier--Ricci curvature, or the heat-kernel signature.
Extended persistence diagrams in degrees zero and one were
shifted to be nonnegative, with relative and extended points
reflected above the diagonal. A stratified third of each
dataset supplied the pilot, which placed and screened up to
$K=200$ landmarks. The band used the remaining observations
and $B=5000$.

Stability was assessed over $20$ random splits. The reported
measures were the fraction of one split's localized landmarks
recovered within their own radius by another split and the
agreement of the top landmark (Table~\ref{tab:chem}).
Three datasets yielded localized contrasts, with substantially
greater stability for MUTAG and DHFR than for COX2.

Under the heat-kernel filtration, the degree-zero maps
localized $72$ of $187$ landmarks on MUTAG and $55$ of $192$
on DHFR. In degree one, the counts were $21$ of $145$ and
$30$ of $178$. Another split recovered $89$--$97\%$ of the
localized region, and the top landmark agreed in
$85$--$100\%$ of splits. The earlier MUTAG map contained one
cluster of short-lived cycles, born near $0.06$--$0.09$ and
dying near $0.09$--$0.12$. DHFR contained two clusters
(Figure~\ref{fig:chemmap}).

A degree-one point in an extended graph diagram can
represent an independent cycle or a branching of the
superlevel sets. For a ring, the relevant values are the
lowest and highest heat-kernel values over its atoms.
In this earlier MUTAG analysis, the localized region
contained $82\%$ of rings in fused systems of three or more
rings and $21\%$ of other rings. Mutagenic compounds placed
an average of $2.66$ rings and $2.76$ branchings in the
region, compared with $0.35$ and $1.03$ in the other class.

Among compounds with the same largest ring system, the
classes placed similar mass in the region. In this series,
$91$--$100\%$ of compounds with three or more fused rings
were mutagenic, compared with $29$--$31\%$ of the remaining
compounds. These patterns are consistent with the fused-ring
indicator of \citet{Debnath1991}. Such ring systems also occur
in larger molecules. An exploratory logistic regression
retained an association between mutagenicity and the number
of rings in the region after adjustment for atom count.
Among compounds with $15$--$22$ atoms, mutagenic compounds
placed $1.8$--$2.8$ rings there, compared with $0.4$--$0.7$.
The region was selected using these data, and the pilot
helped construct it, so these summaries do not provide an
independent test of a chemical explanation.

COX2 localized weakly and unstably: another split recovered
$28$--$42\%$ of the region. PTC, for which \citet{PALACE}
also reports weak topological features, gave one marginal
landmark with $p=0.044$ that no other split reproduced.
This provides no stable localized finding, without
establishing equality of the populations. The node-label
filtration on MUTAG gave a strong global rejection, but its
integer-valued coordinates concentrated the landmarks at
three positions and provided little spatial resolution.
The transport lower bound in Corollary~\ref{cor:units} was
positive whenever a map localized a contrast, but was about
$10^{-4}$ of the frame and numerically small.

\begin{table}[!htbp]
\centering\small
\caption{Earlier molecular analysis under a pilot/inference
split. The table reports the global $p$-value, sizes of
the three localization sets, and stability over $20$
random splits: the fraction of one split's localized
region recovered by another and agreement of the top
landmark. These results use the earlier preprocessing
and family construction in Appendix~\ref{sec:exp-chem}.}
\label{tab:chem}
\IfFileExists{tables/chemgraph.tex}{\resizebox{\textwidth}{!}{\input{tables/chemgraph.tex}}}{\todo{tab\_chemgraph.tex}}
\end{table}

\begin{figure}[!htbp]
  \centering
  \IfFileExists{figures/map_chem2.pdf}{\includegraphics[width=0.62\textwidth]{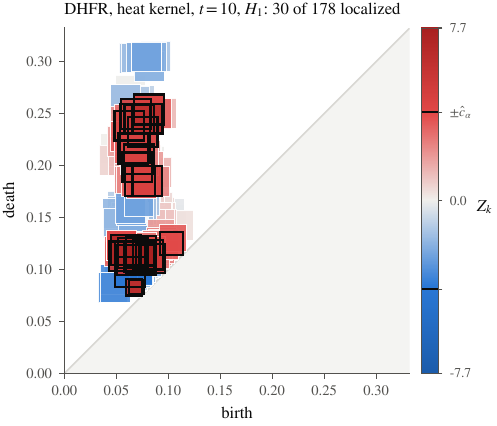}}{\todo{\texttt{\detokenize{figures/map_chem2.pdf}}}}
  \caption{DHFR under the heat-kernel filtration at $t=10$, degree one:
  $30$ of $178$ landmarks localized, in two clusters of cycles born near
  $0.06$--$0.10$, dying near $0.08$--$0.13$ and $0.18$--$0.25$; another
  random pilot/inference split recovers, on average, $94\%$ of the
  region.}
  \label{fig:chemmap}
\end{figure}

\clearpage
\section{Additional simulation studies}
\label{app:results}

These experiments preceded the fixed-budget studies in
Appendix~\ref{app:budget}. We retain their designs and
reported numerical results. Their pilot allocation, jitter,
normalization, and null-coordinate definitions differ from
those used in the main experiments.

\subsection{Design and evaluation}
\label{app:earlier-design}

Except for the cross-fitting arm described below, the pilot
contains $200$ diagrams per class in addition to the $m$
inference diagrams. Thus $m=20$ corresponds to $220$ diagrams
per class in total. Responses are normalized by pilot
variances,
\[
  \homega_k^2
  =\frac{m'}{N}\tilde s_{P,k}^2
   +\frac{m}{N}\tilde s_{Q,k}^2,
\]
instead of deterministic scales. After the radius cap, the
pilot screen removes no landmarks from the two coarser grids,
at most $2$ of $210$, $17$--$22$ of the finest grid's $406$,
and at most $6$ of $400$ pilot-fitted landmarks.

Each diagram has $20$ uniform background points above the
diagonal in $[0,1]^2$ and one feature with Gaussian jitter of
scale $0.02$. Group $A$ has its feature at $(0.30,0.50)$.
In the persistence regime, group $B$ moves it to
$(0.30,0.50+\Delta)$. In the location regime, it moves to
$(0.30+\Delta,0.50+\Delta)$, preserving lifetime.
The calibration and localization studies use the persistence
regime unless stated otherwise.

At $\Delta=0$, every coordinate is null. For $\Delta>0$,
a coordinate is counted as null when its ball, enlarged by
four jitter scales, cannot contain either signal location.
With Gaussian jitter this is a geometric proxy, not an exact
test of $\delta_k=0$. The population contrasts used to
evaluate coverage are estimated from a reference sample of
$20{,}000$ diagrams per class under the same family $\nu$.
Alternative-case error rates below use this geometric
classification. They are distinct from the exact-null error
rates in Appendix~\ref{app:budget}.

The comparisons involving estimated normalization are
finite-sample findings for these designs. Pilot normalization
is not a general requirement or a universal remedy. In
particular, with one coordinate a positive normalization
factor cancels from the interval. The asymptotic conditions
for inference-based normalization are given in
Appendix~\ref{app:studentization}.

\subsection{Level, coverage, and the price of simultaneity}
\label{sec:exp-level}

At $\Delta=0$, FWER is one minus simultaneous coverage.
We use $500$ replications and $B=2000$, with grids of $55$,
$105$, $210$, and $406$ landmarks; between $50$ and $400$
pilot-fitted landmarks; and $m\in\{20,40,60\}$
(Table~\ref{tab:level}, Figure~\ref{fig:critical}).

For the grid, the reported FWER ranges from $0.054$ at
$\ell=55$ to $0.006$ at the finest resolution. For
pilot-fitted landmarks, it ranges from $0.008$ to $0.048$
across the reported sample sizes, including $m=20$ and
$m=60$. The normal-reference Bonferroni procedure has rates
of $0.07$--$0.11$ for $\ell\le105$, $0.18$--$0.25$ near
$\ell=210$, and $0.33$--$0.48$ at the finest resolution.
These failures concern its marginal normal calibration;
Bonferroni adjustment itself controls FWER when applied to
valid marginal tests.

At fine resolution, many responses are sparse, and their
variances relative to the estimated pilot scales vary
substantially. The bootstrap critical value reaches
$5$--$6$, compared with about $3.8$ for the unit-normal
reference. This wider band reduces power on the fine grid
at $m=20$ (Appendix~\ref{sec:exp-localize}). Pilot-fitted
landmarks concentrate on occupied regions and keep the
critical value within $0.7$ of the Bonferroni reference.
The latter still gives FWER $0.054$--$0.088$ in those cells.
This supports pilot-fitted placement under this design with a separate
pilot; the fixed-budget comparisons in
Appendix~\ref{app:budget} address the cost of obtaining the
pilot observations.

\begin{table}[!htbp]
\centering\small
\caption{Level and coverage at $\Delta=0$ in the earlier
additional-pilot studies, for $m\in\{20,60\}$. All sample
sizes appear in Appendix~\ref{app:tables}. Every coordinate
is null, so FWER of $\Rhat$ is one minus coverage, and the
step-down set makes any rejection exactly when $\Rhat$
does. The table compares FWER of $\Rhat$ and
normal-reference Bonferroni, and the bootstrap critical
value $\hc$ with $z_{1-\alpha/(2\ell)}$.}
\label{tab:level}
\IfFileExists{tables/level.tex}{{\input{tables/level.tex}}}{\todo{tab\_level.tex}}
\end{table}

\begin{figure}[!htbp]
  \centering
  \IfFileExists{figures/critical.pdf}{\includegraphics[width=0.62\textwidth]{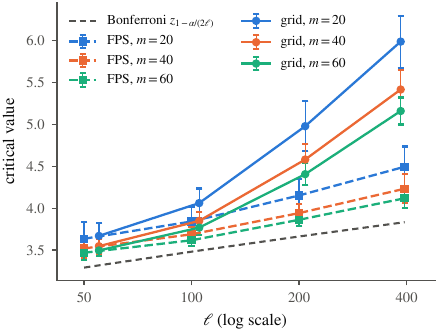}}{\todo{\texttt{\detokenize{figures/critical.pdf}}; generated by \texttt{\detokenize{the figure generator}}}}
  \caption{Bootstrap critical values against $\ell$ in the
  earlier additional-pilot design. For pilot-fitted landmarks,
  the values lie slightly above the normal-reference Bonferroni
  threshold $z_{1-\alpha/(2\ell)}$, which grows at order
  $\sqrt{\log\ell}$. Fine grids give larger values because
  sparse responses yield greater variation in coordinate
  variances relative to the estimated pilot factors
  (Proposition~\ref{prop:detect}). The normal-reference
  Bonferroni procedure overrejects in these comparisons
  (Table~\ref{tab:level}).}
  \label{fig:critical}
\end{figure}

\subsection{Finite-sample behavior of estimated normalization}
\label{sec:exp-scale}

Replacing the pilot scale by the groupwise estimate
$\{(m'/N)s_{P,k}^2+(m/N)s_{Q,k}^2\}^{1/2}$ gives
FWER $0.04$--$0.06$ at $\ell=105$ in balanced designs.
It can overreject substantially under imbalance. With
$20$ versus $60$ inference diagrams and $\Delta=0$,
$2000$ replications give FWER $0.51$ for the sample-normalized
band and $0.45$ for the normal-reference
Bonferroni procedure. The permutation maximum gives $0.064$,
with rates $0.036$--$0.064$ over this experiment's twelve
null cells.

Figure~\ref{fig:mechanism} shows a heavy lower tail in the
null $Z_k$ values. Third-moment-matching wild multipliers
increase FWER to $0.69$, while Efron's bootstrap is
conservative, with critical value $6.7$ and FWER $0.002$.
Matching skewness alone therefore does not resolve the
calibration problem in this experiment. A relevant
mechanism is the dependence between the smaller group's
sample mean and variance on count-like coordinates. When
that group observes little local mass, both estimates can
be small, inflating the normalized difference in those
samples.

With the independent pilot scale, the statistic is linear
in the inference observations. The reported FWER is $0.023$,
coverage is $0.98$, and power is unchanged in this design.
Normal-reference Bonferroni remains liberal at $0.12$,
because the pilot factors only approximate the true standard
deviations. These results do not contradict the asymptotic
estimated-normalization result in
Appendix~\ref{app:studentization}; they show that its
approximation can be poor at these sample sizes.

\begin{figure*}[!htbp]
  \centering
  \IfFileExists{figures/mechanism.pdf}{\includegraphics[width=0.9\textwidth]{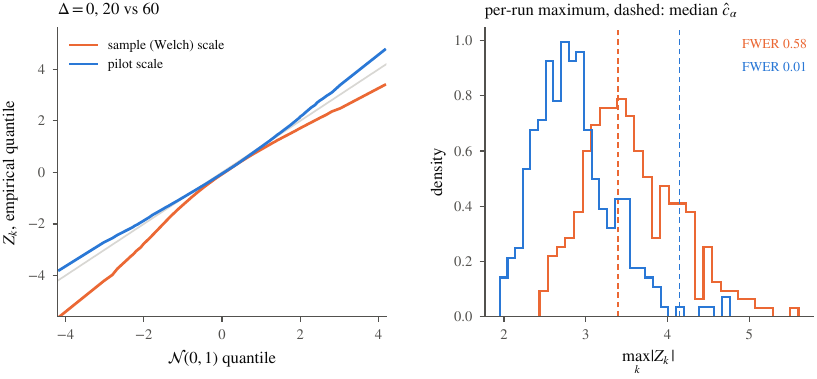}}{\todo{\texttt{\detokenize{figures/mechanism.pdf}}}}
  \caption{Estimated normalization under imbalance.
  At $\Delta=0$, all coordinates are null. The design uses $20$
  versus $60$ inference diagrams, $\ell=105$, and $300$ runs.
  Left: empirical quantiles of $Z_k$ against $\mathcal N(0,1)$
  using groupwise sample SDs or independent pilot factors.
  Right: each run's maximum and bootstrap critical value,
  with FWER calculated over these runs.}
  \label{fig:mechanism}
\end{figure*}

\subsection{Selection on the inference sample}
\label{sec:exp-arms}

Selecting the $K$ largest $|Z_k|$ values from $4K$ candidate
landmarks on the inference sample, then recalibrating a band
on that same sample, gives FWER $0.13$--$0.24$ at
$\Delta=0$. These results span $K\in\{50,100,200\}$ and
$m\in\{20,40,60\}$ and reach almost five times the nominal
level. Making the same selection on an independent pilot
gives $0.016$--$0.036$ (Table~\ref{tab:arms}).

Class-aware farthest-point placement fitted on the inference
sample, also with inference-based scales, gives FWER
$0.026$--$0.068$. This placement seeks to cover occupied
regions rather than directly maximize the observed
contrast. The two findings distinguish the selection rules
in these experiments, but neither supplies a general
guarantee for a family fitted on the inference sample.
Independent selection or cross-fitting provides one route
to a justified analysis. Keeping a subset of an already
calibrated full-family band is also valid, as explained in
Remark~\ref{rem:pilot}.

\begin{table}[!htbp]
\centering\small
\caption{Placement, selection, and normalization in the
earlier studies. Comparisons include inference-sample versus
pilot placement and selection ($m=m'\in\{20,60\}$), sample
versus pilot normalization in balanced and $20$-versus-$60$
designs, and the weak null. Inference-sample arms have no
pilot. At $\Delta=0$, every target contrast is zero and FWER
is the probability of any rejection. Coverage of a random
family is defined, but a family fitted on the inference
sample is not covered by the fixed-family theorem. The
table reports the available coverage and FWER results for
the band, normal-reference Bonferroni, and, where run,
the permutation maximum, together with power at
$\Delta=0.20$.}
\label{tab:arms}
\IfFileExists{tables/arms.tex}{{\input{tables/arms.tex}}}{\todo{tab\_arms.tex}}
\end{table}

\subsection{Localization}
\label{sec:exp-localize}

With pilot-fitted landmarks and $K\ge100$, the band
localizes the signal in at least $99\%$ of runs for every
$\Delta\ge0.10$ when $m\ge40$, and in $70$--$98\%$ when
$m=20$ (Table~\ref{tab:localization}). The reported
false-discovery proportion of $\Rhat$ is below $0.015$;
for Benjamini--Hochberg it is at most $0.04$. The center
with the largest $|Z_k|$ is within $0.006$--$0.036$ of
the signal in bottleneck distance, a fraction of a
landmark radius. Figure~\ref{fig:map-synthetic} shows one run.
The coarse grid behaves similarly.

At the finest grid and $m=20$, the band's power falls to
$0.01$--$0.03$. Bonferroni reaches $0.64$--$0.70$, but
its null FWER is $0.48$; Benjamini--Hochberg has a
false-discovery proportion of $0.16$--$0.24$. These
comparisons illustrate the difficulty of calibrating sparse
responses at a fine resolution. Pilot-fitted placement
avoids this difficulty in the reported design. With only
$K=50$, however, the displaced feature can move beyond
the selected neighborhoods: power is $0.08$--$0.49$ at
$\Delta=0.30$. Increasing to $K=100$ closes that gap.

\begin{figure}[!htbp]
  \centering
  \IfFileExists{figures/map_synthetic.pdf}{\includegraphics[width=0.62\textwidth]{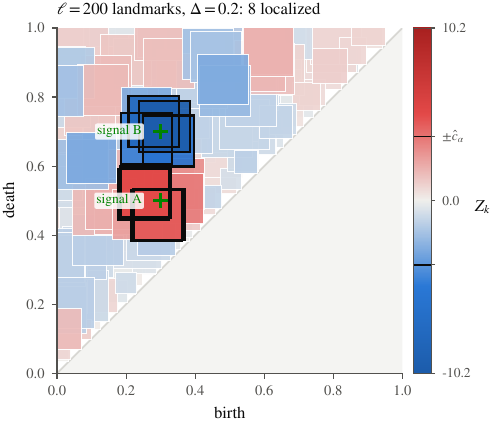}}{\todo{\texttt{\detokenize{figures/map_synthetic.pdf}}; generated by \texttt{\detokenize{the figure generator}}}}
  \caption{The band on the synthetic task, one run at $\Delta=0.20$,
  $m=60$: class $B$ carries the same features as class $A$ with one
  feature displaced from the location marked $A$ to the location marked
  $B$; the band localizes the displacement and nothing else.}
  \label{fig:map-synthetic}
\end{figure}

\begin{table}[!htbp]
\centering\small
\caption{Earlier localization study at
$\Delta\in\{0.10,0.30\}$ and $m\in\{20,60\}$, for the
coarsest and finest placements. Complete tables appear in
Appendix~\ref{app:tables}. Power: at least one signal
landmark localized. FDR: the reported false-discovery
summary using the geometric nulls defined in
Appendix~\ref{app:earlier-design}. Localization error:
bottleneck distance from the top landmark to the nearer
signal. The finest grid retains $384$--$389$ of its
$406$ landmarks after capping and pilot screening,
depending on $\Delta$.}
\label{tab:localization}
\IfFileExists{tables/localization.tex}{{\input{tables/localization.tex}}}{\todo{tab\_localization.tex}}
\end{table}

\subsection{Comparison with pixelwise image inference}
\label{sec:exp-baseline}

The image comparator applies one test per persistence-image
pixel and a Westfall--Young permutation maximum for
family-wise calibration under exchangeability. The image
tiles birth--persistence coordinates with $8$ to $21$
pixels per axis. Each point is weighted by persistence
and spread by a Gaussian kernel one pixel wide, truncated
at three standard deviations. This additive map also
admits the multiplier band. These settings differ from
the constant-weight, untruncated images in
Appendix~\ref{app:budget}.

Table~\ref{tab:baseline} compares both maps under both
calibrations on the same replications. With comparable
numbers of coordinates and $m=20$, pilot-fitted landmarks
localize more often. At $\Delta=0.10$, bootstrap power is
$70$--$94\%$ for $K\in\{100,200\}$, compared with
$15$--$23\%$ for $111$ or $189$ image pixels. At
$\Delta=0.20$, localization error is $0.02$--$0.04$ for
the landmarks and $0.03$--$0.07$ for the images. With
$m=60$, both maps approach full power once the image
contains at least $189$ pixels.

The calibration difference is smaller. At $K=100$ and
$m=20$, the permutation maximum rejects with probability
$0.82$, compared with $0.70$ for the bootstrap. Across
the null comparisons, FWER is $0.040$--$0.084$ for
permutation and $0.012$--$0.066$ for the bootstrap.
The bootstrap critical value reflects variation in
coordinate variances relative to their pilot scales:
$3.84$, compared with $3.59$ for permutation in the
reported comparison. Efron's bootstrap with fixed
normalization is more conservative still, with critical
value $3.92$.

Under alternatives, false-localization rates over the
geometric null coordinates are $0.002$--$0.022$ for
permutation and $0.002$--$0.018$ for the bootstrap.
Permutation thus provides a valid test under
exchangeability and is somewhat more powerful in these
comparisons. That validity does not extend automatically
to equal mean measures with different group laws.

The unequal-covariance designs in Appendix~\ref{app:arms}
make this distinction visible. When the smaller group is
more variable, the permutation maximum rejects in
$9$--$18\%$ of the clustered-design runs. Recomputing
sample normalization for every permutation gives
$18$--$33\%$, and $12$--$24\%$ even in balanced
clustered designs. The band ranges from $0.005$ to
$0.057$ across all sixteen cells. The multiplier method
also supplies simultaneous intervals for local mean
contrasts, in addition to a rejection set.

\begin{table}[!htbp]
\centering\small
\caption{The mass-univariate baseline: persistence-image pixels against
pilot-fitted landmarks, each calibrated by the bootstrap band and by the
permutation maximum on the same $500$ replications. FWER at $\Delta=0$
over all coordinates; power: a signal coordinate localized; loc.\ err.\
at $\Delta=0.20$.}
\label{tab:baseline}
\IfFileExists{tables/baseline.tex}{{\input{tables/baseline.tex}}}{\todo{tab\_baseline.tex}}
\end{table}

\subsection{Lifetime-changing and lifetime-preserving shifts}
\label{sec:exp-gap}

Figure~\ref{fig:power} and Table~\ref{tab:comparative}
compare six procedures at $m=30$ per class: the band's
global maximum test; MMD with an additive Gaussian kernel
on the same pilot-fitted embedding \citep{Gretton2012};
the log-rank statistic of \citet{STRAND} with diagram-label
permutation or pooled $\chi^2_1$ calibration; landscape
MMD \citep{Bubenik15}; and persistence-image MMD
\citep{Adams2017}. Permutation tests use $500$ shuffles.
The survival implementations are described in
Appendix~\ref{app:strand}. All six procedures have null
rejection rates between $0.02$ and $0.07$ in these
experiments.

In the persistence regime, the band reaches power $0.77$
at $\Delta=0.10$ and $0.87$ at $\Delta=0.20$.
Embedding MMD gives $0.76$ and $0.88$. Persistence-image
MMD reaches $0.33$ only at $\Delta=0.30$, and landscape
MMD stays near its null rejection rate. Both survival
tests remain at or below $0.15$. Only one feature among
the twenty background points changes lifetime, so the
pooled lifetime distribution changes little.

In the location regime, the survival tests remain at
their null rejection rates for every $\Delta$, as
Proposition~\ref{prop:invariance} implies. Landscape and
image MMD also remain near their null rates in this
particular design; the lifetime-invariance proposition
does not assert this for those methods. The band's power
is $0.63$ at $\Delta=0.05$ and $0.86$ at $\Delta=0.30$,
compared with $0.67$ and $0.91$ for embedding MMD.
The two embedding-based methods have similar global
power here. The band additionally identifies local
contrasts and supplies their simultaneous intervals.

\begin{figure*}[!htbp]
  \centering
  \IfFileExists{figures/power.pdf}{\includegraphics[width=\textwidth]{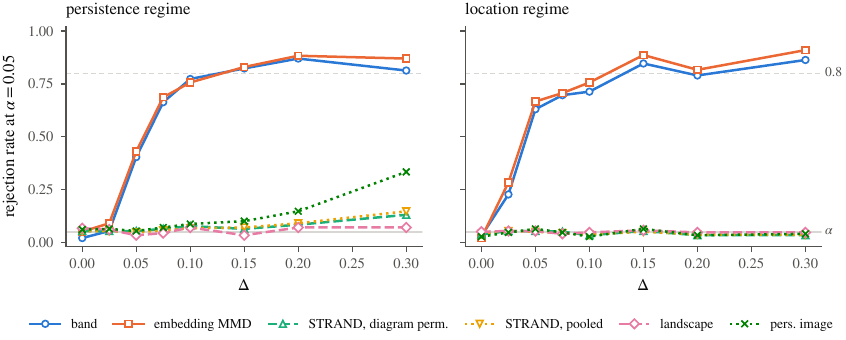}}{\todo{\texttt{\detokenize{figures/power.pdf}}; generated by \texttt{\detokenize{the figure generator}}}}
  \caption{Rejection rate at $\alpha=0.05$ against $\Delta$, $m=30$ per
  class. Left: the signal's lifetime changes. Right: the signal moves
  along the diagonal and its lifetime does not change; every
  lifetime-based procedure is at its level by
  Proposition~\ref{prop:invariance}.}
  \label{fig:power}
\end{figure*}

\begin{table}[!htbp]
\centering\small
\caption{Rejection rates by method and regime.}
\label{tab:comparative}
\IfFileExists{tables/comparative.tex}{{\input{tables/comparative.tex}}}{\todo{tab\_comparative.tex}}
\end{table}

\subsection{Cross-fitting and unequal-covariance nulls}
\label{app:arms}

\paragraph{Cross-fitting.}
This arm uses a fixed total of $40$--$120$ diagrams per
class, split into halves. Landmarks and normalization
are fitted on one half and tested on the other at level
$\alpha/2$, then the roles are swapped. The reported
localization set is the union of the two sets. Every
diagram supplies inference once. Under the conditions
for each held-out band, a union bound gives joint error
control; the two original interval families remain
separate. This experiment does not use the inner CV
score introduced in Appendix~\ref{app:cv-screening}.

Null FWER is $0.014$--$0.046$ across $K$ from $50$ to
$400$ and the reported total sample sizes. With $40$
diagrams per class in total, localization is
$79$--$100\%$. The earlier pilot/inference analysis with
$20$ inference diagrams per class gives $8$--$98\%$,
but also uses its separate $200$-diagram pilot. These
rates therefore describe different total budgets.
With $80$ diagrams per class in total, cross-fitting
localizes in at least $96\%$ of runs.

\paragraph{A multiplicity weak null.}
Group $A$ always has one point at the signal location,
whereas group $B$ has two with probability one half
and none otherwise. The mean measures agree but the
covariances differ. With groups of $20$ and $60$,
the pilot-normalized bootstrap gives FWER
$0.007$--$0.032$ on the grid and with pilot-fitted
landmarks. Normal-reference Bonferroni gives
$0.10$--$0.63$. The sample-normalized band gives
$0.41$--$0.53$, except for $0.03$ on the finest grid.
The permutation maximum with pilot normalization happens to give
$0.045$--$0.051$ in this design. Exchangeability does
not hold; here the covariance difference is confined
to coordinates near the signal, and the larger group
is the more variable one.

\paragraph{More widespread covariance differences.}
In the clustered design, each of group $A$'s $20$
points independently chooses one of four fixed boxes.
All $20$ points of a group-$B$ diagram instead share
one randomly chosen box. Each point has the same
marginal distribution in the two groups, so every
$\delta_k=0$. In the counts design, group $B$ has
$0$ or $40$ noise points where group $A$ has $20$.
These constructions alter variability across the
retained neighborhoods. In the clustered grid design,
only coordinates whose neighborhoods meet the boxes
pass the pilot screen.

Over $1000$ replications per cell (Table~\ref{tab:weaknull-var}), the band's FWER
is $0.005$--$0.057$. The permutation maximum using
pilot normalization gives $0.037$--$0.048$ in balanced
designs and $0.000$--$0.015$ when the larger group is
more variable. When the smaller group is more
variable, it gives $9$--$18\%$ in the clustered design
and $4$--$9\%$ in the counts design. In this allocation,
pooling can understate the variance of the difference.

Recomputing sample normalization for each permutation
does not repair the calibration in these cells:
rejection rates are $18$--$33\%$ in the unbalanced
clustered comparison and $12$--$24\%$ in balanced
clustered designs. Sparse-coordinate variance
estimation remains problematic at these sizes.
Normal-reference Bonferroni with pilot normalization
gives $0.010$--$0.035$ for pilot-fitted landmarks in
the clustered design and $0.10$--$0.53$ in every
other reported cell.

\begin{table}[!htbp]
\centering\small
\caption{Weak nulls with unequal covariance at every landmark. Every
$\delta_k=0$, so the FWER is the probability of any rejection; $1000$
replications per cell. Pilot scale: the band, Bonferroni and the
permutation maximum; sample scale: the band and the permutation maximum
with each coordinate studentized by its sample variance. $m$ / $m'$: the
sizes of classes $A$ and $B$; class $B$ is the more variable. On the
grid only the landmarks that see the boxes pass the pilot screen.}
\label{tab:weaknull-var}
\IfFileExists{tables/weaknull_var.tex}{\input{tables/weaknull_var.tex}}{\todo{tab\_weaknull\_var.tex}}
\end{table}

\subsection{Complete tables for the earlier calibration and localization studies}
\label{app:tables}

The following tables give the full cell grids for
Appendices~\ref{sec:exp-level} and
\ref{sec:exp-localize}. They include grids of $55$,
$105$, $210$, and $406$ landmarks. After the radius
cap and pilot screen, the two finer grids retain
$208$--$210$ and $384$--$389$ coordinates, depending
on $\Delta$. Pilot-fitted families use
$K\in\{50,100,200,400\}$. The inference sizes are
$m\in\{20,40,60\}$ and the shifts are
$\Delta\in\{0,0.10,0.15,0.20,0.30\}$, with $500$
replications and $B=2000$ throughout.

\begin{center}\small
\IfFileExists{tables/level_full.tex}{\input{tables/level_full.tex}}{\todo{tab\_level\_full.tex}}
\end{center}
\begin{center}\small
\IfFileExists{tables/localization_full_grid.tex}{\resizebox{\textwidth}{!}{\input{tables/localization_full_grid.tex}}}{\todo{tab\_localization\_full\_grid.tex}}
\end{center}
\begin{center}\small
\IfFileExists{tables/localization_full_fps.tex}{\resizebox{\textwidth}{!}{\input{tables/localization_full_fps.tex}}}{\todo{tab\_localization\_full\_fps.tex}}
\end{center}

%% file: tables/neigh_s1_cover.tex
\begin{tabular}{llll c c c c c c c c}
\toprule
design & $T_P$ & $T_Q$ & $\Delta$ & F0 & F1 & F2 & A1 & A2 & F1-perm & I1 & ML \\
\midrule
shift-pers & 40 & 40 & 0.0 & 0.97 & 0.97 & 0.96 & 0.94 & 0.95 & -- & 0.97 & -- \\
shift-pers & 80 & 80 & 0.0 & 0.96 & 0.95 & 0.95 & 0.95 & 0.96 & -- & 0.95 & -- \\
shift-pers & 120 & 120 & 0.0 & 0.96 & 0.96 & 0.95 & 0.94 & 0.95 & -- & 0.96 & -- \\
shift-pers & 40 & 120 & 0.0 & 0.97 & 0.95 & 0.96 & 0.96 & 0.97 & -- & 0.96 & -- \\
wn-cluster & 80 & 80 & 0.0 & 0.94 & 0.95 & 0.94 & 0.95 & 0.95 & -- & 0.94 & -- \\
wn-cluster & 120 & 40 & 0.0 & 0.95 & 0.96 & 0.94 & 0.95 & 0.94 & -- & 0.96 & -- \\
wn-cluster & 40 & 120 & 0.0 & 0.95 & 0.95 & 0.94 & 0.95 & 0.95 & -- & 0.95 & -- \\
wn-count & 80 & 80 & 0.0 & 0.95 & 0.95 & 0.97 & 0.96 & 0.97 & -- & 0.95 & -- \\
wn-count & 120 & 40 & 0.0 & 0.97 & 0.97 & 0.98 & 0.96 & 0.97 & -- & 0.96 & -- \\
wn-count & 40 & 120 & 0.0 & 0.96 & 0.96 & 0.96 & 0.95 & 0.96 & -- & 0.95 & -- \\
\bottomrule
\end{tabular}

%% file: tables/neigh_norm_fwer.tex
\begin{tabular}{l r cccccccc}
\toprule
 & & \multicolumn{3}{c}{F0} & \multicolumn{5}{c}{F1} \\
\cmidrule(lr){3-5}\cmidrule(lr){6-10}
design & $T_P$/$T_Q$ & fixed & group & pooled & fixed & group & pooled & hyb.\ 1\% & hyb.\ 0.5\% \\
\midrule
strict null & 40/40 & 0.032 & 0.034 & 0.014 & 0.033 & 0.040 & 0.010 & 0.063 & 0.047 \\
strict null & 80/80 & 0.035 & 0.032 & 0.012 & 0.048 & 0.040 & 0.016 & 0.061 & 0.041 \\
strict null & 120/120 & 0.040 & 0.029 & 0.024 & 0.043 & 0.034 & 0.023 & 0.042 & 0.035 \\
strict null & 40/120 & 0.033 & \textbf{0.375} & 0.022 & 0.054 & \textbf{0.392} & 0.007 & \textbf{0.160} & \textbf{0.341} \\
\addlinespace
clustered & 80/80 & 0.063 & 0.035 & 0.021 & 0.054 & 0.034 & 0.019 & \textbf{0.067} & 0.045 \\
clustered & 120/40 & 0.047 & \textbf{0.403} & 0.032 & 0.041 & \textbf{0.367} & 0.020 & \textbf{0.164} & \textbf{0.330} \\
clustered & 40/120 & 0.045 & \textbf{0.371} & 0.017 & 0.050 & \textbf{0.297} & 0.002 & \textbf{0.093} & \textbf{0.271} \\
\addlinespace
counts & 80/80 & 0.046 & 0.046 & 0.031 & 0.045 & 0.052 & 0.031 & 0.057 & \textbf{0.068} \\
counts & 120/40 & 0.032 & \textbf{0.426} & 0.011 & 0.034 & \textbf{0.479} & 0.009 & \textbf{0.126} & \textbf{0.252} \\
counts & 40/120 & 0.040 & \textbf{0.285} & 0.034 & 0.043 & \textbf{0.292} & 0.008 & \textbf{0.069} & \textbf{0.110} \\
\addlinespace
unequal variances & 40/40 & 0.053 & 0.043 & 0.017 & 0.053 & 0.053 & 0.009 & \textbf{0.082} & \textbf{0.082} \\
unequal variances & 80/80 & 0.045 & 0.049 & 0.036 & 0.049 & 0.055 & 0.027 & 0.063 & \textbf{0.069} \\
unequal variances & 120/120 & 0.054 & 0.046 & 0.035 & 0.051 & 0.040 & 0.023 & 0.051 & 0.049 \\
unequal variances & 40/120 & 0.063 & \textbf{0.397} & 0.028 & 0.057 & \textbf{0.404} & 0.008 & \textbf{0.087} & \textbf{0.146} \\
\bottomrule
\end{tabular}

%% file: tables/neigh_norm_located.tex
\begin{tabular}{l r r cccccccc}
\toprule
 & & & \multicolumn{3}{c}{F0} & \multicolumn{5}{c}{F1} \\
\cmidrule(lr){4-6}\cmidrule(lr){7-11}
design & $T$ & $\Delta$ & fixed & group & pooled & fixed & group & pooled & hyb.\ 1\% & hyb.\ 0.5\% \\
\midrule
unequal variances & 40 & 0.05 & 0.00 & 1.00 & 1.00 & 0.00 & 1.00 & 1.00 & 1.00 & 1.00 \\
unequal variances & 40 & 0.10 & 0.00 & 1.00 & 1.00 & 0.00 & 1.00 & 1.00 & 1.00 & 1.00 \\
unequal variances & 80 & 0.05 & 0.00 & 1.00 & 1.00 & 0.00 & 1.00 & 1.00 & 1.00 & 1.00 \\
unequal variances & 80 & 0.10 & 0.00 & 1.00 & 1.00 & 0.00 & 1.00 & 1.00 & 1.00 & 1.00 \\
unequal variances & 120 & 0.05 & 0.00 & 1.00 & 1.00 & 0.00 & 1.00 & 1.00 & 1.00 & 1.00 \\
unequal variances & 120 & 0.10 & 0.00 & 1.00 & 1.00 & 0.00 & 1.00 & 1.00 & 1.00 & 1.00 \\
\addlinespace
shift, lifetime changes & 40 & 0.05 & 0.77 & 0.51 & 0.37 & 0.19 & 0.45 & 0.29 & 0.50 & 0.46 \\
shift, lifetime changes & 40 & 0.10 & 0.85 & 0.68 & 0.58 & 0.60 & 0.64 & 0.49 & 0.70 & 0.65 \\
shift, lifetime changes & 40 & 0.15 & 0.97 & 0.83 & 0.75 & 0.79 & 0.79 & 0.65 & 0.84 & 0.80 \\
\addlinespace
shift, lifetime kept & 40 & 0.05 & 0.92 & 0.73 & 0.63 & 0.37 & 0.64 & 0.48 & 0.71 & 0.66 \\
shift, lifetime kept & 40 & 0.10 & 0.82 & 0.63 & 0.52 & 0.63 & 0.64 & 0.49 & 0.71 & 0.66 \\
shift, lifetime kept & 40 & 0.15 & 0.91 & 0.76 & 0.65 & 0.73 & 0.73 & 0.57 & 0.78 & 0.73 \\
\bottomrule
\end{tabular}

%% file: tables/neigh_s2.tex
\begin{tabular}{llll c c c c c c c c}
\toprule
design & $T_P$ & $T_Q$ & $\Delta$ & F0 & F1 & F2 & A1 & A2 & F1-perm & I1 & ML \\
\midrule
shift-pers & 40 & 40 & 0.05 & 0.77 & 0.19 & 0.29 & 0.10 & 0.17 & 0.19 & 0.10 & 0.07 \\
shift-pers & 40 & 40 & 0.1 & 0.85 & 0.60 & 0.49 & 0.40 & 0.47 & 0.61 & 0.44 & 0.36 \\
shift-pers & 40 & 40 & 0.15 & 0.97 & 0.79 & 0.64 & 0.52 & 0.57 & 0.80 & 0.62 & 0.58 \\
shift-pers & 80 & 80 & 0.05 & 0.98 & 0.58 & 0.72 & 0.27 & 0.39 & 0.59 & 0.31 & 0.25 \\
shift-pers & 80 & 80 & 0.1 & 1.00 & 0.97 & 0.90 & 0.78 & 0.80 & 0.98 & 0.87 & 0.85 \\
shift-pers & 80 & 80 & 0.15 & 1.00 & 1.00 & 0.96 & 0.87 & 0.88 & 1.00 & 0.96 & 0.96 \\
shift-pers & 120 & 120 & 0.05 & 1.00 & 0.90 & 0.91 & 0.47 & 0.56 & 0.91 & 0.51 & 0.45 \\
shift-pers & 120 & 120 & 0.1 & 1.00 & 1.00 & 0.98 & 0.92 & 0.92 & 1.00 & 0.98 & 0.98 \\
shift-pers & 120 & 120 & 0.15 & 1.00 & 1.00 & 1.00 & 0.96 & 0.96 & 1.00 & 1.00 & 1.00 \\
shift-loc & 40 & 40 & 0.05 & 0.92 & 0.37 & 0.47 & 0.17 & 0.26 & 0.39 & 0.13 & 0.09 \\
shift-loc & 40 & 40 & 0.1 & 0.82 & 0.63 & 0.49 & 0.45 & 0.49 & 0.65 & 0.45 & 0.39 \\
shift-loc & 40 & 40 & 0.15 & 0.91 & 0.73 & 0.56 & 0.49 & 0.52 & 0.74 & 0.60 & 0.57 \\
shift-loc & 80 & 80 & 0.05 & 1.00 & 0.90 & 0.88 & 0.45 & 0.57 & 0.90 & 0.38 & 0.31 \\
shift-loc & 80 & 80 & 0.1 & 0.99 & 0.97 & 0.89 & 0.83 & 0.84 & 0.97 & 0.91 & 0.89 \\
shift-loc & 80 & 80 & 0.15 & 1.00 & 0.99 & 0.94 & 0.85 & 0.86 & 0.99 & 0.96 & 0.97 \\
shift-loc & 120 & 120 & 0.05 & 1.00 & 0.99 & 0.98 & 0.70 & 0.77 & 0.99 & 0.59 & 0.51 \\
shift-loc & 120 & 120 & 0.1 & 1.00 & 1.00 & 0.98 & 0.94 & 0.95 & 1.00 & 0.99 & 0.99 \\
shift-loc & 120 & 120 & 0.15 & 1.00 & 1.00 & 0.99 & 0.95 & 0.95 & 1.00 & 1.00 & 1.00 \\
\bottomrule
\end{tabular}

%% file: tables/neigh_s2_fwer.tex
\begin{tabular}{llll c c c c c c c c}
\toprule
design & $T_P$ & $T_Q$ & $\Delta$ & F0 & F1 & F2 & A1 & A2 & F1-perm & I1 & ML \\
\midrule
shift-pers & 40 & 40 & 0.05 & 0.023 & 0.030 & 0.019 & 0.034 & 0.025 & 0.034 & 0.000 & 0.003 \\
shift-pers & 40 & 40 & 0.1 & 0.025 & 0.026 & 0.026 & 0.029 & 0.022 & 0.027 & 0.004 & 0.008 \\
shift-pers & 40 & 40 & 0.15 & 0.015 & 0.037 & 0.030 & 0.035 & 0.030 & 0.040 & 0.001 & 0.004 \\
shift-pers & 80 & 80 & 0.05 & 0.016 & 0.016 & 0.018 & 0.035 & 0.029 & 0.017 & 0.005 & 0.006 \\
shift-pers & 80 & 80 & 0.1 & 0.027 & 0.033 & 0.018 & 0.028 & 0.025 & 0.036 & 0.003 & 0.010 \\
shift-pers & 80 & 80 & 0.15 & 0.023 & 0.026 & 0.023 & 0.030 & 0.027 & 0.025 & 0.001 & 0.013 \\
shift-pers & 120 & 120 & 0.05 & 0.023 & 0.032 & 0.025 & 0.027 & 0.021 & 0.031 & 0.005 & 0.008 \\
shift-pers & 120 & 120 & 0.1 & 0.024 & 0.032 & 0.025 & 0.032 & 0.027 & 0.031 & 0.002 & 0.015 \\
shift-pers & 120 & 120 & 0.15 & 0.026 & 0.030 & 0.028 & 0.030 & 0.025 & 0.027 & 0.000 & 0.018 \\
shift-loc & 40 & 40 & 0.05 & 0.028 & 0.033 & 0.024 & 0.029 & 0.027 & 0.035 & 0.004 & 0.004 \\
shift-loc & 40 & 40 & 0.1 & 0.028 & 0.034 & 0.029 & 0.028 & 0.028 & 0.036 & 0.001 & 0.008 \\
shift-loc & 40 & 40 & 0.15 & 0.023 & 0.041 & 0.027 & 0.035 & 0.028 & 0.044 & 0.004 & 0.010 \\
shift-loc & 80 & 80 & 0.05 & 0.029 & 0.037 & 0.024 & 0.027 & 0.027 & 0.039 & 0.003 & 0.004 \\
shift-loc & 80 & 80 & 0.1 & 0.033 & 0.027 & 0.025 & 0.034 & 0.035 & 0.025 & 0.001 & 0.013 \\
shift-loc & 80 & 80 & 0.15 & 0.022 & 0.032 & 0.035 & 0.044 & 0.041 & 0.035 & 0.003 & 0.022 \\
shift-loc & 120 & 120 & 0.05 & 0.022 & 0.026 & 0.025 & 0.044 & 0.026 & 0.025 & 0.003 & 0.011 \\
shift-loc & 120 & 120 & 0.1 & 0.035 & 0.034 & 0.036 & 0.041 & 0.034 & 0.031 & 0.003 & 0.026 \\
shift-loc & 120 & 120 & 0.15 & 0.032 & 0.037 & 0.032 & 0.040 & 0.037 & 0.031 & 0.003 & 0.040 \\
\bottomrule
\end{tabular}

%% file: tables/neigh_s3.tex
\begin{tabular}{lllll c c c c c c c c}
\toprule
variant & design & $T_P$ & $T_Q$ & $\Delta$ & F0 & F1 & F2 & A1 & A2 & F1-perm & I1 & ML \\
\midrule
baseline & shift-pers & 80 & 80 & 0.1 & 1.00 & 0.96 & 0.91 & 0.80 & 0.82 & 0.96 & 0.87 & 0.84 \\
q=0.25 & shift-pers & 80 & 80 & 0.1 & 0.53 & 0.25 & 0.18 & 0.14 & 0.18 & 0.26 & 0.16 & 0.13 \\
q=1 & shift-pers & 80 & 80 & 0.1 & 1.00 & 1.00 & 1.00 & 1.00 & 1.00 & 1.00 & 1.00 & 1.00 \\
eps=0.01 & shift-pers & 80 & 80 & 0.1 & 1.00 & 0.99 & 0.97 & 0.87 & 0.89 & 0.99 & 0.94 & 0.95 \\
eps=0.04 & shift-pers & 80 & 80 & 0.1 & 0.39 & 0.59 & 0.31 & 0.34 & 0.33 & 0.59 & 0.41 & 0.34 \\
delta=0.05 & shift-pers & 80 & 80 & 0.05 & 1.00 & 0.58 & 0.75 & 0.28 & 0.37 & 0.58 & 0.31 & 0.24 \\
\bottomrule
\end{tabular}

%% file: tables/neigh_s4.tex
\begin{tabular}{lllll c c c c c c c c}
\toprule
variant & design & $T_P$ & $T_Q$ & $\Delta$ & F0 & F1 & F2 & A1 & A2 & F1-perm & I1 & ML \\
\midrule
two-feature & two-feature & 80 & 80 & -- & 1.00 & 0.91 & 0.83 & 0.72 & 0.73 & 0.91 & 0.48 & 0.41 \\
two-feature & two-feature & 120 & 120 & -- & 1.00 & 0.99 & 0.95 & 0.92 & 0.87 & 0.99 & 0.77 & 0.72 \\
prevalence & prevalence & 80 & 80 & -- & 0.98 & 0.92 & 0.79 & 0.66 & 0.68 & 0.92 & 0.83 & 0.75 \\
continuous jitter & shift-pers & 80 & 80 & 0.0 & 0.00 & 0.00 & 0.00 & 0.00 & 0.00 & 0.00 & 0.00 & 0.00 \\
continuous jitter & shift-pers & 80 & 80 & 0.1 & 1.00 & 1.00 & 0.95 & 0.87 & 0.89 & 1.00 & 0.94 & 0.94 \\
continuous jitter & shift-loc & 80 & 80 & 0.1 & 1.00 & 0.99 & 0.95 & 0.89 & 0.90 & 0.99 & 0.93 & 0.94 \\
\bottomrule
\end{tabular}

%% file: tables/neigh_s5.tex
\begin{tabular}{lllll c c c c c c c c c c c c c c c c c}
\toprule
variant & design & $T_P$ & $T_Q$ & $\Delta$ & F0 & F1 & F2 & A1 & A2 & F1-perm & I1 & ML & F0-permHolm & I1-permHolm & I1\_10 & I1\_20 & I1\_20-permHolm & I1\_40 & ML\_10 & ML\_20 & ML\_40 \\
\midrule
f=1/4 & shift-pers & 80 & 80 & 0.0 & 0.00 & 0.00 & 0.00 & 0.00 & 0.00 & 0.00 & 0.00 & 0.00 & -- & -- & -- & -- & -- & -- & -- & -- & -- \\
f=1/4 & shift-pers & 80 & 80 & 0.1 & 1.00 & 0.96 & 0.90 & 0.87 & 0.88 & 0.96 & 0.87 & 0.85 & -- & -- & -- & -- & -- & -- & -- & -- & -- \\
f=1/3 & shift-pers & 80 & 80 & 0.0 & 0.00 & 0.00 & 0.00 & 0.00 & 0.00 & 0.00 & 0.00 & 0.00 & 0.00 & 0.00 & -- & -- & -- & -- & -- & -- & -- \\
f=1/3 & shift-pers & 80 & 80 & 0.1 & 1.00 & 0.96 & 0.88 & 0.75 & 0.79 & 0.95 & 0.87 & 0.85 & 0.97 & 0.64 & -- & -- & -- & -- & -- & -- & -- \\
f=1/2 & shift-pers & 80 & 80 & 0.0 & 0.00 & 0.00 & 0.00 & 0.00 & 0.00 & 0.00 & 0.00 & 0.00 & -- & -- & -- & -- & -- & -- & -- & -- & -- \\
f=1/2 & shift-pers & 80 & 80 & 0.1 & 1.00 & 0.96 & 0.78 & 0.56 & 0.65 & 0.96 & 0.87 & 0.85 & -- & -- & -- & -- & -- & -- & -- & -- & -- \\
K=50 & shift-pers & 80 & 80 & 0.0 & 0.00 & 0.00 & 0.00 & 0.00 & 0.00 & 0.00 & 0.00 & 0.00 & -- & -- & -- & -- & -- & -- & -- & -- & -- \\
K=50 & shift-pers & 80 & 80 & 0.1 & 1.00 & 0.96 & 0.87 & 0.58 & 0.62 & 0.96 & 0.88 & 0.85 & -- & -- & -- & -- & -- & -- & -- & -- & -- \\
K=200 & shift-pers & 80 & 80 & 0.0 & 0.00 & 0.00 & 0.00 & 0.00 & 0.00 & 0.00 & 0.00 & 0.00 & -- & -- & -- & -- & -- & -- & -- & -- & -- \\
K=200 & shift-pers & 80 & 80 & 0.1 & 1.00 & 0.96 & 0.87 & 0.85 & 0.89 & 0.96 & 0.88 & 0.85 & -- & -- & -- & -- & -- & -- & -- & -- & -- \\
screen=0.0 & shift-pers & 80 & 80 & 0.0 & 0.00 & 0.00 & 0.00 & 0.00 & 0.00 & 0.00 & 0.00 & 0.00 & -- & -- & -- & -- & -- & -- & -- & -- & -- \\
screen=0.0 & shift-pers & 80 & 80 & 0.1 & 1.00 & 0.96 & 0.87 & 0.76 & 0.79 & 0.96 & 0.88 & 0.85 & -- & -- & -- & -- & -- & -- & -- & -- & -- \\
screen=0.01 & shift-pers & 80 & 80 & 0.0 & 0.00 & 0.00 & 0.00 & 0.00 & 0.00 & 0.00 & 0.00 & 0.00 & -- & -- & -- & -- & -- & -- & -- & -- & -- \\
screen=0.01 & shift-pers & 80 & 80 & 0.1 & 1.00 & 0.96 & 0.87 & 0.76 & 0.79 & 0.96 & 0.88 & 0.85 & -- & -- & -- & -- & -- & -- & -- & -- & -- \\
zeta=0.005 & shift-pers & 80 & 80 & 0.0 & 0.00 & 0.00 & 0.00 & 0.00 & 0.00 & 0.00 & 0.00 & 0.00 & -- & -- & -- & -- & -- & -- & -- & -- & -- \\
zeta=0.005 & shift-pers & 80 & 80 & 0.1 & 1.00 & 0.96 & 0.87 & 0.21 & 0.29 & 0.96 & 0.88 & 0.85 & -- & -- & -- & -- & -- & -- & -- & -- & -- \\
zeta=0.01 & shift-pers & 80 & 80 & 0.0 & 0.00 & 0.00 & 0.00 & 0.00 & 0.00 & 0.00 & 0.00 & 0.00 & -- & -- & -- & -- & -- & -- & -- & -- & -- \\
zeta=0.01 & shift-pers & 80 & 80 & 0.1 & 1.00 & 0.96 & 0.87 & 0.55 & 0.50 & 0.96 & 0.88 & 0.85 & -- & -- & -- & -- & -- & -- & -- & -- & -- \\
grid den=10 & shift-pers & 80 & 80 & 0.0 & 0.00 & 0.00 & 0.00 & 0.00 & 0.00 & 0.00 & 0.00 & 0.00 & -- & -- & -- & -- & -- & -- & -- & -- & -- \\
grid den=10 & shift-pers & 80 & 80 & 0.1 & 1.00 & 1.00 & 0.99 & 0.76 & 0.79 & 1.00 & 0.88 & 0.85 & -- & -- & -- & -- & -- & -- & -- & -- & -- \\
grid den=20 & shift-pers & 80 & 80 & 0.0 & 0.00 & 0.00 & 0.00 & 0.00 & 0.00 & 0.00 & 0.00 & 0.00 & -- & -- & -- & -- & -- & -- & -- & -- & -- \\
grid den=20 & shift-pers & 80 & 80 & 0.1 & 1.00 & 1.00 & 0.97 & 0.76 & 0.79 & 1.00 & 0.88 & 0.85 & -- & -- & -- & -- & -- & -- & -- & -- & -- \\
image grid & shift-pers & 80 & 80 & 0.0 & 0.00 & 0.00 & 0.00 & 0.00 & 0.00 & 0.00 & -- & -- & 0.00 & -- & 0.00 & 0.00 & 0.00 & 0.00 & 0.00 & 0.00 & 0.00 \\
image grid & shift-pers & 80 & 80 & 0.1 & 1.00 & 0.98 & 0.88 & 0.79 & 0.80 & 0.97 & -- & -- & 0.97 & -- & 0.12 & 0.88 & 0.66 & 1.00 & 0.12 & 0.87 & 1.00 \\
offset a=(0.31,0.51) & shift-pers & 80 & 80 & 0.05 & 0.94 & 0.37 & 0.55 & 0.30 & 0.39 & 0.38 & 0.32 & 0.24 & -- & -- & -- & -- & -- & -- & -- & -- & -- \\
offset a=(0.31,0.51) & shift-pers & 80 & 80 & 0.1 & 0.93 & 0.94 & 0.77 & 0.79 & 0.83 & 0.94 & 0.89 & 0.86 & -- & -- & -- & -- & -- & -- & -- & -- & -- \\
offset a=(0.31,0.51) & shift-pers & 80 & 80 & 0.15 & 0.99 & 0.98 & 0.88 & 0.89 & 0.90 & 0.98 & 0.96 & 0.96 & -- & -- & -- & -- & -- & -- & -- & -- & -- \\
q=0 null & background-only & 80 & 80 & 0.0 & 0.00 & 0.00 & 0.00 & 0.00 & 0.00 & 0.00 & 0.00 & 0.00 & -- & -- & -- & -- & -- & -- & -- & -- & -- \\
\bottomrule
\end{tabular}

%% file: tables/neigh_s5_fwer.tex
\begin{tabular}{lllll c c c c c c c c c c c c c c c c c}
\toprule
variant & design & $T_P$ & $T_Q$ & $\Delta$ & F0 & F1 & F2 & A1 & A2 & F1-perm & I1 & ML & F0-permHolm & I1-permHolm & I1\_10 & I1\_20 & I1\_20-permHolm & I1\_40 & ML\_10 & ML\_20 & ML\_40 \\
\midrule
f=1/4 & shift-pers & 80 & 80 & 0.0 & 0.040 & 0.068 & 0.050 & 0.048 & 0.044 & 0.070 & 0.044 & 0.022 & -- & -- & -- & -- & -- & -- & -- & -- & -- \\
f=1/4 & shift-pers & 80 & 80 & 0.1 & 0.028 & 0.020 & 0.030 & 0.034 & 0.034 & 0.024 & 0.006 & 0.014 & -- & -- & -- & -- & -- & -- & -- & -- & -- \\
f=1/3 & shift-pers & 80 & 80 & 0.0 & 0.040 & 0.068 & 0.044 & 0.050 & 0.048 & 0.068 & 0.046 & 0.022 & 0.044 & 0.028 & -- & -- & -- & -- & -- & -- & -- \\
f=1/3 & shift-pers & 80 & 80 & 0.1 & 0.028 & 0.020 & 0.028 & 0.040 & 0.030 & 0.020 & 0.006 & 0.014 & 0.036 & 0.006 & -- & -- & -- & -- & -- & -- & -- \\
f=1/2 & shift-pers & 80 & 80 & 0.0 & 0.040 & 0.068 & 0.026 & 0.042 & 0.036 & 0.070 & 0.046 & 0.022 & -- & -- & -- & -- & -- & -- & -- & -- & -- \\
f=1/2 & shift-pers & 80 & 80 & 0.1 & 0.028 & 0.020 & 0.018 & 0.028 & 0.022 & 0.024 & 0.006 & 0.014 & -- & -- & -- & -- & -- & -- & -- & -- & -- \\
K=50 & shift-pers & 80 & 80 & 0.0 & 0.040 & 0.068 & 0.042 & 0.044 & 0.048 & 0.070 & 0.046 & 0.022 & -- & -- & -- & -- & -- & -- & -- & -- & -- \\
K=50 & shift-pers & 80 & 80 & 0.1 & 0.028 & 0.020 & 0.028 & 0.028 & 0.018 & 0.024 & 0.006 & 0.014 & -- & -- & -- & -- & -- & -- & -- & -- & -- \\
K=200 & shift-pers & 80 & 80 & 0.0 & 0.040 & 0.068 & 0.042 & 0.046 & 0.042 & 0.070 & 0.046 & 0.022 & -- & -- & -- & -- & -- & -- & -- & -- & -- \\
K=200 & shift-pers & 80 & 80 & 0.1 & 0.028 & 0.020 & 0.028 & 0.032 & 0.028 & 0.024 & 0.006 & 0.014 & -- & -- & -- & -- & -- & -- & -- & -- & -- \\
screen=0.0 & shift-pers & 80 & 80 & 0.0 & 0.040 & 0.068 & 0.042 & 0.054 & 0.050 & 0.070 & 0.046 & 0.022 & -- & -- & -- & -- & -- & -- & -- & -- & -- \\
screen=0.0 & shift-pers & 80 & 80 & 0.1 & 0.028 & 0.020 & 0.028 & 0.036 & 0.028 & 0.024 & 0.006 & 0.014 & -- & -- & -- & -- & -- & -- & -- & -- & -- \\
screen=0.01 & shift-pers & 80 & 80 & 0.0 & 0.040 & 0.068 & 0.042 & 0.054 & 0.050 & 0.070 & 0.046 & 0.022 & -- & -- & -- & -- & -- & -- & -- & -- & -- \\
screen=0.01 & shift-pers & 80 & 80 & 0.1 & 0.028 & 0.020 & 0.028 & 0.036 & 0.028 & 0.024 & 0.006 & 0.014 & -- & -- & -- & -- & -- & -- & -- & -- & -- \\
zeta=0.005 & shift-pers & 80 & 80 & 0.0 & 0.040 & 0.068 & 0.042 & 0.046 & 0.060 & 0.070 & 0.046 & 0.022 & -- & -- & -- & -- & -- & -- & -- & -- & -- \\
zeta=0.005 & shift-pers & 80 & 80 & 0.1 & 0.028 & 0.020 & 0.028 & 0.028 & 0.030 & 0.024 & 0.006 & 0.014 & -- & -- & -- & -- & -- & -- & -- & -- & -- \\
zeta=0.01 & shift-pers & 80 & 80 & 0.0 & 0.040 & 0.068 & 0.042 & 0.036 & 0.042 & 0.070 & 0.046 & 0.022 & -- & -- & -- & -- & -- & -- & -- & -- & -- \\
zeta=0.01 & shift-pers & 80 & 80 & 0.1 & 0.028 & 0.020 & 0.028 & 0.036 & 0.028 & 0.024 & 0.006 & 0.014 & -- & -- & -- & -- & -- & -- & -- & -- & -- \\
grid den=10 & shift-pers & 80 & 80 & 0.0 & 0.048 & 0.068 & 0.038 & 0.054 & 0.052 & 0.060 & 0.046 & 0.022 & -- & -- & -- & -- & -- & -- & -- & -- & -- \\
grid den=10 & shift-pers & 80 & 80 & 0.1 & 0.028 & 0.020 & 0.028 & 0.036 & 0.028 & 0.020 & 0.006 & 0.014 & -- & -- & -- & -- & -- & -- & -- & -- & -- \\
grid den=20 & shift-pers & 80 & 80 & 0.0 & 0.046 & 0.058 & 0.030 & 0.054 & 0.052 & 0.062 & 0.046 & 0.022 & -- & -- & -- & -- & -- & -- & -- & -- & -- \\
grid den=20 & shift-pers & 80 & 80 & 0.1 & 0.014 & 0.026 & 0.020 & 0.036 & 0.028 & 0.026 & 0.006 & 0.014 & -- & -- & -- & -- & -- & -- & -- & -- & -- \\
image grid & shift-pers & 80 & 80 & 0.0 & 0.040 & 0.050 & 0.032 & 0.056 & 0.058 & 0.052 & -- & -- & 0.036 & -- & 0.050 & 0.044 & 0.026 & 0.040 & 0.042 & 0.042 & 0.020 \\
image grid & shift-pers & 80 & 80 & 0.1 & 0.030 & 0.036 & 0.026 & 0.024 & 0.032 & 0.036 & -- & -- & 0.038 & -- & 0.000 & 0.006 & 0.010 & 0.008 & 0.000 & 0.014 & 0.158 \\
offset a=(0.31,0.51) & shift-pers & 80 & 80 & 0.05 & 0.032 & 0.034 & 0.022 & 0.034 & 0.032 & 0.038 & 0.006 & 0.012 & -- & -- & -- & -- & -- & -- & -- & -- & -- \\
offset a=(0.31,0.51) & shift-pers & 80 & 80 & 0.1 & 0.048 & 0.024 & 0.022 & 0.010 & 0.014 & 0.026 & 0.004 & 0.018 & -- & -- & -- & -- & -- & -- & -- & -- & -- \\
offset a=(0.31,0.51) & shift-pers & 80 & 80 & 0.15 & 0.024 & 0.036 & 0.030 & 0.030 & 0.028 & 0.036 & 0.004 & 0.018 & -- & -- & -- & -- & -- & -- & -- & -- & -- \\
q=0 null & background-only & 80 & 80 & 0.0 & 0.048 & 0.052 & 0.036 & 0.040 & 0.028 & 0.054 & 0.042 & 0.022 & -- & -- & -- & -- & -- & -- & -- & -- & -- \\
\bottomrule
\end{tabular}

%% file: tables/cv_selection.tex
\begin{tabular}{lrrrr}
\toprule
Method & Null FWER & Power, $1$ & Power, $4$ & Coverage, $4$ \\
\midrule
All radii jointly & 5.5 & 25.9 & 69.8 & 92.7 \\
Pilot: radius & 8.2 & 10.3 & 41.6 & 93.9 \\
Pilot: region and radius & 5.1 & 3.2 & 70.8 & 94.7 \\
Cross-fitted: radius & 6.9 & 8.1 & 39.4 & 91.4 \\
Cross-fitted: region and radius & 6.5 & 4.5 & 82.6 & 93.3 \\
Oracle region & 5.4 & 58.2 & 98.1 & 93.9 \\
\bottomrule
\end{tabular}

%% file: tables/chem_primary.tex
\begin{tabular}{l l c r cc cc cc cc cc}
\toprule
 & & & & \multicolumn{2}{c}{F0} & \multicolumn{2}{c}{F1} & \multicolumn{2}{c}{F2} & \multicolumn{2}{c}{A1} & \multicolumn{2}{c}{A2} \\
dataset & filtration & deg. & $n_0$/$n_1$ & loc. & $p$ & loc. & $p$ & loc. & $p$ & loc. & $p$ & loc. & $p$ \\
\midrule
MUTAG & HKS $t=10$ & 1 & 63/112 & 1 & 0.024 & 3 & 0.005 & 1 & 0.001 & 58 & $<$0.001 & 26 & $<$0.001 \\
MUTAG & HKS $t=10$ & 0 & 63/112 & 2 & $<$0.001 & 7 & $<$0.001 & 2 & $<$0.001 & 116 & $<$0.001 & 58 & $<$0.001 \\
COX2 & HKS $t=10$ & 1 & 363/102 & 0 & 0.262 & 1 & 0.015 & 0 & 0.123 & 0 & 0.252 & 0 & 0.383 \\
DHFR & HKS $t=10$ & 1 & 269/449 & 1 & $<$0.001 & 4 & $<$0.001 & 1 & $<$0.001 & 20 & $<$0.001 & 7 & $<$0.001 \\
COX2 & HKS $t=10$ & 0 & 363/102 & 0 & 0.337 & 0 & 0.381 & 0 & 0.317 & 8 & 0.020 & 1 & 0.019 \\
COX2 & HKS $t=1$ & 1 & 363/102 & 0 & 0.166 & 0 & 0.216 & 0 & 0.191 & 0 & 0.054 & 4 & 0.040 \\
PTC & HKS $t=10$ & 1 & 183/142 & 0 & 0.805 & 0 & 0.390 & 0 & 0.593 & 1 & 0.039 & 0 & 0.340 \\
PTC & HKS $t=10$ & 0 & 183/142 & 0 & 0.141 & 2 & 0.033 & 0 & 0.206 & 0 & 0.618 & 0 & 0.567 \\
PTC & HKS $t=5$ & 1 & 183/142 & 0 & 0.312 & 0 & 0.267 & 0 & 0.493 & 0 & 0.505 & 0 & 0.521 \\
PTC & HKS $t=5$ & 0 & 183/142 & 0 & 0.549 & 0 & 0.252 & 0 & 0.401 & 0 & 0.874 & 0 & 0.826 \\
DHFR & HKS $t=10$ & 0 & 269/449 & 1 & $<$0.001 & 3 & $<$0.001 & 1 & 0.001 & 51 & $<$0.001 & 4 & $<$0.001 \\
COX2 & HKS $t=1$ & 0 & 363/102 & 0 & 0.167 & 0 & 0.201 & 0 & 0.063 & 23 & 0.005 & 5 & 0.038 \\
\bottomrule
\end{tabular}

%% file: tables/chem_stability.tex
\begin{tabular}{l l c l r c c r}
\toprule
dataset & filtration & deg. & method & empty & overlap & matched & loc. (median) \\
\midrule
MUTAG & HKS $t=10$ & 1 & F2 & 2/20 & 1.00 & 0.89 & 1.0 \\
MUTAG & HKS $t=10$ & 1 & A1 & 0/20 & 0.62 & 0.97 & 50.5 \\
MUTAG & HKS $t=10$ & 1 & A2 & 0/20 & 0.65 & 0.97 & 23.5 \\
MUTAG & HKS $t=10$ & 0 & F2 & 0/20 & 0.66 & 0.93 & 2.0 \\
MUTAG & HKS $t=10$ & 0 & A1 & 0/20 & 0.59 & 1.00 & 57.0 \\
MUTAG & HKS $t=10$ & 0 & A2 & 0/20 & 0.20 & 0.93 & 14.5 \\
DHFR & HKS $t=10$ & 1 & F2 & 1/20 & 1.00 & 0.95 & 1.0 \\
DHFR & HKS $t=10$ & 1 & A1 & 0/20 & 0.64 & 0.93 & 24.0 \\
DHFR & HKS $t=10$ & 1 & A2 & 0/20 & 0.36 & 0.87 & 8.0 \\
DHFR & HKS $t=10$ & 0 & F2 & 0/20 & 1.00 & 1.00 & 1.0 \\
DHFR & HKS $t=10$ & 0 & A1 & 0/20 & 0.94 & 1.00 & 53.0 \\
DHFR & HKS $t=10$ & 0 & A2 & 0/20 & 0.17 & 0.97 & 14.5 \\
COX2 & HKS $t=10$ & 1 & F2 & 10/20 & 1.00 & 0.47 & 0.5 \\
COX2 & HKS $t=10$ & 1 & A1 & 4/20 & 0.34 & 0.48 & 2.0 \\
COX2 & HKS $t=10$ & 1 & A2 & 4/20 & 0.26 & 0.48 & 2.0 \\
COX2 & HKS $t=10$ & 0 & F2 & 16/20 & 1.00 & 0.16 & 0.0 \\
COX2 & HKS $t=10$ & 0 & A1 & 13/20 & 0.48 & 0.32 & 0.0 \\
COX2 & HKS $t=10$ & 0 & A2 & 13/20 & 0.43 & 0.32 & 0.0 \\
COX2 & HKS $t=1$ & 1 & F2 & 17/20 & 0.50 & 0.11 & 0.0 \\
COX2 & HKS $t=1$ & 1 & A1 & 11/20 & 0.00 & 0.12 & 0.0 \\
COX2 & HKS $t=1$ & 1 & A2 & 12/20 & 0.31 & 0.19 & 0.0 \\
COX2 & HKS $t=1$ & 0 & F2 & 16/20 & 0.74 & 0.16 & 0.0 \\
COX2 & HKS $t=1$ & 0 & A1 & 10/20 & 0.24 & 0.29 & 1.5 \\
COX2 & HKS $t=1$ & 0 & A2 & 13/20 & 0.16 & 0.17 & 0.0 \\
PTC & HKS $t=10$ & 1 & F2 & 19/20 & -- & 0.00 & 0.0 \\
PTC & HKS $t=10$ & 1 & A1 & 20/20 & -- & -- & 0.0 \\
PTC & HKS $t=10$ & 1 & A2 & 20/20 & -- & -- & 0.0 \\
PTC & HKS $t=10$ & 0 & F2 & 16/20 & 0.12 & 0.16 & 0.0 \\
PTC & HKS $t=10$ & 0 & A1 & 17/20 & 0.80 & 0.11 & 0.0 \\
PTC & HKS $t=10$ & 0 & A2 & 19/20 & -- & 0.00 & 0.0 \\
PTC & HKS $t=5$ & 1 & F2 & 20/20 & -- & -- & 0.0 \\
PTC & HKS $t=5$ & 1 & A1 & 20/20 & -- & -- & 0.0 \\
PTC & HKS $t=5$ & 1 & A2 & 20/20 & -- & -- & 0.0 \\
PTC & HKS $t=5$ & 0 & F2 & 20/20 & -- & -- & 0.0 \\
PTC & HKS $t=5$ & 0 & A1 & 20/20 & -- & -- & 0.0 \\
PTC & HKS $t=5$ & 0 & A2 & 20/20 & -- & -- & 0.0 \\
\bottomrule
\end{tabular}

%% file: tables/chem_curves.tex
\begin{tabular}{l r l c r c}
\toprule
dataset & per class & method & any localized & loc. (median) & overlap with full map \\
\midrule
DHFR & 40 & A1 & 0.28 & 0.0 & 0.00 \\
DHFR & 40 & A2 & 0.28 & 0.0 & 0.00 \\
DHFR & 40 & F1 & 0.30 & 0.0 & 0.00 \\
DHFR & 80 & A1 & 0.52 & 1.0 & 0.07 \\
DHFR & 80 & A2 & 0.58 & 1.0 & 0.02 \\
DHFR & 80 & F1 & 0.48 & 0.0 & 0.00 \\
DHFR & 120 & A1 & 0.90 & 5.5 & 0.19 \\
DHFR & 120 & A2 & 0.84 & 3.0 & 0.09 \\
DHFR & 120 & F1 & 0.92 & 3.0 & 1.00 \\
MUTAG & 20 & A1 & 0.94 & 7.5 & 0.48 \\
MUTAG & 20 & A2 & 0.92 & 5.0 & 0.46 \\
MUTAG & 20 & F1 & 0.28 & 0.0 & 0.00 \\
MUTAG & 40 & A1 & 1.00 & 29.5 & 0.69 \\
MUTAG & 40 & A2 & 1.00 & 16.0 & 0.63 \\
MUTAG & 40 & F1 & 0.56 & 1.0 & 0.50 \\
MUTAG & 60 & A1 & 1.00 & 46.5 & 0.74 \\
MUTAG & 60 & A2 & 1.00 & 22.5 & 0.67 \\
MUTAG & 60 & F1 & 0.88 & 2.0 & 0.50 \\
\bottomrule
\end{tabular}

%% file: tables/mutag_structure.tex
\begin{tabular}{l l r c c}
\toprule
map & cycle basis & localized & rings in $\ge3$-fused systems in region & other rings in region \\
\midrule
A1 & minimum cycle basis & 58 & 331/346 (0.96) & 68/140 (0.49) \\
A1 & cycle basis & 58 & 297/312 (0.95) & 68/140 (0.49) \\
F1 & minimum cycle basis & 3 & 7/346 (0.02) & 24/140 (0.17) \\
F1 & cycle basis & 3 & 7/312 (0.02) & 24/140 (0.17) \\
\bottomrule
\end{tabular}

%% file: tables/chemgraph.tex
\begin{tabular}{llc rr c cc rrrcc}
\toprule
dataset & filtration & $H_d$ & $m{+}m'$ & $\ell$ & $\max_k|Z_k|$ & $\hat c_\alpha$ & $p$ & $|\hat R|$ & $|\hat R_{\mathrm{sd}}|$ & $|\hat R_{\mathrm{BH}}|$ & recovery & top agree \\
\midrule
COX2 & \texttt{hks\_t1} & $H_0$ & $243{+}68$ & 196 & 3.92 & 3.23 & 0.0076 & 14 & 14 & 26 & 0.42 & 0.80 \\
COX2 & \texttt{hks\_t1} & $H_1$ & $243{+}68$ & 196 & 3.11 & 3.32 & 0.0820 & 0 & 0 & 0 & 0.40 & 1.00 \\
COX2 & \texttt{hks\_t10} & $H_0$ & $243{+}68$ & 182 & 3.73 & 4.27 & 0.1144 & 0 & 0 & 2 & 0.12 & 0.00 \\
COX2 & \texttt{hks\_t10} & $H_1$ & $243{+}68$ & 168 & 7.83 & 6.78 & 0.0268 & 1 & 3 & 9 & 0.34 & 0.65 \\
COX2 & \texttt{ricci} & $H_0$ & $243{+}68$ & 164 & 4.19 & 4.07 & 0.0418 & 1 & 1 & 1 & 0.28 & 0.60 \\
COX2 & \texttt{ricci} & $H_1$ & $243{+}68$ & 192 & 3.96 & 3.18 & 0.0048 & 2 & 2 & 2 & 0.35 & 0.55 \\
DHFR & \texttt{hks\_t10} & $H_0$ & $197{+}307$ & 192 & 9.52 & 3.61 & 0.0002 & 55 & 56 & 82 & 0.97 & 1.00 \\
DHFR & \texttt{hks\_t10} & $H_1$ & $197{+}307$ & 178 & 7.69 & 3.93 & 0.0002 & 30 & 32 & 65 & 0.94 & 0.90 \\
MUTAG & \texttt{hks\_t10} & $H_0$ & $42{+}83$ & 187 & 7.71 & 3.89 & 0.0002 & 72 & 87 & 132 & 0.89 & 0.85 \\
MUTAG & \texttt{hks\_t10} & $H_1$ & $42{+}83$ & 145 & 8.80 & 5.07 & 0.0002 & 21 & 22 & 76 & 0.94 & 0.90 \\
MUTAG & \texttt{nodelabel} & $H_0$ & $42{+}83$ & 73 & 1.61 & 2.27 & 0.2012 & 0 & 0 & 0 & 0.00 & 0.30 \\
MUTAG & \texttt{nodelabel} & $H_1$ & $42{+}83$ & 164 & 3.58 & 2.30 & 0.0004 & 141 & 141 & 141 & 0.85 & 0.70 \\
PTC & \texttt{hks\_t10} & $H_0$ & $128{+}101$ & 177 & 2.92 & 4.36 & 0.4777 & 0 & 0 & 0 & 0.04 & 0.60 \\
PTC & \texttt{hks\_t10} & $H_1$ & $128{+}101$ & 145 & 3.17 & 4.17 & 0.3015 & 0 & 0 & 0 & -- & 0.40 \\
PTC & \texttt{hks\_t5} & $H_0$ & $128{+}101$ & 181 & 3.20 & 4.32 & 0.2478 & 0 & 0 & 0 & -- & 0.00 \\
PTC & \texttt{hks\_t5} & $H_1$ & $128{+}101$ & 163 & 3.88 & 3.82 & 0.0442 & 1 & 1 & 1 & -- & 0.30 \\
\bottomrule
\end{tabular}

%% file: tables/level.tex
\begin{tabular}{llrr ccc c}
\toprule
placement & $\ell$ & $m$ & & cover & FWER$_{\mathrm{boot}}$ & FWER$_{\mathrm{Bonf}}$ & $\hat c_\alpha$ / $z_{\mathrm{Bonf}}$ \\
\midrule
FPS & 50 & 20 & & 0.978 & 0.022 & 0.070 & 3.64 / 3.29 \\
FPS & 50 & 60 & & 0.970 & 0.030 & 0.054 & 3.47 / 3.29 \\
FPS & 100 & 20 & & 0.984 & 0.016 & 0.086 & 3.84 / 3.48 \\
FPS & 100 & 60 & & 0.952 & 0.048 & 0.082 & 3.62 / 3.48 \\
FPS & 200 & 20 & & 0.988 & 0.012 & 0.086 & 4.15 / 3.66 \\
FPS & 200 & 60 & & 0.968 & 0.032 & 0.056 & 3.86 / 3.66 \\
FPS & 395 & 20 & & 0.992 & 0.008 & 0.082 & 4.50 / 3.83 \\
FPS & 395 & 60 & & 0.966 & 0.034 & 0.066 & 4.12 / 3.83 \\
grid & 55 & 20 & & 0.966 & 0.034 & 0.088 & 3.67 / 3.32 \\
grid & 55 & 60 & & 0.946 & 0.054 & 0.084 & 3.50 / 3.32 \\
grid & 105 & 20 & & 0.986 & 0.014 & 0.106 & 4.06 / 3.49 \\
grid & 105 & 60 & & 0.968 & 0.032 & 0.074 & 3.77 / 3.49 \\
grid & 208 & 20 & & 0.982 & 0.018 & 0.252 & 4.98 / 3.67 \\
grid & 208 & 60 & & 0.978 & 0.022 & 0.198 & 4.41 / 3.67 \\
grid & 385 & 20 & & 0.990 & 0.010 & 0.482 & 5.99 / 3.83 \\
grid & 385 & 60 & & 0.994 & 0.006 & 0.332 & 5.16 / 3.83 \\
\bottomrule
\end{tabular}

%% file: tables/arms.tex
\begin{tabular}{l l r r cccc c c}
\toprule
arm & scale & $\ell$ & $m$ / $m'$ & cover & FWER$_{\mathrm{boot}}$ & FWER$_{\mathrm{Bonf}}$ & FWER$_{\mathrm{perm}}$ & power$_{0.2}$ & $\hat c_\alpha$ \\
\midrule
FPS in-sample & sample & 49 & 20 / 20 & -- & 0.042 & 0.036 & -- & 0.87 & 3.16 \\
FPS in-sample & sample & 49 & 60 / 60 & -- & 0.044 & 0.030 & -- & 0.98 & 3.23 \\
FPS in-sample & sample & 99 & 20 / 20 & -- & 0.068 & 0.040 & -- & 0.96 & 3.33 \\
FPS in-sample & sample & 99 & 60 / 60 & -- & 0.028 & 0.026 & -- & 1.00 & 3.41 \\
FPS in-sample & sample & 199 & 20 / 20 & -- & 0.060 & 0.032 & -- & 1.00 & 3.49 \\
FPS in-sample & sample & 199 & 60 / 60 & -- & 0.026 & 0.022 & -- & 1.00 & 3.57 \\
\addlinespace
grid, 20 vs 60 & pilot & 105 & 20 / 60 & 0.977 & 0.023 & 0.124 & 0.051 & 1.00 & 4.18 \\
grid, 20 vs 60 & pilot & 385 & 20 / 60 & 0.989 & 0.011 & 0.658 & 0.051 & -- & 6.64 \\
grid, 20 vs 60 & sample & 105 & 20 / 60 & 0.489 & 0.510 & 0.446 & 0.064 & 1.00 & 3.39 \\
grid, 20 vs 60 & sample & 385 & 20 / 60 & 0.969 & 0.031 & 0.019 & 0.044 & -- & 3.70 \\
\addlinespace
grid, balanced & pilot & 105 & 20 / 20 & 0.979 & 0.021 & 0.111 & 0.052 & 1.00 & 4.06 \\
grid, balanced & pilot & 385 & 20 / 20 & 0.992 & 0.008 & 0.503 & 0.036 & -- & 6.01 \\
grid, balanced & sample & 105 & 20 / 20 & 0.945 & 0.056 & 0.036 & 0.057 & 1.00 & 3.37 \\
grid, balanced & sample & 385 & 20 / 20 & 0.997 & 0.003 & 0.002 & 0.046 & -- & 3.61 \\
\addlinespace
pilot FPS & pilot & 50 & 20 / 20 & 0.978 & 0.022 & 0.070 & -- & 0.97 & 3.64 \\
pilot FPS & pilot & 50 & 60 / 60 & 0.970 & 0.030 & 0.054 & -- & 1.00 & 3.47 \\
pilot FPS & pilot & 100 & 20 / 20 & 0.984 & 0.016 & 0.086 & -- & 0.81 & 3.84 \\
pilot FPS & pilot & 100 & 60 / 60 & 0.952 & 0.048 & 0.082 & -- & 1.00 & 3.62 \\
pilot FPS & pilot & 200 & 20 / 20 & 0.988 & 0.012 & 0.086 & -- & 0.94 & 4.15 \\
pilot FPS & pilot & 200 & 60 / 60 & 0.968 & 0.032 & 0.056 & -- & 1.00 & 3.86 \\
pilot FPS & pilot & 395 & 20 / 20 & 0.992 & 0.008 & 0.082 & -- & -- & 4.50 \\
pilot FPS & pilot & 395 & 60 / 60 & 0.966 & 0.034 & 0.066 & -- & -- & 4.12 \\
\addlinespace
selection in-sample & sample & 50 & 20 / 20 & -- & 0.142 & 0.108 & -- & 1.00 & 3.15 \\
selection in-sample & sample & 50 & 60 / 60 & -- & 0.154 & 0.120 & -- & 1.00 & 3.21 \\
selection in-sample & sample & 100 & 20 / 20 & -- & 0.218 & 0.124 & -- & 1.00 & 3.31 \\
selection in-sample & sample & 100 & 60 / 60 & -- & 0.136 & 0.092 & -- & 1.00 & 3.38 \\
selection in-sample & sample & 200 & 20 / 20 & -- & 0.242 & 0.108 & -- & 1.00 & 3.40 \\
selection in-sample & sample & 200 & 60 / 60 & -- & 0.128 & 0.070 & -- & 1.00 & 3.50 \\
\addlinespace
selection on pilot & pilot & 50 & 20 / 20 & 0.984 & 0.016 & 0.080 & -- & 0.99 & 3.66 \\
selection on pilot & pilot & 50 & 60 / 60 & 0.966 & 0.034 & 0.062 & -- & 1.00 & 3.47 \\
selection on pilot & pilot & 100 & 20 / 20 & 0.980 & 0.020 & 0.080 & -- & 1.00 & 3.92 \\
selection on pilot & pilot & 100 & 60 / 60 & 0.968 & 0.032 & 0.074 & -- & 1.00 & 3.65 \\
selection on pilot & pilot & 200 & 20 / 20 & 0.980 & 0.020 & 0.068 & -- & 1.00 & 4.15 \\
selection on pilot & pilot & 200 & 60 / 60 & 0.974 & 0.026 & 0.044 & -- & 1.00 & 3.81 \\
\addlinespace
weak null, FPS & pilot & 100 & 20 / 60 & 0.968 & 0.032 & 0.104 & 0.048 & -- & 3.98 \\
weak null, FPS & pilot & 396 & 20 / 60 & 0.981 & 0.019 & 0.174 & 0.049 & -- & 4.83 \\
weak null, FPS & sample & 100 & 20 / 60 & 0.589 & 0.411 & 0.327 & 0.058 & -- & 3.35 \\
weak null, FPS & sample & 396 & 20 / 60 & 0.482 & 0.518 & 0.386 & 0.053 & -- & 3.64 \\
\addlinespace
weak null, grid & pilot & 105 & 20 / 60 & 0.980 & 0.020 & 0.150 & 0.051 & -- & 4.21 \\
weak null, grid & pilot & 382 & 20 / 60 & 0.993 & 0.007 & 0.626 & 0.045 & -- & 6.45 \\
weak null, grid & sample & 105 & 20 / 60 & 0.472 & 0.528 & 0.446 & 0.064 & -- & 3.39 \\
weak null, grid & sample & 382 & 20 / 60 & 0.973 & 0.027 & 0.015 & 0.040 & -- & 3.70 \\
\bottomrule
\end{tabular}

%% file: tables/localization.tex
\begin{tabular}{llrr c ccc cc c}
\toprule
placement & $\ell$ & $\Delta$ & $m$ & cover & power$_{\mathrm{boot}}$ & power$_{\mathrm{Bonf}}$ & power$_{\mathrm{BH}}$ & FDR$_{\mathrm{BH}}$ & FDR$_{\mathrm{boot}}$ & loc.\ err.\ \\
\midrule
FPS & 50 & 0.10 & 20 & 0.970 & 0.46 & 0.54 & 0.55 & 0.025 & 0.006 & 0.073 \\
FPS & 50 & 0.10 & 60 & 0.976 & 0.99 & 0.99 & 1.00 & 0.020 & 0.005 & 0.043 \\
FPS & 50 & 0.30 & 20 & 0.974 & 0.08 & 0.15 & 0.15 & 0.021 & 0.013 & 0.108 \\
FPS & 50 & 0.30 & 60 & 0.966 & 0.49 & 0.56 & 0.58 & 0.027 & 0.013 & 0.075 \\
FPS & 399 & 0.10 & 20 & 0.994 & 0.95 & 0.99 & 0.99 & 0.040 & 0.000 & 0.007 \\
FPS & 399 & 0.10 & 60 & 0.974 & 1.00 & 1.00 & 1.00 & 0.034 & 0.002 & 0.006 \\
FPS & 395 & 0.30 & 20 & 0.988 & 0.96 & 1.00 & 1.00 & 0.019 & 0.000 & 0.016 \\
FPS & 395 & 0.30 & 60 & 0.972 & 1.00 & 1.00 & 1.00 & 0.013 & 0.001 & 0.013 \\
grid & 55 & 0.10 & 20 & 0.958 & 0.99 & 1.00 & 1.00 & 0.033 & 0.006 & 0.016 \\
grid & 55 & 0.10 & 60 & 0.966 & 1.00 & 1.00 & 1.00 & 0.027 & 0.007 & 0.016 \\
grid & 55 & 0.30 & 20 & 0.970 & 0.98 & 1.00 & 1.00 & 0.025 & 0.004 & 0.019 \\
grid & 55 & 0.30 & 60 & 0.954 & 1.00 & 1.00 & 1.00 & 0.026 & 0.006 & 0.018 \\
grid & 389 & 0.10 & 20 & 0.998 & 0.01 & 0.64 & 0.71 & 0.237 & 0.000 & 0.078 \\
grid & 389 & 0.10 & 60 & 0.990 & 0.89 & 1.00 & 1.00 & 0.125 & 0.002 & 0.005 \\
grid & 385 & 0.30 & 20 & 0.994 & 0.01 & 0.67 & 0.75 & 0.164 & 0.000 & 0.068 \\
grid & 385 & 0.30 & 60 & 0.988 & 0.89 & 1.00 & 1.00 & 0.085 & 0.003 & 0.006 \\
\bottomrule
\end{tabular}

%% file: tables/baseline.tex
\begin{tabular}{l r r cc cc cc c}
\toprule
 & & & \multicolumn{2}{c}{FWER, $\Delta=0$} & \multicolumn{2}{c}{power, $\Delta=0.10$} & \multicolumn{2}{c}{power, $\Delta=0.20$} & \\
map & $\ell$ & $m$ & boot & perm & boot & perm & boot & perm & loc.\ err.\ \\
\midrule
image pixels & 63 & 20 & 0.044 & 0.052 & 0.09 & 0.10 & 0.33 & 0.34 & 0.118 \\
image pixels & 63 & 60 & 0.042 & 0.042 & 0.35 & 0.35 & 0.90 & 0.90 & 0.084 \\
image pixels & 111 & 20 & 0.066 & 0.084 & 0.15 & 0.19 & 0.52 & 0.57 & 0.065 \\
image pixels & 111 & 60 & 0.046 & 0.050 & 0.67 & 0.70 & 0.99 & 0.99 & 0.039 \\
image pixels & 189 & 20 & 0.018 & 0.052 & 0.23 & 0.30 & 0.66 & 0.74 & 0.052 \\
image pixels & 189 & 60 & 0.054 & 0.066 & 0.97 & 0.98 & 1.00 & 1.00 & 0.033 \\
image pixels & 325 & 20 & 0.024 & 0.076 & 0.63 & 0.76 & 0.89 & 0.94 & 0.022 \\
image pixels & 325 & 60 & 0.034 & 0.054 & 1.00 & 1.00 & 1.00 & 1.00 & 0.019 \\
landmarks (FPS) & 100 & 20 & 0.016 & 0.064 & 0.70 & 0.82 & 0.81 & 0.88 & 0.036 \\
landmarks (FPS) & 100 & 60 & 0.048 & 0.076 & 1.00 & 1.00 & 1.00 & 1.00 & 0.034 \\
landmarks (FPS) & 200 & 20 & 0.012 & 0.050 & 0.94 & 0.97 & 0.94 & 0.97 & 0.020 \\
landmarks (FPS) & 200 & 60 & 0.032 & 0.040 & 1.00 & 1.00 & 1.00 & 1.00 & 0.016 \\
\bottomrule
\end{tabular}

%% file: tables/comparative.tex
\begin{tabular}{lcccccccc}
\toprule
$\Delta$ & 0 & 0.025 & 0.05 & 0.075 & 0.1 & 0.15 & 0.2 & 0.3 \\
\midrule
\multicolumn{9}{l}{\emph{persistence regime}} \\
band (max-t, bootstrap) & 0.02 & 0.05 & 0.40 & 0.66 & 0.77 & 0.82 & 0.87 & 0.81 \\
embedding MMD (permutation) & 0.05 & 0.09 & 0.43 & 0.69 & 0.76 & 0.83 & 0.88 & 0.87 \\
STRAND (diagram perm.) & 0.05 & 0.05 & 0.05 & 0.06 & 0.08 & 0.06 & 0.08 & 0.13 \\
STRAND (pooled asympt.) & 0.05 & 0.05 & 0.05 & 0.06 & 0.07 & 0.07 & 0.09 & 0.15 \\
landscape MMD & 0.07 & 0.06 & 0.03 & 0.04 & 0.07 & 0.03 & 0.07 & 0.07 \\
persistence-image MMD & 0.06 & 0.06 & 0.05 & 0.07 & 0.09 & 0.10 & 0.15 & 0.33 \\
\midrule
\multicolumn{9}{l}{\emph{location regime}} \\
band (max-t, bootstrap) & 0.03 & 0.23 & 0.63 & 0.70 & 0.71 & 0.85 & 0.79 & 0.86 \\
embedding MMD (permutation) & 0.02 & 0.28 & 0.67 & 0.71 & 0.76 & 0.89 & 0.82 & 0.91 \\
STRAND (diagram perm.) & 0.04 & 0.05 & 0.05 & 0.05 & 0.03 & 0.05 & 0.03 & 0.03 \\
STRAND (pooled asympt.) & 0.04 & 0.05 & 0.05 & 0.04 & 0.04 & 0.05 & 0.04 & 0.04 \\
landscape MMD & 0.05 & 0.05 & 0.05 & 0.04 & 0.05 & 0.06 & 0.05 & 0.05 \\
persistence-image MMD & 0.03 & 0.05 & 0.06 & 0.05 & 0.03 & 0.06 & 0.03 & 0.04 \\
\bottomrule
\end{tabular}

%% file: tables/weaknull_var.tex
\begin{tabular}{l l r r ccc cc}
\toprule
 & & & & \multicolumn{3}{c}{pilot scale} & \multicolumn{2}{c}{sample scale} \\
\cmidrule(lr){5-7}\cmidrule(lr){8-9}
design & placement & $\ell$ & $m$ / $m'$ & band & Bonf. & perm & band & perm \\
\midrule
clustered & FPS & 100 & 20 / 20 & 0.032 & 0.024 & 0.042 & 0.252 & 0.201 \\
clustered & FPS & 100 & 20 / 60 & 0.041 & 0.028 & 0.000 & 0.090 & 0.024 \\
clustered & FPS & 100 & 60 / 20 & 0.048 & 0.035 & 0.178 & 0.343 & 0.184 \\
clustered & FPS & 400 & 20 / 20 & 0.039 & 0.014 & 0.048 & 0.283 & 0.219 \\
clustered & FPS & 400 & 20 / 60 & 0.057 & 0.010 & 0.001 & 0.121 & 0.024 \\
clustered & FPS & 400 & 60 / 20 & 0.026 & 0.021 & 0.135 & 0.352 & 0.187 \\
\addlinespace
clustered & grid & 59 & 20 / 20 & 0.027 & 0.107 & 0.041 & 0.245 & 0.239 \\
clustered & grid & 59 & 20 / 60 & 0.024 & 0.099 & 0.007 & 0.121 & 0.010 \\
clustered & grid & 59 & 60 / 20 & 0.021 & 0.132 & 0.125 & 0.649 & 0.333 \\
clustered & grid & 183 & 20 / 20 & 0.012 & 0.244 & 0.037 & 0.049 & 0.115 \\
clustered & grid & 183 & 20 / 60 & 0.011 & 0.276 & 0.015 & 0.163 & 0.016 \\
clustered & grid & 183 & 60 / 20 & 0.006 & 0.367 & 0.092 & 0.826 & 0.233 \\
\addlinespace
counts & FPS & 100 & 60 / 20 & 0.026 & 0.102 & 0.092 & 0.446 & 0.125 \\
counts & FPS & 392 & 60 / 20 & 0.005 & 0.135 & 0.065 & 0.548 & 0.135 \\
\addlinespace
counts & grid & 105 & 60 / 20 & 0.014 & 0.136 & 0.065 & 0.547 & 0.122 \\
counts & grid & 384 & 60 / 20 & 0.005 & 0.533 & 0.038 & 0.033 & 0.058 \\
\bottomrule
\end{tabular}

%% file: tables/level_full.tex
\begin{tabular}{llrr ccc c}
\toprule
placement & $\ell$ & $m$ & & cover & FWER$_{\mathrm{boot}}$ & FWER$_{\mathrm{Bonf}}$ & $\hat c_\alpha$ / $z_{\mathrm{Bonf}}$ \\
\midrule
FPS & 50 & 20 & & 0.978 & 0.022 & 0.070 & 3.64 / 3.29 \\
FPS & 50 & 40 & & 0.958 & 0.042 & 0.088 & 3.52 / 3.29 \\
FPS & 50 & 60 & & 0.970 & 0.030 & 0.054 & 3.47 / 3.29 \\
FPS & 100 & 20 & & 0.984 & 0.016 & 0.086 & 3.84 / 3.48 \\
FPS & 100 & 40 & & 0.962 & 0.038 & 0.062 & 3.69 / 3.48 \\
FPS & 100 & 60 & & 0.952 & 0.048 & 0.082 & 3.62 / 3.48 \\
FPS & 200 & 20 & & 0.988 & 0.012 & 0.086 & 4.15 / 3.66 \\
FPS & 200 & 40 & & 0.966 & 0.034 & 0.064 & 3.94 / 3.66 \\
FPS & 200 & 60 & & 0.968 & 0.032 & 0.056 & 3.86 / 3.66 \\
FPS & 395 & 20 & & 0.992 & 0.008 & 0.082 & 4.50 / 3.83 \\
FPS & 395 & 40 & & 0.982 & 0.018 & 0.088 & 4.23 / 3.83 \\
FPS & 395 & 60 & & 0.966 & 0.034 & 0.066 & 4.12 / 3.83 \\
grid & 55 & 20 & & 0.966 & 0.034 & 0.088 & 3.67 / 3.32 \\
grid & 55 & 40 & & 0.956 & 0.044 & 0.076 & 3.55 / 3.32 \\
grid & 55 & 60 & & 0.946 & 0.054 & 0.084 & 3.50 / 3.32 \\
grid & 105 & 20 & & 0.986 & 0.014 & 0.106 & 4.06 / 3.49 \\
grid & 105 & 40 & & 0.960 & 0.040 & 0.080 & 3.85 / 3.49 \\
grid & 105 & 60 & & 0.968 & 0.032 & 0.074 & 3.77 / 3.49 \\
grid & 208 & 20 & & 0.982 & 0.018 & 0.252 & 4.98 / 3.67 \\
grid & 208 & 40 & & 0.990 & 0.010 & 0.180 & 4.58 / 3.67 \\
grid & 208 & 60 & & 0.978 & 0.022 & 0.198 & 4.41 / 3.67 \\
grid & 385 & 20 & & 0.990 & 0.010 & 0.482 & 5.99 / 3.83 \\
grid & 385 & 40 & & 0.988 & 0.012 & 0.374 & 5.42 / 3.83 \\
grid & 385 & 60 & & 0.994 & 0.006 & 0.332 & 5.16 / 3.83 \\
\bottomrule
\end{tabular}

%% file: tables/localization_full_grid.tex
\begin{tabular}{llrr c ccc cc c}
\toprule
placement & $\ell$ & $\Delta$ & $m$ & cover & power$_{\mathrm{boot}}$ & power$_{\mathrm{Bonf}}$ & power$_{\mathrm{BH}}$ & FDR$_{\mathrm{BH}}$ & FDR$_{\mathrm{boot}}$ & loc.\ err.\ \\
\midrule
grid & 55 & 0.10 & 20 & 0.958 & 0.99 & 1.00 & 1.00 & 0.033 & 0.006 & 0.016 \\
grid & 55 & 0.10 & 40 & 0.960 & 1.00 & 1.00 & 1.00 & 0.029 & 0.007 & 0.016 \\
grid & 55 & 0.10 & 60 & 0.966 & 1.00 & 1.00 & 1.00 & 0.027 & 0.007 & 0.016 \\
grid & 55 & 0.15 & 20 & 0.972 & 0.97 & 0.99 & 0.99 & 0.028 & 0.003 & 0.019 \\
grid & 55 & 0.15 & 40 & 0.950 & 1.00 & 1.00 & 1.00 & 0.034 & 0.009 & 0.017 \\
grid & 55 & 0.15 & 60 & 0.952 & 1.00 & 1.00 & 1.00 & 0.024 & 0.004 & 0.016 \\
grid & 55 & 0.20 & 20 & 0.972 & 1.00 & 1.00 & 1.00 & 0.017 & 0.003 & 0.016 \\
grid & 55 & 0.20 & 40 & 0.968 & 1.00 & 1.00 & 1.00 & 0.023 & 0.004 & 0.016 \\
grid & 55 & 0.20 & 60 & 0.970 & 1.00 & 1.00 & 1.00 & 0.014 & 0.002 & 0.016 \\
grid & 55 & 0.30 & 20 & 0.970 & 0.98 & 1.00 & 1.00 & 0.025 & 0.004 & 0.019 \\
grid & 55 & 0.30 & 40 & 0.952 & 1.00 & 1.00 & 1.00 & 0.027 & 0.003 & 0.018 \\
grid & 55 & 0.30 & 60 & 0.954 & 1.00 & 1.00 & 1.00 & 0.026 & 0.006 & 0.018 \\
grid & 105 & 0.10 & 20 & 0.964 & 0.92 & 0.99 & 0.99 & 0.054 & 0.007 & 0.005 \\
grid & 105 & 0.10 & 40 & 0.974 & 1.00 & 1.00 & 1.00 & 0.038 & 0.003 & 0.003 \\
grid & 105 & 0.10 & 60 & 0.970 & 1.00 & 1.00 & 1.00 & 0.035 & 0.003 & 0.003 \\
grid & 105 & 0.15 & 20 & 0.988 & 0.99 & 1.00 & 1.00 & 0.028 & 0.003 & 0.006 \\
grid & 105 & 0.15 & 40 & 0.976 & 1.00 & 1.00 & 1.00 & 0.033 & 0.003 & 0.005 \\
grid & 105 & 0.15 & 60 & 0.968 & 1.00 & 1.00 & 1.00 & 0.032 & 0.004 & 0.004 \\
grid & 105 & 0.20 & 20 & 0.990 & 1.00 & 1.00 & 1.00 & 0.034 & 0.001 & 0.003 \\
grid & 105 & 0.20 & 40 & 0.976 & 1.00 & 1.00 & 1.00 & 0.025 & 0.001 & 0.003 \\
grid & 105 & 0.20 & 60 & 0.974 & 1.00 & 1.00 & 1.00 & 0.023 & 0.001 & 0.003 \\
grid & 105 & 0.30 & 20 & 0.976 & 0.97 & 1.00 & 1.00 & 0.032 & 0.004 & 0.006 \\
grid & 105 & 0.30 & 40 & 0.974 & 1.00 & 1.00 & 1.00 & 0.035 & 0.003 & 0.003 \\
grid & 105 & 0.30 & 60 & 0.978 & 1.00 & 1.00 & 1.00 & 0.026 & 0.002 & 0.003 \\
grid & 210 & 0.10 & 20 & 0.992 & 0.22 & 0.87 & 0.88 & 0.093 & 0.004 & 0.034 \\
grid & 210 & 0.10 & 40 & 0.986 & 0.93 & 1.00 & 1.00 & 0.068 & 0.002 & 0.017 \\
grid & 210 & 0.10 & 60 & 0.980 & 1.00 & 1.00 & 1.00 & 0.063 & 0.001 & 0.016 \\
grid & 210 & 0.15 & 20 & 0.992 & 0.15 & 0.79 & 0.84 & 0.089 & 0.002 & 0.037 \\
grid & 210 & 0.15 & 40 & 0.984 & 0.92 & 1.00 & 1.00 & 0.070 & 0.003 & 0.019 \\
grid & 210 & 0.15 & 60 & 0.984 & 1.00 & 1.00 & 1.00 & 0.055 & 0.003 & 0.018 \\
grid & 210 & 0.20 & 20 & 0.992 & 0.29 & 0.86 & 0.89 & 0.083 & 0.006 & 0.036 \\
grid & 210 & 0.20 & 40 & 0.988 & 0.94 & 1.00 & 1.00 & 0.058 & 0.002 & 0.017 \\
grid & 210 & 0.20 & 60 & 0.982 & 1.00 & 1.00 & 1.00 & 0.050 & 0.002 & 0.017 \\
grid & 210 & 0.30 & 20 & 0.996 & 0.23 & 0.82 & 0.84 & 0.074 & 0.000 & 0.037 \\
grid & 210 & 0.30 & 40 & 0.978 & 0.96 & 1.00 & 1.00 & 0.058 & 0.003 & 0.020 \\
grid & 210 & 0.30 & 60 & 0.978 & 1.00 & 1.00 & 1.00 & 0.049 & 0.003 & 0.019 \\
grid & 389 & 0.10 & 20 & 0.998 & 0.01 & 0.64 & 0.71 & 0.237 & 0.000 & 0.078 \\
grid & 389 & 0.10 & 40 & 0.984 & 0.39 & 0.98 & 0.99 & 0.153 & 0.008 & 0.015 \\
grid & 389 & 0.10 & 60 & 0.990 & 0.89 & 1.00 & 1.00 & 0.125 & 0.002 & 0.005 \\
grid & 386 & 0.15 & 20 & 0.988 & 0.02 & 0.64 & 0.74 & 0.215 & 0.001 & 0.080 \\
grid & 386 & 0.15 & 40 & 0.982 & 0.25 & 0.98 & 0.99 & 0.130 & 0.006 & 0.023 \\
grid & 386 & 0.15 & 60 & 0.980 & 0.78 & 1.00 & 1.00 & 0.094 & 0.000 & 0.012 \\
grid & 384 & 0.20 & 20 & 0.996 & 0.03 & 0.70 & 0.75 & 0.172 & 0.003 & 0.055 \\
grid & 384 & 0.20 & 40 & 0.992 & 0.53 & 0.99 & 0.99 & 0.099 & 0.002 & 0.009 \\
grid & 384 & 0.20 & 60 & 0.978 & 0.96 & 1.00 & 1.00 & 0.083 & 0.004 & 0.004 \\
grid & 385 & 0.30 & 20 & 0.994 & 0.01 & 0.67 & 0.75 & 0.164 & 0.000 & 0.068 \\
grid & 385 & 0.30 & 40 & 0.986 & 0.36 & 0.97 & 0.99 & 0.111 & 0.005 & 0.017 \\
grid & 385 & 0.30 & 60 & 0.988 & 0.89 & 1.00 & 1.00 & 0.085 & 0.003 & 0.006 \\
\bottomrule
\end{tabular}

%% file: tables/localization_full_fps.tex
\begin{tabular}{llrr c ccc cc c}
\toprule
placement & $\ell$ & $\Delta$ & $m$ & cover & power$_{\mathrm{boot}}$ & power$_{\mathrm{Bonf}}$ & power$_{\mathrm{BH}}$ & FDR$_{\mathrm{BH}}$ & FDR$_{\mathrm{boot}}$ & loc.\ err.\ \\
\midrule
FPS & 50 & 0.10 & 20 & 0.970 & 0.46 & 0.54 & 0.55 & 0.025 & 0.006 & 0.073 \\
FPS & 50 & 0.10 & 40 & 0.966 & 0.89 & 0.91 & 0.92 & 0.017 & 0.005 & 0.045 \\
FPS & 50 & 0.10 & 60 & 0.976 & 0.99 & 0.99 & 1.00 & 0.020 & 0.005 & 0.043 \\
FPS & 49 & 0.15 & 20 & 0.966 & 0.95 & 0.99 & 0.99 & 0.042 & 0.009 & 0.033 \\
FPS & 49 & 0.15 & 40 & 0.966 & 1.00 & 1.00 & 1.00 & 0.033 & 0.005 & 0.032 \\
FPS & 49 & 0.15 & 60 & 0.964 & 1.00 & 1.00 & 1.00 & 0.030 & 0.004 & 0.031 \\
FPS & 50 & 0.20 & 20 & 0.968 & 0.97 & 0.99 & 0.99 & 0.028 & 0.005 & 0.015 \\
FPS & 50 & 0.20 & 40 & 0.966 & 1.00 & 1.00 & 1.00 & 0.029 & 0.003 & 0.013 \\
FPS & 50 & 0.20 & 60 & 0.960 & 1.00 & 1.00 & 1.00 & 0.027 & 0.004 & 0.012 \\
FPS & 50 & 0.30 & 20 & 0.974 & 0.08 & 0.15 & 0.15 & 0.021 & 0.013 & 0.108 \\
FPS & 50 & 0.30 & 40 & 0.982 & 0.28 & 0.35 & 0.36 & 0.013 & 0.002 & 0.087 \\
FPS & 50 & 0.30 & 60 & 0.966 & 0.49 & 0.56 & 0.58 & 0.027 & 0.013 & 0.075 \\
FPS & 100 & 0.10 & 20 & 0.974 & 0.70 & 0.87 & 0.89 & 0.028 & 0.011 & 0.034 \\
FPS & 100 & 0.10 & 40 & 0.968 & 0.99 & 1.00 & 1.00 & 0.031 & 0.003 & 0.028 \\
FPS & 100 & 0.10 & 60 & 0.962 & 1.00 & 1.00 & 1.00 & 0.035 & 0.002 & 0.028 \\
FPS & 100 & 0.15 & 20 & 0.974 & 0.97 & 0.99 & 0.99 & 0.029 & 0.004 & 0.011 \\
FPS & 100 & 0.15 & 40 & 0.972 & 1.00 & 1.00 & 1.00 & 0.022 & 0.002 & 0.010 \\
FPS & 100 & 0.15 & 60 & 0.982 & 1.00 & 1.00 & 1.00 & 0.023 & 0.002 & 0.009 \\
FPS & 100 & 0.20 & 20 & 0.982 & 0.81 & 0.94 & 0.95 & 0.022 & 0.002 & 0.036 \\
FPS & 100 & 0.20 & 40 & 0.956 & 1.00 & 1.00 & 1.00 & 0.023 & 0.004 & 0.034 \\
FPS & 100 & 0.20 & 60 & 0.962 & 1.00 & 1.00 & 1.00 & 0.022 & 0.003 & 0.034 \\
FPS & 99 & 0.30 & 20 & 0.966 & 0.75 & 0.88 & 0.89 & 0.027 & 0.007 & 0.028 \\
FPS & 99 & 0.30 & 40 & 0.958 & 1.00 & 1.00 & 1.00 & 0.021 & 0.005 & 0.021 \\
FPS & 99 & 0.30 & 60 & 0.948 & 1.00 & 1.00 & 1.00 & 0.017 & 0.004 & 0.021 \\
FPS & 200 & 0.10 & 20 & 0.986 & 0.94 & 1.00 & 1.00 & 0.037 & 0.004 & 0.016 \\
FPS & 200 & 0.10 & 40 & 0.978 & 1.00 & 1.00 & 1.00 & 0.025 & 0.001 & 0.015 \\
FPS & 200 & 0.10 & 60 & 0.978 & 1.00 & 1.00 & 1.00 & 0.029 & 0.001 & 0.014 \\
FPS & 200 & 0.15 & 20 & 0.972 & 0.94 & 0.99 & 1.00 & 0.020 & 0.001 & 0.022 \\
FPS & 200 & 0.15 & 40 & 0.978 & 1.00 & 1.00 & 1.00 & 0.024 & 0.000 & 0.021 \\
FPS & 200 & 0.15 & 60 & 0.980 & 1.00 & 1.00 & 1.00 & 0.018 & 0.001 & 0.020 \\
FPS & 200 & 0.20 & 20 & 0.970 & 0.94 & 0.99 & 0.99 & 0.031 & 0.001 & 0.020 \\
FPS & 200 & 0.20 & 40 & 0.972 & 1.00 & 1.00 & 1.00 & 0.023 & 0.001 & 0.016 \\
FPS & 200 & 0.20 & 60 & 0.962 & 1.00 & 1.00 & 1.00 & 0.021 & 0.001 & 0.016 \\
FPS & 199 & 0.30 & 20 & 0.974 & 0.98 & 1.00 & 1.00 & 0.017 & 0.001 & 0.016 \\
FPS & 199 & 0.30 & 40 & 0.976 & 1.00 & 1.00 & 1.00 & 0.017 & 0.002 & 0.015 \\
FPS & 199 & 0.30 & 60 & 0.970 & 1.00 & 1.00 & 1.00 & 0.017 & 0.000 & 0.015 \\
FPS & 399 & 0.10 & 20 & 0.994 & 0.95 & 0.99 & 0.99 & 0.040 & 0.000 & 0.007 \\
FPS & 399 & 0.10 & 40 & 0.982 & 1.00 & 1.00 & 1.00 & 0.032 & 0.002 & 0.006 \\
FPS & 399 & 0.10 & 60 & 0.974 & 1.00 & 1.00 & 1.00 & 0.034 & 0.002 & 0.006 \\
FPS & 394 & 0.15 & 20 & 0.988 & 0.97 & 1.00 & 1.00 & 0.026 & 0.000 & 0.010 \\
FPS & 394 & 0.15 & 40 & 0.970 & 1.00 & 1.00 & 1.00 & 0.022 & 0.001 & 0.009 \\
FPS & 394 & 0.15 & 60 & 0.974 & 1.00 & 1.00 & 1.00 & 0.026 & 0.001 & 0.009 \\
FPS & 398 & 0.20 & 20 & 0.988 & 0.92 & 0.99 & 1.00 & 0.031 & 0.001 & 0.023 \\
FPS & 398 & 0.20 & 40 & 0.984 & 1.00 & 1.00 & 1.00 & 0.030 & 0.000 & 0.020 \\
FPS & 398 & 0.20 & 60 & 0.974 & 1.00 & 1.00 & 1.00 & 0.028 & 0.000 & 0.019 \\
FPS & 395 & 0.30 & 20 & 0.988 & 0.96 & 1.00 & 1.00 & 0.019 & 0.000 & 0.016 \\
FPS & 395 & 0.30 & 40 & 0.974 & 1.00 & 1.00 & 1.00 & 0.015 & 0.001 & 0.014 \\
FPS & 395 & 0.30 & 60 & 0.972 & 1.00 & 1.00 & 1.00 & 0.013 & 0.001 & 0.013 \\
\bottomrule
\end{tabular}